\documentclass[a4paper,11pt,reqno]{article}

\usepackage[hmargin=3cm,vmargin=2.8cm]{geometry}
\usepackage[utf8x]{inputenc}

\usepackage{amsmath, amsfonts, amssymb, amsthm}
\usepackage{graphicx, float, mathtools}

\usepackage[pdfencoding=auto, psdextra]{hyperref}
\usepackage{cleveref, braket, enumitem, color}
\usepackage[dvipsnames]{xcolor}
\usepackage{dsfont}
\usepackage{enumitem}
\usepackage{tikz}
\usetikzlibrary{patterns}

\DeclareSymbolFont{CMletters}{OML}{cmm}{m}{it}
\DeclareMathSymbol{\xi}{\mathord}{CMletters}{"18}

\renewcommand{\[}{\begin{equation} \begin{aligned}}
\renewcommand{\]}{\end{aligned} \end{equation}}

\theoremstyle{plain}

\newtheorem{theorem}{Theorem}[section]
\newtheorem*{theorem*}{Theorem}

\newtheorem{lemma}[theorem]{Lemma}
\newtheorem*{lemma*}{Lemma}

\newtheorem{proposition}[theorem]{Proposition}

\theoremstyle{definition}

\newcommand{\dx}{{\rm d}x}
\newcommand{\dt}{{\rm d}t}
\newcommand{\ds}{{\rm d}s}
\newcommand{\dy}{{\rm d}y}
\newcommand{\dz}{{\rm d}z}
\newcommand{\dr}{{\rm d}r}

\newcommand{\dG}{{\rm d}\Gamma}

\newcommand{\vphi}{{\varphi}}
\newcommand{\eps}{\varepsilon}
\newcommand{\R}{{\BR}}
\newcommand{\Z}{{\mathbb Z}_{\leq 1}}

\newcommand{\BR}{{\mathbb R}}

\newcommand{\cB}{{\cal{B}}}

\newcommand{\cF}{{\cal{F}}}
\newcommand{\cH}{{\cal{H}}}

\newcommand{\cN}{{\cal{N}}}

\newcommand{\tr}{{\rm{Tr\, }}}
\newcommand{\vep}{{{\varepsilon}}}
\newcommand{\supp}{{\rm{supp }}}
\newcommand{\ao}{\mathfrak{a}}
\newcommand{\dd}{{\rm d}}
\newcommand\1{{\ensuremath {\mathds 1} }}

\newcommand{\nn}{\nonumber}
\newcommand{\DETAILS}[1]{}
\newcommand{\hc}{{\rm{h.c.}}}

\allowdisplaybreaks

\makeatletter
\newcounter{parentsubequation}

\makeatother

\numberwithin{equation}{section}

\makeatletter
\renewcommand{\sideset}[3]{%
  \@mathmeasure\z@\displaystyle{#3}%
  \global\setbox\@ne\vbox to\ht\z@{}\dp\@ne\dp\z@
  \setbox\tw@\box\@ne
  \@mathmeasure4\displaystyle{\copy\tw@#1}%
  \@mathmeasure6\displaystyle{#3\nolimits#2}%
  \dimen@-\wd6 \advance\dimen@\wd4 \advance\dimen@\wd\z@
  \mathop{\hbox to\dimen@{}}\!
  \mathop{\kern-\dimen@\box4\box6}%
}
\makeatother

\makeatletter
\newtheorem*{rep@theorem}{\rep@title}
\newcommand{\newreptheorem}[2]{%
\newenvironment{rep#1}[1]{%
 \def\rep@title{#2 \ref{##1}}%
 \begin{rep@theorem}}%
 {\end{rep@theorem}}}
\makeatother

\newreptheorem{theorem}{Theorem}
\newreptheorem{lemma}{Lemma}

\makeatletter
\newcommand{\hfillcell}[1]{ \ifmeasuring@#1\else\omit \hfill $\displaystyle#1$  \ignorespaces\fi}
\makeatother

\title{The free energy of dilute Bose gases at low temperatures}

\author{Florian Haberberger\thanks{Department of Mathematics, LMU Munich, Theresienstrasse 39, 80333 Munich, Germany.} , Christian Hainzl$^*$,  Phan Th\`anh Nam$^*$, \\ Robert Seiringer\thanks{Institute of Science and Technology Austria (ISTA), Am Campus 1, 3400 Klosterneuburg, Austria. Emails: haberberger@math.lmu.de, hainzl@math.lmu.de, nam@math.lmu.de, robert.seiringer@ist.ac.at, triay@math.lmu.de}
, Arnaud Triay$^*$ 
}

\begin{document}

\maketitle

\begin{abstract}
We consider a low density Bose gas interacting through a repulsive potential in the thermodynamic limit.  We justify the Lee--Huang--Yang  conjecture of 1957 concerning the shape of the excitation spectrum. Rigorously we prove a lower bound for the free energy at suitably low temperatures, where the modified excitation spectrum leads to a second order correction to the ground state energy.
\end{abstract} 

\section{Introduction}

Although thermodynamic properties of the ideal Bose gas have been well understood since the pioneering work of Bose and Einstein \cite{Bose-24,Einstein-24}, the rigorous understanding of interacting Bose gases remains a major challenge. In particular for dilute systems,  Bose--Einstein condensation and related phenomena at low temperatures have been observed experimentally since 1995 \cite{AEMWC-95,DMADDKK-95}, but the derivation of these collective effects from first principles of quantum mechanics is mostly open. 

In 1957,  Lee, Huang and Yang \cite{LeeHuaYan-57} used a pseudopotential method to analyze the spectrum of dilute Bose gases. To be precise, for a Bose gas with density $\rho>0$ interacting through a repulsive potential with scattering length $\ao>0$, they predicted that the ground state energy per unit volume is given by
\begin{align}\label{eq:LHY-E0-intro}
E_0 = 4\pi \ao \rho^2 \left(1 + \frac{128}{15 \sqrt \pi} \sqrt{\rho \ao^3}\right),
\end{align}
and that the low-lying eigenvalues have the form
\begin{align}\label{eq:LHY-Ek-intro}
E_0 + \sum_{p\ne 0} m_p \sqrt{p^4+ 16 \pi \ao \rho p^2}, \quad m_p=0,1,2,...,
\end{align}
up to small errors in the dilute limit $\rho \ao^3 \to 0$ (see Eqs.\,(25) and~(34) in \cite{LeeHuaYan-57}, respectively). 
Although the work in  \cite{LeeHuaYan-57} focuses on the hard-sphere interaction, the Lee--Huang--Yang formulas are expected to hold true for a large class of repulsive interactions. Thus \eqref{eq:LHY-E0-intro} and \eqref{eq:LHY-Ek-intro} exhibit a universality of dilute Bose gases. Namely, the ground state energy and the excitation spectrum are well approximated solely in terms of the density of the system and the scattering length of the interaction. Their rigorous justification from the many-body Schr\"odinger equation
has been an important problem in mathematical physics.

For the ground state energy, the Lee--Huang--Yang formula \eqref{eq:LHY-E0-intro} has been established in a series of remarkable works over the last six decades. The upper bound to the leading order term $4\pi \ao \rho^2$ 
was achieved by Dyson already in 1957 \cite{Dyson-57}, but it took more than 40 years until the matching lower bound was proved by Lieb--Yngvason in 1998 \cite{LieYng-98}. The second order term $4\pi \ao \rho^2 \times \frac{128}{15 \sqrt \pi} \sqrt{\rho \ao^3}$ was proved by Yau--Yin in 2009 \cite{YauYin-09} for the upper bound, and finally established by Fournais--Solovej in 2020 \cite{FouSol-20} for the lower bound. For further developments, we refer to \cite{FouSol-21} for an extension of the second order lower bound to hard-sphere interactions, \cite{BasCenSch-21} for an alternative derivation of the second order upper bound, and \cite{ARS-22, FGJMO-22} for related results in 1D and 2D, respectively. 


The existing literature, however, does not provide information on the excitation spectrum. The goal of the present paper is to address this second aspect of the Lee--Huang--Yang conjecture. Instead of justifying \eqref{eq:LHY-Ek-intro} for each individual eigenvalue, which is virtually impossible as it would require a precision way beyond what the current technologies are capable of, we derive a collective version of \eqref{eq:LHY-Ek-intro} in terms of the free energy at low temperatures. To be precise, combining \eqref{eq:LHY-E0-intro} and \eqref{eq:LHY-Ek-intro} suggests that the free energy per unit volume at low temperatures  $T>0$ can be approximated by 
\begin{align}\label{eq:LHY-intro}
&E_0 + \frac{T}{(2\pi)^3} \int_{\mathbb{R}^{3}} \log\left(1-e^{-T^{-1}\sqrt{p^4 + 16 \pi \rho \ao p^2}}\right)\dd p  \nn\\
&=4\pi \ao \rho^2 \left(1 + \frac{128}{15 \sqrt \pi} \sqrt{\rho \ao^3}\right) + \frac{T^{5/2}}{(2\pi)^3} \int_{\mathbb{R}^{3}} \log\left(1-e^{-\sqrt{p^4 + \frac{16 \pi \rho \ao}{T} p^2}}\right)\dd p. 
\end{align}

It is important to remark that the higher the temperature, the more challenging it is to justify \eqref{eq:LHY-intro}. In fact, at temperatures around $T\sim \rho^{2/3} = \rho \ao (\rho\ao^3)^{-1/3}$, which  is of the order of the critical temperature for  Bose--Einstein condensation (BEC), the expansion \eqref{eq:LHY-intro} is no longer correct; see \cite{Seiringer-08,Yin-10} for detailed analysis at the leading order. This fact is not surprising since at this critical temperature regime, BEC only holds partially, and hence the Lee--Huang--Yang computation does not apply anymore. In the present work, we are interested in \eqref{eq:LHY-intro} at low temperatures $T \sim \rho \ao$ for which the entropy contribution is proportional to the second order term of the Lee--Huang--Yang  ground state energy. This is the natural parameter regime to resolve the Lee--Huang--Yang conjecture  \eqref{eq:LHY-E0-intro}-\eqref{eq:LHY-Ek-intro} within a single formula. The precise statement of our result and an outline of the proof will be provided in the next section.

\subsection{Main result}

Let $V\in L^1(\mathbb{R}^{3})$ be non-negative, compactly supported, radially symmetric decreasing and $\ao>0$ its scattering length (see Section \ref{sec:scattering} for the definition).  Let $\Delta$ denote the Laplacian with Neumann boundary conditions on $\Lambda_{L}=[-L/2,L/2]^3$. For  integers $N\geq 2$, consider the Hamiltonian
\begin{equation} \label{eq:HN}
H_{N} = \sum_{i=1}^N -\Delta_{x_i} + \sum_{1\leq i < j \leq N} V(x_i-x_j)
\end{equation}
acting on the bosonic space $L^2_s (\Lambda_{L}^N) := \bigotimes_{\rm sym}^N L^2(\Lambda_{L})$. This operator can be defined as a self-adjoint operator by Friedrichs' method, and it has  compact resolvent. 

The free energy of the system at temperature $T>0$ is defined by
$$
F_L(N)  = \inf_{\Gamma} \left( \tr(H_N \Gamma) - TS(\Gamma)\right),
$$
where the infimum is taken over all mixed states $\Gamma$, that is all bounded operators $\Gamma \geq 0$ with $\tr \Gamma = 1$, and where $S(\Gamma)=-\tr(\Gamma \log \Gamma)$ denotes the entropy of $\Gamma$. By the Gibbs variational principle the infimum is attained by the Gibbs state $\Gamma_N=Z_N^{-1} e^{-H_N/T}$ and the free energy can be computed from the partition function $Z_N$ as
\begin{equation*}
F_L(N) = - T \log Z_N = -T \log \tr e^{\frac{-H_{N}}{T}}.
\end{equation*}

We are interested in the free energy per unit volume in the thermodynamic limit
\begin{equation}
f(\rho,T) := \lim_{\substack{N \to \infty \\ NL^{-3} \to \rho}} \frac{F_L(N)}{L^3} \label{eq:def_f_rho_T}.
\end{equation}
It is well-known that the free energy density $f(\rho,T)$ is well-defined and actually independent of the boundary conditions we imposed on  $\Lambda_{L}$.  Our main result is the following justification of \eqref{eq:LHY-intro} as a lower bound.

\begin{theorem}
	\label{thm:main} Let $\nu=1/5000$. In the dilute limit $\rho \ao^3 \to 0$, for any $0\le  T \le \rho \ao (\rho\ao^3)^{-\nu}$, the free energy density in \eqref{eq:def_f_rho_T} satisfies  
\begin{align}\nonumber
f(\rho,T) & \geq 4\pi \ao \rho^2 \left(1 + \frac{128}{15 \sqrt \pi} \sqrt{\rho \ao^3}\right) + \frac{T^{5/2}}{(2\pi)^3} \int_{\mathbb{R}^{3}} \log\left(1-e^{-\sqrt{p^4 + \tfrac{16 \pi \rho \ao}{T} p^2}}\right)\dd p \\ & \quad - C (\rho \ao)^{5/2} (\rho \ao^3)^{\nu}. \label{eq:main-eq-thm}
\end{align}
Here the constant $C > 0$ depends only on $V$.
\end{theorem}

Here are some remarks on our result. 

\begin{itemize}

\item[{\bf 1.}] The free energy formula \eqref{eq:main-eq-thm} holds for $T\ge 0$, thus not only recovering the result on the ground state energy as established in \cite{FouSol-20} but also resolving the question on the excitation spectrum as predicted in \cite{LeeHuaYan-57}. Our condition $T \le \rho \ao (\rho\ao^3)^{-\nu}$ allows the case $T\sim \rho \ao$, which is particularly interesting, since in this case the temperature correction is of the same order as the second order  Lee--Huang--Yang correction to the ground state energy.
An upper bound condition on $T$ is not merely technical, but it is conceptually necessary. The formula  \eqref{eq:main-eq-thm} fails in the higher temperature regime $T\sim \rho^{2/3} = \rho \ao (\rho\ao^3)^{-1/3}$, and in this case deriving the correction to the leading order term in \cite{Seiringer-08,Yin-10} remains a very interesting open problem. 


%


\item[{\bf 2.}] Our assumptions on the potential $V$ can be relaxed in many ways. For example, if $V$ is not decreasing but it is radial and satisfies $V(x)\le C V(y)$ for $|x|\ge |y|$, then our proof applies equally well. In our analysis we fix $V$ (and in particular the scattering length $\ao$) and consider the low-density and low-temperature limits $\rho\to 0$ and $T\to 0$.  However, by simple scaling the relevant small parameters are the dimensionless quantities $\ao^3\rho$ and $\ao^2 T$. Our error terms will be bounded only in terms of the range $R$ of $V$ and its integral; more precisely, the constant $C$ in Theorem~\ref{thm:main} depends only on the dimensionless quantities $R/\ao$ and  $\|V\|_{L^1(\mathbb{R}^{3})}/\ao$. 

\item[{\bf 3.}] We expect that a matching upper bound for  \eqref{eq:main-eq-thm} also holds, and that the result can be extended to hard-sphere interactions. Proving such results requires new techniques, which hopefully will be addressed in the near future. 


\end{itemize}

Our proof strategy of Theorem \ref{thm:main} is different from the approach to the ground state energy problem in \cite{FouSol-20,FouSol-21}. In an effort to obtain information on the excitation spectrum, we introduce a new method, that revolves around a detailed analysis of local systems with Neumann boundary conditions. We use unitary transformations in the spirit of Bogoliubov's diagonalization idea  \cite{Bogoliubov-47} together with subtle renormalization techniques. While incorporating insights from recent developments \cite{BocBreCenSch-19,NamTri-21,BreSchSch-21,HaiSchTri-22} on the excitation spectrum in the fixed volume setting, our analysis in the thermodynamic limit introduces several novel ingredients which serve to not only simplify but also extend existing approaches on a conceptual level.

To quickly explain the novelty of the methodology, let us mention that the Lee-- Huang--Yang prediction \cite{LeeHuaYan-57} was based on the heuristic assumption of Bose--Einstein condensation (BEC), namely a macroscopic fraction of particles occupy the zero-momentum mode. Proving BEC in the thermodynamic limit is a major open problem in mathematical physics, but to compute the energy it is possible to consider localized systems in small boxes where BEC is easier to prove. Therefore, the localization method is of central importance. This idea was already used in 1998 by Lieb and Yngvason \cite{LieYng-98} in their proof of the leading order of the ground state energy, where they divided the thermodynamic box into smaller cells with Neumann boundary conditions on each cell, which is the appropriate method for a lower bound.

While this approach sounds plausible, the handling of the Neumann boundary conditions for the second order term of the ground state energy poses a major challenge compared to the typical periodic setting in unit volume \cite{BocBreCenSch-19,BreSchSch-21}.  In \cite{FouSol-20,FouSol-21}, Fournais and Solovej introduced a very subtle argument to localize the kinetic energy operator, 
 which allows them to keep the calculation essentially in the periodic setting. 
The intricate analysis in \cite{FouSol-20,FouSol-21}, however, does not seem to give access to the excitation spectrum. The main new contribution of the present work is to perform a rigorous analysis of the excitation spectrum on Neumann boxes, thus resolving the  Lee--Huang--Yang prediction in a very natural way. Further details of our proof will be given below. 



\subsection{Outline of the proof} 




{\bf General ideas.} Our proof strategy is inspired by  Bogoliubov's 1947 approach \cite{Bogoliubov-47} where he proposed an effective method to transform the Hamiltonian of an interacting Bose gas to  a non-interacting one, thus enabling an approximation for not only the ground state energy but also the excitation spectrum (the latter is particularly interesting due to its connection to superfluidity). As mentioned already in \cite{Bogoliubov-47}, this method is reasonably good in a mean-field situation where the particles are more or less independent, but it is insufficient for dilute gases where the particles are highly correlated. In fact, a formal application of the Bogoliubov approximation produces an incomplete form of \eqref{eq:LHY-intro} where the first two terms in a Born approximation of the scattering length $\ao$ appear instead of $\ao$ itself \cite{LieSeiSolYng-05}. Thus the main conceptual difficulty in our proof is to put the Bogoliubov approximation on a rigorous footing, including the subtle correction due to the correlation between particles. 

Heuristically, an important input for the Bogoliubov approximation is  BEC. Although proving  BEC in the thermodynamic limit is a major open problem, we are able to prove  BEC in localized systems in small boxes, which is sufficient to estimate the free energy. This idea has been carried out in the ground state problem  \cite{LieYng-98,FouSol-20}. To be precise, we decompose $\Lambda_{L}$ into smaller cubes $\Lambda_{\ell}$ of side length
\begin{equation}
\label{eq:def_ell}
\ell = \frac{\ao}{(\rho \ao^3)^{1/2+\kappa}}
\end{equation} for some small parameter $\kappa >0$ that will be chosen later. This length scale is chosen larger than the Gross--Pitaevskii length scale (also called healing length)
\begin{align*}
\ell_{\rm GP} = \frac{1}{\sqrt{\rho \ao}}.
\end{align*}
At the Gross--Pitaevskii length scale, the gap of the kinetic energy operator is of the same order as the interaction energy of one particle, which makes the proof of  BEC easier. On the other hand, at the Gross--Pitaevskii length scale, the contribution from  boundary conditions affects the second order term of the energy \cite{BocBreCenSch-19}. Therefore, by focusing on the length scale $\ell$  slightly larger than the Gross--Pitaevskii length scale, we still have a reasonably good control on the number of excitations, and at the same time we control boundary effects caused by the localization procedure. 

The Gross--Pitaevskii regime has been studied extensively in the literature, often in the equivalent formulation of having $n$ particles in the unit box with the interaction potential of the form $n^2 V(n(x-y))$. In this setting, the boundary of the domain matters. For periodic boundary conditions,  BEC was first derived in \cite{LieSei-02}, and the  excitation spectrum was first computed in  \cite{BocBreCenSch-19}. The key idea of \cite{BocBreCenSch-19} is that the Bogoliubov approximation can be justified rigorously by using suitable unitary transformations. Later, the excitation spectrum of inhomogeneous trapped Bose gases in $\R^3$ was derived independently in \cite{NamTri-21} and \cite{BreSchSch-21}. For us \cite{NamTri-21} is particularly relevant, as it contains several modifications of the strategy in \cite{BocBreCenSch-19}. This already led to a simplified proof in the periodic setting in \cite{HaiSchTri-22} and will further be helpful for the analysis of the present paper. 

For our purpose,  we have to deal with the Gross--Pitaevskii regime with Neumann boundary conditions. In this case, BEC with an almost optimal bound was derived recently in \cite{BocSei-22}, based on a suitable extension of the strategy in \cite{BocBreCenSch-19}, but it turns out that the Neumann boundary conditions cause a serious problem in the computation of the ground state energy and the excitation spectrum. In the study of the ground state problem in \cite{FouSol-20}, a completely different localization technique has been used, which allows to avoid the Neumann boundary issue but requires a subtle modification of the kinetic energy operator. 

Thus, while Neumann boundary conditions appear very naturally when seeking a lower bound, their are de facto incompatible with the translational invariant form of the interaction potential $V(x-y)$, making the justification of the Bogoliubov approximation in this case intricate. Solving that problem is the main new contribution of the present work. Roughly speaking, we will handle the Neumann boundary conditions by introducing a mirror symmetrization technique to relevant transformation kernels, thereby enabling the necessary extension of the strategy in \cite{NamTri-21,HaiSchTri-22}. Moreover, while these works consider the Gross--Pitaesvskii regime, we need push the analysis to much larger length scales, where the interaction potential dominates the kinetic energy and the LHY term is visible compared to boundary effects. This makes the diagonalization of the Hamiltonian harder but is necessary to recover the correct free energy in the thermodynamic limit when summing up the local free energies in all small boxes.

\bigskip
\noindent
{\bf Detailed setting.} We shall now explain the proof strategy in detail. It is convenient to consider, for $n \geq 0$, the rescaled Hamiltonian
\begin{equation}
	\label{eq:def_Hnl}
H_{n,\ell}= \sum_{i=1}^n -\Delta_{x_i} + \sum_{1\leq i < j \leq n} \ell^2 V(\ell(x_i-x_j))
\end{equation}
acting on $L^2_s(\Lambda^{n})$, the space of square integrable functions that are invariant under permutation of their variables,  where $\Delta$ is the Neumann Laplacian on the unit box $\Lambda = [-1/2,1/2]^3$. 

The Hamiltonians $H_n$, defined as in \eqref{eq:HN} with $(N,\Lambda_L)$ replaced by $(n,\Lambda_\ell)$, and $H_{n,\ell}$ defined in \eqref{eq:def_Hnl} are related via
\begin{equation*}
H_n = \frac{1}{\ell^2} \mathcal T_\ell ^* H_{n,\ell} \mathcal T_\ell 
\end{equation*} 
with the unitary scaling transformation $\mathcal T_\ell \Psi(\,\cdot\,) = \ell^{3n/2} \Psi(\ell\, \cdot\,)$. Hence, we are interested in the free energy  
\begin{equation} \label{eq:def_energy_small_box}
F_\ell(n)= - T \log \tr(e^{-\frac{H_{n,\ell}}{T\ell^2}})\,.
\end{equation}
The main part of our work is devoted to the proof of the following theorem. 

\begin{theorem}[Free energy on small boxes]
	\label{theo:free_energy_small_box} Let $\ell$ be given in \eqref{eq:def_ell} with $\kappa=5\nu=1/1000$. Let $0\le T \leq  (\rho \ao)(\rho\ao^3)^{-\nu}$ and $0 \leq n\leq C \rho \ell^3$ for some $C>0$. Then, for $\rho \ao^3$ small enough,
	\begin{equation}
	\label{eq:theo:free_energy_small_box}
F_\ell(n) \geq f_{\rm Bog}(n,\ell) + \mathcal O(\ell^3 (\rho\ao)^{5/2}(\rho\ao^3)^{\nu}),
\end{equation}
where
\begin{equation}
	\label{eq:fbog}
 f_{\rm Bog}(n,\ell) =  4\pi \frac{\ao}{\ell^3} n^2\left(1 + \frac{128}{15 \sqrt \pi} n^{1/2} \frac{\ao^{3/2}}{\ell^{3/2}}\right) + T \sum_{p\in \pi \mathbb{N}_0^{3} \setminus \{0\}} \log \left(1-e^{\tfrac{-1}{T\ell^2}\sqrt{p^4 + 16\pi \ao n \ell^{-1} p^2}}\right).
\end{equation}
\end{theorem}

From Theorem \ref{theo:free_energy_small_box}, our main result in Theorem \ref{thm:main} then essentially follows from the superadditivity of the free energy. 

We shall now explain the main ingredients in the proof of Theorem \ref{theo:free_energy_small_box}. In this introduction, to make the ideas transparent, we will not describe the error estimates in detail and simply write $A\approx B$ if the error is of order $\mathcal O(\ell^3 (\rho\ao)^{5/2}(\rho\ao^3)^{\nu})$ which appears in \eqref{eq:theo:free_energy_small_box}. 

We will use the Fock space formalism (see Section \ref{sec:Fockspace}) and  the unitary transformation 
$U:L^2_s(\Lambda^n)\to \cF_+^{\leq n} = \bigoplus_{m=0}^n (u_0^{\bot})^{\otimes_s m}$ 
introduced in \cite{LewNamSerSol-15}, defined in \eqref{def:U}, to factor out the contribution of the condensate described by the constant function $u_0=1\in L^2(\Lambda)$. As explained in Lemma \ref{lem_excitation_Hamil}, using the projection $\mathds{1}_+^{\leq n}$ onto $\mathcal{F}_+^{\leq n}\subset \cF=\bigoplus_{m=0}^\infty L^2_s(\Lambda^m)$ we can write 
\begin{equation*}
U H_{n,\ell} U^* = \mathds{1}_+^{\leq n} \mathcal{H} \mathds{1}_+^{\leq n},
\end{equation*}
where 
\begin{align}
	\label{eq_excitation_Hamil-intro}
\mathcal{H} \approx \frac{n^2}{2} V_{\ell}^{0000}   + Q_1 + \dG(-\Delta) + H_2^{(U)} + Q_2 + Q_3^{(U)} + Q_4 
\end{align}
is an operator on the full Fock space $\cF$, with
\begin{align} \label{ed1}
 V_{\ell}^{0000}   &= \int_{\Lambda^2} V_{\ell}(x-y)\dd x \dd y,\\
Q_1 &= {n}^{3/2} \int_{\Lambda^2} V_{\ell}(x-y) a_x^* \dx\dy +  \hc,
\\
Q_2 &= \frac{n}{2} \int_{\Lambda^2} V_{\ell}(x-y) a_x^* a_y^* \dx\dy +  \hc,
\\
Q_3^{(U)} &= \sqrt{(n-\cN + 1)_+} \int_{\Lambda^2} V_{\ell}(x-y) a_x^* a_y^* a_x \dx\dy +  \hc,
\\
Q_4 &= \frac{1}{2}\int_{\Lambda^2} V_{\ell}(x-y) a_x^* a_y^* a_x a_y \dx\dy,
\\
H_2^{(U)}&= n \int_{\Lambda^2} V_\ell(x-y) (a_x^*a_x + a_x^*a_y) \dx\dy - nV_{\ell}^{0000} \cN  \nonumber
\\ 
&\qquad - \left(\frac{1}{2} \int_{\Lambda^2} V_{\ell}(x-y) a_x^* a_y^* \dx\dy \,\cN +  \hc\right). \label{ed2}
\end{align}
Here $\cN$ is the number operator on Fock space.

We then conjugate this excitation Hamiltonian with the unitary maps $e^{\mathcal B_1}, e^{\mathcal B_c}$ and $e^{\mathcal B_2}$, where the kernels $\mathcal B_1, \mathcal B_2$ are quadratic in creation and annihilation operators and $\mathcal B_c$ is cubic in those, such that
\begin{itemize}
\item The first quadratic transformation $e^{\mathcal B_1}$ extracts the leading order of the correlation, effectively renormalizing $Q_2$ where the short-range interaction $V_\ell$ gets replaced by a long-range one of  mean-field type;

\item The cubic transformation  $e^{\mathcal B_c}$ removes the cubic term $Q_3^{(U)}$, and also renormalizes $H_2^{(U)}$ with a similar replacement for $V_\ell$;

\item The second quadratic transformation $e^{\mathcal B_2}$ diagonalizes the quadratic Hamiltonian emerging from the Bogoliubov approximation, thereby resulting in the correct ground state energy and  excitation spectrum.
\end{itemize}


Note that the Hamiltonian $\mathcal H$ is defined on the full Fock space $\mathcal F = \mathcal F(L^2(\Lambda))$ even though $U H_{n,\ell} U^*$ is only defined on the subspace $\mathcal F_+^{\leq n}$.
Lifting the restriction on the number of particles allows us to use the exact canonical commutation relations. Eventually, we will evaluate the error terms on the Gibbs state of the system, which lives on $\cF_+^{\leq n}$. Moreover, our estimates will always hold on $\mathcal F_+$, which is left invariant under the unitary transformations. 
In particular, we will often say that two expressions agree on $\mathcal F_+$ when they agree in the sense of quadratic forms on $\mathcal F_+$.

In the mean-field regime, where $V_\ell$ is replaced by a long-range potential, only one quadratic transformation is needed to justify the  Bogoliubov approximation. This was first done in \cite{Seiringer-11} in the periodic setting and extended to  trapped gases in $\R^3$ in \cite{GreSei-13,LewNamSerSol-15}. However, in the Gross--Pitaevskii regime,  the use of a cubic transformation is crucial to effectively get back to the mean-field regime and capture correctly the excitation spectrum.  This key idea was first implemented in \cite{BocBreCenSch-19} in the periodic setting and extended to general trapped cases in $\R^3$ in \cite{NamTri-21,BreSchSch-21}. An attempt of adapting this strategy to Neumann boundary conditions was given in \cite{BocSei-22}, but it is insufficient to obtain the correct excitation spectrum. Here we will resolve this issue. The main challenge for us is to choose the correct kernels $\cB_1$, $\cB_c$, and $\cB_2$ adapted to the Neumann boundary conditions and to be able to compute the action of the corresponding transformation to the LHY order.

\bigskip
\noindent{\bf Modified scattering solution.} To define the kernels we use the zero-scattering solution $\omega_\ell$ associated with $V_\ell$ (see Section \ref{sec:scattering}). Following the approach in \cite{NamTri-21}, we introduce  a modified scattering solution 
$$\omega_{\ell,\lambda}(x) = \omega_{\ell}(x) \chi_\lambda(x),$$
where $\omega_\ell(x) = \omega(\ell x)$ and $\chi_\lambda(x)=\chi(\lambda^{-1} x)$ with $\chi$ a fixed $C^\infty$ radial function approximating $\1_{|x|\le 1}$. The function $\omega_{\ell,\lambda}$ satisfies 
\begin{align}  \label{eq:scattering_equation_truncated-intro}
-\Delta \omega_{\ell,\lambda} =  \frac{1}{2} V_\ell (1-\omega_\ell) - \frac{1}{2}\epsilon_{\ell, \lambda},\quad \frac{1}{2}\epsilon_{\ell, \lambda} (x) =  \frac{\ao}{\ell} \lambda^{-3}\left(\frac{\chi''}{|\cdot|}\right) (\lambda^{-1} x).
\end{align}
In our final estimate we will eventually choose 
$$\ell^{-1} \ll \lambda  \ll 1.$$
The first constraint $\ell^{-1} \ll \lambda$, which is inspired by \cite{BocBreCenSch-19}, ensures that the range of $\epsilon_{\ell, \lambda}$ in \eqref{eq:scattering_equation_truncated-intro} is much longer than that of $V_\ell$, and hence in our calculation it plays the role of a renormalized version of $V_\ell$. Moreover, the second constraint $\lambda\ll 1$, which is inspired by \cite{NamTri-21}, ensures that $\ell \epsilon_{\ell,\lambda}$ tends to a delta interaction, thus simplifying several estimates and also enabling us to go beyond the Gross--Pitaevskii regime. 


Heuristically, as proposed in the previous works on the Gross--Pitaevskii regime  \cite{BocBreCenSch-19,NamTri-21,BreSchSch-21}, the correlation structure of particles can be encoded using two transformations $e^{\mathcal B_1}$ and $e^{\mathcal B_c}$. More precisely, by putting the scattering solution $-n\omega_{\ell,\lambda}$ in the kernels $\cB_1$, $\cB_c$, we hope to replace the short range potential $V_\ell$ by the longer-range one $\epsilon_{\ell, \lambda}$. The naive choice of $\cB_1$ 
$$
-n \int_{\Lambda^2} \omega_{\ell,\lambda}(x-y) a_x^* a_y^* \dd x \dd y - \hc,
$$
does not work in our case since the function $-n\omega_{\ell,\lambda}$ does not satisfy the Neumann boundary conditions. To fix this issue, we use a symmetrization technique as follows.
 


%
%
%
%
%
%
%
%
%

\bigskip
\noindent{\bf Neumann symmetrization.} We shall construct a kernel $\widetilde{K}(x,y)$ that can be interpreted as a symmetrized version of $-n\omega_{\ell,\lambda}(x-y)$ satisfying Neumann boundary conditions in an appropriate sense. It belongs to $H^1(\Lambda^2)$ and satisfies the following two useful properties:
\begin{align}\label{eq:tK-desired-1}
\widetilde{K}(x,y) = -n\omega_{\ell,\lambda}(x-y), \quad \forall x,y\in \{z\in \Lambda :  {\rm dist}(z,\partial \Lambda) > \lambda\}
\end{align}
and that the operator with kernel $\widetilde K(x,y)$ is diagonal in the Neumann basis, see (\ref{eq:def-tK-momentum-intro}).

The construction uses the same mirroring technique as in the construction of the Neumann Green's function.
Denoting 
$$\Lambda + z =\{ x+z: x\in \Lambda\}, \quad z \in \mathbb{Z}^3$$ 
we define the transformation
\begin{align}\label{eq:Pz-def-intro}
P_z: \Lambda \to \Lambda+z, \quad (P_z(x))_i = (-1)^{z_i} x_i + z_i,
\end{align}
which maps a point $x\in \Lambda$ to its mirror point in the box $\Lambda+z$. For a visual illustration in 2D, we refer to Figure \ref{fig:projection-example-2D}, where the mirror points of $x\in \Lambda$ are plotted in the neighboring boxes of $\Lambda$.

\begin{figure}[!h]  
\center
\begin{tikzpicture} [scale = 2.2]
\draw (-1.2,-1.2) grid (2.2,2.2);

\draw [very thick] (0,0) rectangle (1,1);
\node [above right] at (0,0) {$\Lambda$};

\draw[fill] (-0.8,1.35) circle [radius=.5pt] ;
\draw[fill] (0.8,1.35) circle [radius=.5pt] ;
\draw[fill] (1.2,1.35) circle [radius=.5pt] node [right] {\small $P_{(1,1)}(x)$};

\draw[fill] (-0.8,.65) circle [radius=.5pt] node [below right] {\small $P_{(-1,0)}(x)$};
\draw[fill] (.8,.65) circle [radius=.5pt] node [left]  {\small $x$};
\draw[fill] (1.2,.65) circle [radius=.5pt] node [right] {\small $P_{(1,0)}(x)$};

\draw[fill] (-0.8,-.65) circle [radius=.5pt] ;
\draw[fill] (0.8,-.65) circle [radius=.5pt] ;
\draw[fill] (1.2,-.65) circle [radius=.5pt] ;

\end{tikzpicture}
\caption{Relevant mirror points of $x$ shown in two dimensions. }
\label{fig:projection-example-2D}
\end{figure}
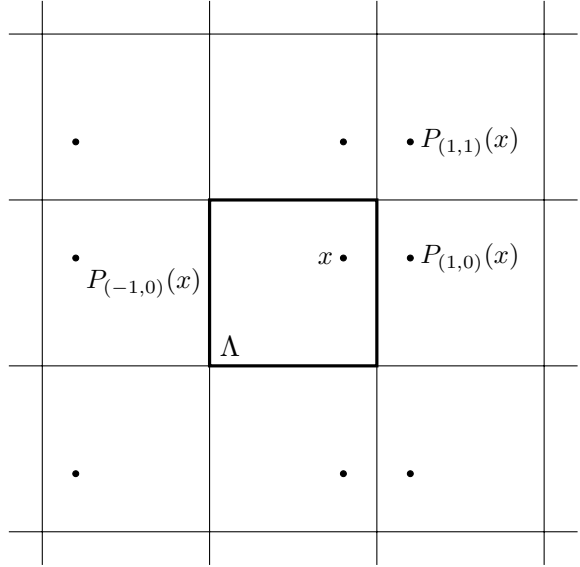

\begin{figure}  
\center

\begin{tikzpicture} [scale = 2]

\draw [] (-0.2,-0.2) grid (2.2,2.2);
\draw [very thick] (0,0) rectangle (1,1);
\node [above right] at (0,0) {$\Lambda$};
\draw[fill] (.9,.65) circle [radius=.5pt] node [below] {\small $x$};
\draw[fill] (1.1,1.35) circle [radius=.5pt] node [above]  {\small $P_z(x)$};
\draw[fill] (.7,.75) circle [radius=.5pt] node [left] {\small $y$};
\draw[fill] (1.3,1.25) circle [radius=.5pt] node [right]  {\small $P_z(y)$};

\draw [blue] (.9,.65) -- (1.3,1.25);
\draw (1.1,1.35) -- (.7,.75);

\end{tikzpicture}

\caption{The distance is conserved.}

\label{fig_integration} 

\end{figure}
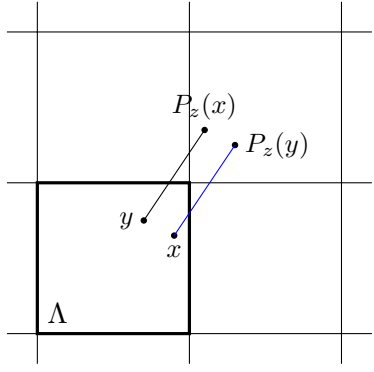

We  define the function $\widetilde{K}: \Lambda^2 \to \R$ as 
\begin{equation} \label{eq:tK-intro}
\widetilde{K}(x,y) = - \sum_{z \in \mathbb{Z}^3} n \omega_{\ell,\lambda}(P_z(x)-y). 
\end{equation}
Observe that while $\{\Lambda + z : z \in \mathbb{Z}^3\}$ covers all of $\R^3$,  due to the cutoff $\chi_\lambda$ there is a contribution to the sum only if $|P_z(x) - y| \leq \lambda \ll 1$ and the property \eqref{eq:tK-desired-1} follows immediately. Taking into account that $y \in \Lambda$ we find that in the last term of \eqref{eq:tK-intro} only the summands with
$$
z \in \mathbb{Z}^3_{\leq 1} := \{z \in \mathbb{Z}^3 : \max_i |z_i| \leq 1 \}
$$ 
are non-zero. Thus, the sum is finite and $\widetilde{K}$ is well-defined by \eqref{eq:tK-intro}. Moreover, $\widetilde{K}$ is symmetric, namely $\widetilde{K}(x,y)=\widetilde{K}(y,x)$, since $|P_z(x)-y| = |P_z(y)-x|$ for all $z \in \Z^3$ (see Figure \ref{fig_integration} for a 2D illustration). 

In fact, as we will see in Lemma \ref{lem_K_properties}, it is also diagonal in the Neumann basis
\begin{align} \label{eq:def-tK-momentum-intro}
\widetilde K(x,y) = -\sum_{p\in \pi \mathbb{N}_0^3} n \widehat \omega_{\ell,\lambda}(p) u_p(x) u_p(y), 
\end{align}
where $u_p \in L^2(\Lambda)$ are  Neumann eigenfunctions given in \eqref{eq:up_def} and we used the following convention of the Fourier transform
\begin{align} \label{eq:def_FT}
\widehat{f}(p)= \int_{\R^3} f(x) e^{-i p\cdot x}\dx.
\end{align}
Since the local property \eqref{eq:tK-desired-1}  does not obviously follow from \eqref{eq:def-tK-momentum-intro}, both the forms \eqref{eq:tK-intro} and  \eqref{eq:def-tK-momentum-intro} will be useful in the following.

Next, we  remove from the function $\widetilde{K}$ any contribution from the zero-momentum mode by using the projection $Q=1-|u_0\rangle \langle u_0|$. This results in the function
\begin{align}\label{eq:K-intro}
K(x,y) = (Q^{\otimes 2}\widetilde{K})(x,y) =  \widetilde{K}(x,y) + n\widehat{\omega}_{\ell,\lambda}(0) = - \sum_{p\in \pi \mathbb{N}_0^3\backslash\{0\}} n \widehat \omega_{\ell,\lambda}(p) u_p(x) u_p(y),
\end{align}
which is the key tool to define the transformations $e^{\cB_1}$ and $e^{\cB_c}$.

\bigskip
\noindent{\bf First quadratic transformation.} We define the first transformation kernel
\begin{equation} \label{eq_B1-intro}
\cB_1 = \frac{1}{2} \int_{\Lambda^2} K(x,y) a_x^* a_y^* \dd x \dd y - \hc\, 
\end{equation}
with $K$ given in \eqref{eq:K-intro}. We will show that by conjugating the excitation Hamiltonian $\cH$ in \eqref{eq_excitation_Hamil-intro} by the quadratic transformation $e^{\cB_1}$, we essentially renormalize $Q_2$ and extract the leading order contribution $4\pi\ao n^2 \ell^{-1}$ (see \Cref{lem_firsttransform}).

Note that thanks to the last identity in \eqref{eq:K-intro}, we may rewrite \eqref{eq_B1-intro} as 
\begin{align}\label{eq:def-B1-momentum}
\cB_1= \frac{1}{2} \sum_{p\in \pi \mathbb{N}_0^3 \setminus \{0\}} (-n\widehat \omega_{\ell,\lambda}(p)) a_p^* a_p^* - \hc,
\end{align}
where we denoted $a^*_p=a^*(u_p)$ the creation operator on Fock space. The formula \eqref{eq:def-B1-momentum}  can be compared with the kernel in the translation-invariant case in \cite{BocBreCenSch-19,HaiSchTri-22}, where $a_p^* a_{-p}^*$ is used instead of $a_p^* a_p^*$. However, the algebraic structure is not as nice as in the translation-invariant case and it is often more convenient to work in configuration space where pointwise estimates, the identity (\ref{eq:tK-desired-1}) or the non-negativity of $V$ are available.


We shall now explain some  details of the action of the transformation $e^{\cB_1}$. As we show in Lemma \ref{lem:comm_K_Q4_T1}, we have
\begin{align*}
[\dG(-\Delta) +Q_4 , \cB_1] &\approx  \int_{\Lambda^2} \Big( (-\Delta_2 K)(x,y) + \frac{1}{2} V_{\ell}(x-y) K(x,y) \Big) a_x^*a_y^* \dx\dy + \hc,
\end{align*}
which, together with the definition \eqref{eq:K-intro} and the scattering equation \eqref{eq:scattering_equation_truncated-intro}, gives 
\begin{align} \label{eq:B1-comm-heu-intro}
[\dG(-\Delta) + Q_4,\cB_1] + Q_2 \approx \widetilde{Q}_2= \int_{\Lambda^2} \widetilde{Q}_2(x,y) a_x^* a_y^* \dx\dy +  \hc,
\end{align}
where
$$
\widetilde{Q}_2(x,y) = \frac{n}{2}\sum_{z \in \mathbb{Z}^3} \epsilon_{\ell,\lambda}(P_z(x)-y) +\frac{n}{2} \sum_{z \in \mathbb{Z}^3 \setminus \{0\}} (V_\ell(\omega_{\ell,\lambda} - 1)(P_z(x)-y) - V_{\ell}(x-y) \omega_{\ell,\lambda}(P_z(x)-y)) .
$$
Here $\widetilde{Q}_2(x,y)$ contains the function $\epsilon_{\ell,\lambda}$ in the first sum, which can be interpreted as a renormalized version of the short-range potential $V_\ell$ and  inherits the symmetrization of $K$, as well as some boundary contribution in the second sum which will disappear after the cubic transformation $e^{\mathcal{B}_c}$. The approximation \eqref{eq:B1-comm-heu-intro} is exactly the motivation for the choice  of the kernel $\cB_1$, in the same spirit as in  \cite{NamTri-21,Hainzl-21}.  

From \eqref{eq:B1-comm-heu-intro} and the Duhamel expansion (see \eqref{eq_cancellation_trick} for an explanation) we can write 
\begin{align}\label{eq_cancellation_trick-intro}
&e^{-\cB_1}   (\dG\left( -\Delta \right) + Q_2 + Q_4)e^{\cB_1} - \dG(-\Delta) - Q_4 \nn\\
&= \int_0^1 e^{-t\cB_1} \Big([\dG\left( -\Delta \right) + Q_4,\cB_1]+Q_2\Big) e^{t\cB_1} \dt + \int_0^1 \int_t^1 e^{-s\cB_1} [Q_2,\cB_1] e^{s\cB_1} \ds \dt \nn\\
&\approx  \int_0^1 e^{-t\cB_1} \widetilde{Q}_2 e^{t\cB_1} \dt  + \int_0^1 \int_t^1 e^{-s\cB_1} [Q_2,\cB_1] e^{s\cB_1} \ds \dt \nn\\
&= \widetilde{Q}_2 + \int_0^1 \int_0^t  e^{-s\cB_1} [\widetilde{Q}_2,\cB_1] e^{s\cB_1} \ds \dt + \int_0^1 \int_t^1 e^{-s\cB_1} [Q_2,\cB_1] e^{s\cB_1} \ds \dt.
\end{align}
As proved in \Cref{prop_Q2}, the last two terms in \eqref{eq_cancellation_trick-intro} are essentially two constant contributions. In particular the last term helps us to correct the constant in \eqref{eq_excitation_Hamil-intro}  and we recover the full leading order of the energy, 
\begin{align*}
&\frac{1}{2}n^2 V_{\ell}^{0000} + \int_0^1 \int_t^1 e^{-s\cB_1} [Q_2,\cB_1] e^{s\cB_1} \ds \dt \\
&\approx \frac{1}{2}n^2 V_{\ell}^{0000} + \frac{1}{2} [Q_2,\cB_1] \approx  \frac{1}{2}n^2 V_{\ell}^{0000} + \frac{n}{2}\int_{\Lambda^2} V_{\ell}(x-y) K(x,y) \dx\dy \approx 4\pi\mathfrak{a}n^2 \ell^{-1}
\end{align*}
with an error smaller than the second order in the LHY formula. The other term is 
\begin{align}
\int_0^1 \int_0^t  e^{-s\cB_1} [\widetilde{Q}_2,\cB_1] e^{s\cB_1} \ds \dt \approx \frac{1}{2} [\widetilde{Q}_2,\cB_1] \approx \sum_{p \in \pi \mathbb{N}_0^3 \backslash\{0\}} \frac{|n \widehat{\epsilon}_{\ell,\lambda}(p)|^2}{2p^2},
\end{align}
which will be combined with another constant contribution coming from the transformation $e^{\cB_2}e^{\cB_c}$ to give the correct LHY second order term.

So far, we have seen that the quadratic transformation $e^{\cB_1}$ essentially replaces $\dG(-\Delta)+Q_4+Q_2$ by $\dG(-\Delta)+Q_4+\widetilde Q_2$ plus some  constants. Furthermore, we will  show in \Cref{prop_H2_quadratictrafo} that
\begin{align} \label{eq:H2-B1-intro} 
e^{-\cB_1} H_2^{(U)} e^{\cB_1} \approx n \int_{\Lambda^2} V_\ell(x-y) (a_x^*a_x + a_x^*a_y) \dx\dy - 8\pi\mathfrak{a} n \ell^{-1} \mathcal N,     
\end{align}
 namely the term $\cN nV_{\ell}^{0000}+\left(\frac{1}{2} \int_{\Lambda^2} V_{\ell}(x-y) a_x^* a_y^* \dx\dy\cN +  \hc\right)$ in $H_2^{(U)}$  is replaced by $8\pi\mathfrak{a} \frac{n}{\ell} \mathcal N$. Moreover, as proved in \Cref{prop_Q3_quadratictrafo} we have
\begin{align} \label{eq:Q3-intro}
e^{-\cB_1} (Q_1+Q_3^{(U)}) e^{\cB_1}\approx   Q_3 =  \sqrt{n} \int_{\Lambda^2} V_\ell(x-y) a_x^*a_y^*a_x \dx\dy +  \hc
\end{align}
The cubic term $Q_3$ then will be handled by the cubic transformation $e^{\cB_c}$ below. In summary we have
\begin{align*}
e^{-\cB_1}\cH e^{\cB_1} &\approx 4\pi\ao n^2 \ell^{-1} + \dG(-\Delta) + Q_4 + \widetilde Q_2 + Q_3 + n\int_{\Lambda^2} V_\ell(x-y) (a_x^*a_x + a_x^*a_y) \dx\dy 
\\
& \qquad - 8\pi\mathfrak{a} n \ell^{-1} \mathcal N.
\end{align*}

\bigskip
\noindent{\bf Cubic transformation.} Next, for the cubic transformation, we define 
\begin{align} \label{eq:def_Bc-intro}
\cB_c = \frac{\theta_M(\mathcal N)}{\sqrt{n}}  \int_{\Lambda^2}  K(x,y) q_x^* a_y^* q_x  \dd x \dd y -  \hc
\end{align}
Here $q_x=a(Q_x)$, where $Q_x(y) = Q(x,y)$, is used instead of $a_x$ to ensure that $\cB_c$ leaves $\cF_+$ invariant, and $\theta_M(\mathcal N)$ is a smooth cut-off on the sector $\{\cN\le M\}$, with $1\ll M\ll n$, which prevents $e^{\cB_c}$ from creating too many excitations. 

As proved in \Cref{lem_cubictrafo}, by using $e^{\mathcal B_c}$ we can remove the cubic term $Q_3$ in \eqref{eq:Q3-intro} and also renormalize some quadratic terms. More precisely, using the Duhamel formula, we can expand  
\begin{align}\label{eq:cubic-action-intro}
 &e^{-\cB_c}  \Big(\dG(-\Delta) + Q_4 + Q_3\Big) e^{\cB_c} - \dG(-\Delta) - Q_4 \nn\\
 &=  \int_0^1 e^{-t\cB_c} \Big( [\dG(-\Delta)+Q_4,\cB_c]+Q_3 \Big)e^{t\cB_c} \dt + \int_0^1 \int_t^1 e^{-s\cB_c}[Q_3,\cB_c]e^{s\cB_c}\ds\dt .
\end{align}
We have chosen the cubic kernel $\cB_c$ such that
$$
[\dG (-\Delta)+Q_4,\cB_c] + Q_3 \approx 0,
$$
and hence the first term on the right-hand side of \eqref{eq:cubic-action-intro} is negligible. Moreover, the last term in   \eqref{eq:cubic-action-intro} can be put together with the transformations of \eqref{eq:H2-B1-intro}, and we can show that 
$$
e^{-\cB_c} \Big(\int_{\Lambda^2} nV_\ell(x-y) (a_x^*a_x + a_x^*a_y)  \dx\dy - 8\pi\mathfrak{a} n\ell^{-1} \mathcal N\Big) e^{\cB_c} + \int_0^1 \int_t^1 e^{-s\cB_c}[Q_3,\cB_c]e^{s\cB_c}\ds\dt
$$
is essentially $8\pi\ao n\ell^{-1} \mathcal N$. After the cubic transformation, we may remove the boundary contribution in $\widetilde{Q}_2$ and obtain the desired  pairing term (see \Cref{prop_Q2_commutator_cubic})
$$
e^{-\cB_c} \widetilde{Q}_2 e^{\cB_c} \approx  \frac{n}{2} \int_{\Lambda^2}\sum_{z \in \mathbb{Z}^3} \epsilon_{\ell,\lambda}(P_z(x)-y) a_x^* a_y^* \dx\dy + \hc
$$
which coincides with 
$$
 \frac{1}{2} \sum_{p \in \pi\mathbb{N}_0^3 \setminus \{0\} } 
 n\widehat{\epsilon}_{\ell,\lambda}(p)(a_p^*a_p^*+a_pa_p)
 $$
 when restricted to $\cF_+$. Thus we arrive at 
$$
e^{-\cB_c} e^{-\cB_1} \cH e^{\cB_1}e^{\cB_c} \approx 4\pi\ao n^2\ell^{-1}  +  \mathbb{H}_{\rm Bog} + Q_4
$$ 
with the quadratic Bogoliubov  Hamiltonian 
\begin{align}
	\label{eq:H_mom-intro}
\mathbb{H}_{\rm Bog} &= \sum_{p \in \pi\mathbb{N}_0^3 \setminus \{0\} } \left(p^2+8\pi\mathfrak{a} \frac{n}{\ell}\right) a_p^*a_p + \frac{1}{2} \sum_{p \in \pi\mathbb{N}_0^3 \setminus \{0\} } 
 n\widehat{\epsilon}_{\ell,\lambda}(p)(a_p^*a_p^*+a_pa_p) \nn
 \\
 & \qquad + \frac{1}{2} \sum_{p \in \pi\mathbb{N}_0^3 \setminus \{0\}}
\frac{|n\widehat{\epsilon}_{\ell,\lambda}(p)|^2}{2p^2}.
\end{align}

\bigskip
\noindent{\bf Second quadratic transformation.} 
It is well-known that the quadratic operator in \eqref{eq:H_mom-intro} can be diagonalized explicitly. To be precise, by choosing the second quadratic transformation $e^{\cB_2}$ with  
\begin{align} 
\cB_2 = \frac{1}{2}\sum_{p \in \pi\mathbb{N}_0^3 \setminus \{0\}} \vphi_p (a_p^*a_p^*-a_pa_p), \label{eq:B2_def-intro}
\end{align}
for suitable $\vphi_p$
one can show that
$$
 e^{-\cB_2} \mathbb{H}_{\rm Bog}  e^{\cB_2} + 4\pi \ao n^2 \ell^{-1} \approx E_{n,\ell} + \dd \Gamma(E_{\rm Bog})  
$$
on $\cF_+$ with the ground state energy
\begin{align*}
 E_{n,\ell} 
 	:=  4\pi \ao n^2 \ell^{-1}  + \frac{1}{2}\sum_{p \in \pi\mathbb{N}_0^3 \setminus \{0\}} \left[ \sqrt{p^4+16\pi \ao n \ell^{-1}p^2} - p^2 - 8\pi\mathfrak{a}\frac{n}{\ell} + \frac{(8\pi \mathfrak{a}n \ell^{-1} )^2}{2p^2} \right]
\end{align*}
and the  effective Hamiltonian
$$
\dd \Gamma(E_{\rm Bog})   = \sum_{p \in \pi\mathbb{N}_0^3 \setminus \{0\}} \sqrt{p^4+16\pi \ao n \ell^{-1}p^2} \, a_p^*a_p.
$$
In this way, we recover all information on the excitation spectrum predicted by the Bogoliubov approximation \cite{Bogoliubov-47}. 

\bigskip
\noindent{\bf Localization on Fock space.} On the technical level, we can only estimate the relevant errors in the above analysis on the low particle number sectors in $\cF_+^{\le n}$. The high particle number sectors have to be handled differently. By adapting the analysis  in \cite{LieSei-02}, we are able to prove  BEC for the Gibbs state, i.e. there are only few excitations. In combination with the Gibbs variational principle, this allows to ignore the free energy coming from the high particle number sectors. Finally, using the  localization method on the number of excited particles in the spirit of \cite{LieSol-01,LewNamSerSol-15}, we put together the low and high particle number sectors, thus concluding the proof of Theorem \ref{theo:free_energy_small_box}.

\bigskip
\noindent{\bf Organization of the proof.} In Section \ref{sec:prelim} we recall general properties of the scattering length $\ao$ and the scattering solution. We also introduce some notation on the Fock space formalism; in particular we use the excitation map $U$ defined in \cite{LewNamSerSol-15} to link $H_{n,\ell}$ to an excitation Hamiltonian on the Fock space of excitations $\mathcal F(u_0^\perp)$. In Section \ref{sec:transf_kernel}, we explain in detail the construction of the Neumann kernel $K(x,y)$ by symmetrization. We then conjugate this excitation Hamiltonian with the unitary maps $e^{\mathcal B_1}, e^{\mathcal B_c}$ and $e^{\mathcal B_2}$. The actions of the transformations  $e^{\mathcal B_1}, e^{\mathcal B_c}$ and $e^{\mathcal B_2}$ are carried out in Sections \ref{sec:firsttransform}, \ref{sec:cubic_trafo} and \ref{sec:last_trafo}, respectively. In Section \ref{sec:a-priori}, we prove  BEC for the Gibbs state associated with $\mathbb{H}_{n,\ell}$ as well as derive some rough estimates for the kinetic and interaction energies, which are needed for the localization technique on the number of excited particles. Finally, we prove Theorem \ref{theo:free_energy_small_box} in Section \ref{sec:proof_theo2}, and conclude Theorem \ref{thm:main} in Section \ref{sec:proof-thm-1}. 

\bigskip
\noindent{\bf Notation.} We always use $C>0$ to denote a general constant which depends only on $V$. We also write $A\lesssim B$ if $A\le CB$, and write $A\ll B$ if $A/B\to 0$ when $\rho\ao^3\to 0$. Moreover, all operator inequalities are interpreted as quadratic forms, namely we write $S\ge T$ on $\mathfrak{H}$ if $\langle u,S u\rangle \ge \langle u, T u\rangle$ for all $u \in \mathfrak{H}$ (which is in particular convenient when $S$ and $T$ act on a larger Hilbert space and do not leave $\mathfrak{H}$ invariant). When writing an operator in terms of the distributional creation and annihilation operator, we omit the integration variable for shortness if it is unambiguous, e.g. we write $\int T(x,y) a^*_x a_y$ instead of $\int T(x,y) a^*_x a_y \dd x \dd y$. 

\bigskip
\noindent{\bf Acknowledgments.}
This work was partially funded by the Deutsche Forschungsgemeinschaft (DFG, German Research Foundation) – Project-ID 470903074 – TRR 352. 
PTN was partially supported by the European Research Council (ERC CoG RAMBAS, Project Nr.  101044249). 

\section{Preliminaries}
	\label{sec:prelim}
	
	In this section we collect some standard tools, which are helpful to transform the Hamiltonian $H_{n,\ell}$ in \eqref{eq:def_Hnl}.

\subsection{Scattering Problem} \label{sec:scattering}

Here we recall some well-known facts about the scattering length of the potential $V$ and its scattering solution. Under the assumption that $V$ is non-negative, compactly supported and radially symmetric, it is well-known (see e.g. \cite[Appendix C]{LieSeiSolYng-05} or \cite[Section 2]{NamRicTri-21}) that the equation
\begin{align} \label{eq:scattering_equation-0}
-\Delta \omega = \frac{1}{2} V (1-\omega) \text{ in }\R^3, \quad \lim_{|x|\to \infty} \omega (x)=0
\end{align}
has a unique solution $\omega$ satisfying $0 \leq \omega \leq 1$. The scattering length $\ao$ of $V$ is  defined as
\begin{align} \label{eq:scattering_length_int}
8\pi \ao &= \int_{\R^3} V (1-\omega).
\end{align}
Since the scattering solution $\omega$ is harmonic outside the support of $V$, we have the exact formula for $x \in \mathbb{R}^{3} \setminus \supp (V)$: 
\begin{align} \label{eq:omega-full-a}
\omega(x) = \frac{\ao}{|x|}. 
\end{align}

In our application, we will consider a modified version of the scattering solution $\omega$. Let $\ell$ be as in \eqref{eq:def_ell}. Defining $\omega_\ell(x) = \omega(\ell x)$, and recalling that $V_\ell(x) = \ell^2 V(\ell x)$, obviously from \eqref{eq:scattering_equation-0} we have the rescaled equation 
\begin{equation} \label{eq:scattering_equation}
- \Delta \omega_\ell = \frac{1}{2} V_\ell (1-\omega_\ell).
\end{equation}
Next, let us introduce a cut-off version of $\omega_\ell$. Following \cite{NamTri-21}, for 
\begin{equation} \label{eq:lambda-choice}
2R/\ell < \lambda < 1/4,
\end{equation}
we define 
$$\omega_{\ell,\lambda}(x) = \omega_{\ell}(x) \chi_\lambda(x), \quad \chi_\lambda(x)=\chi(\lambda^{-1} x),$$
where $\chi$ is a fixed $C^\infty$ radial function satisfying 
$$\chi(x) = 0\text { for }|x| \geq 1\quad \text{and}\quad \chi(x) = 1 \text { for } |x| < \frac{1}{2}.$$ 

Under the assumption that $\supp\, V \subset \{|x|\leq R\}$ and that $\lambda > 2 R \ell^{-1}$, the truncated scattering solution $\omega_{\ell,\lambda}$ satisfies a modified scattering equation
\begin{align}  \label{eq:scattering_equation_truncated}
-\Delta \omega_{\ell,\lambda} =  \frac{1}{2} V_\ell (1-\omega_{\ell}) - \frac{1}{2}\epsilon_{\ell, \lambda},
\end{align}
where
\begin{align}  \label{eq:def_omega_ell}
\frac{1}{2}\epsilon_{\ell, \lambda} 
	&= \Delta (\omega_{\ell,\lambda} -   \omega_\ell) = 2 \nabla \omega_\ell \cdot \nabla \chi_\lambda + \omega_\ell \Delta \chi_\lambda.
\end{align}
From (\ref{eq:def_omega_ell}) we see that $\epsilon_{\ell,\lambda}$ and $V_\ell$ have disjoint support. Therefore, we may use \eqref{eq:omega-full-a} in (\ref{eq:def_omega_ell}) and that $\chi$ is radial to arrive at
\begin{equation}
\label{eq:def_omega_ell_simple}
\frac{1}{2}\epsilon_{\ell, \lambda} (x) =  \frac{\ao}{\ell} \lambda^{-3}\left(\frac{\chi''}{|\cdot|}\right) (\lambda^{-1} x),
\end{equation}
where we interpreted $\chi(x)=\chi(|x|)$ when writing $\chi''$. 

Finally, we gather some of their properties in the following lemma.

\begin{lemma} 	\label{lem_vep}
	 Let $2 R/\ell < \lambda$. Then for all $x\in \R^3$, we have the pointwise bounds
\begin{align} \label{eq:w-pointwise}
0 \leq \omega_{\ell,\lambda} (x) \leq  \frac{C \mathds{1}_{ \{|x| \leq \lambda \} }}{|\ell x|+1},\quad 
|\nabla \omega_{\ell,\lambda} (x) | \leq \frac{C \ell  \mathds{1}_{\{|x|\leq \lambda \}}}{|\ell x|^2+1}, \quad 
|\epsilon_{\ell,\lambda} (x)| \leq \frac{C}{\ell} \lambda^{-3} \mathds{1}_{ \{\lambda/2 \leq |x| \leq \lambda \}}.
\end{align}
Moreover,
\begin{align}\label{item:prop:prop_epsilon_1-Vwe} 
\int_{\R^3} \epsilon_{\ell, \lambda}  =  8 \pi \ao \ell^{-1}.
\end{align}
\end{lemma}

Note that the last bound in \eqref{eq:w-pointwise} implies that $\epsilon_{\ell,\lambda}(P_z(x)-y) = 0$ for $z \notin \Z^3$.

\begin{proof}
From \eqref{eq:omega-full-a} and $0 \leq \omega \leq 1$ we obtain 
$$
0\leq \omega (x) \leq \frac{C}{|x|+1}, \quad
|\nabla \omega (x) | \leq \frac{C}{|x|^2+1}.
$$
Moreover $\supp \left( \chi_\lambda \right) \subset B_\lambda(0)$ and $|\nabla \chi_\lambda| \leq C\lambda^{-1} \mathds{1}_{\{\lambda/2 \leq |x| \leq \lambda \}}$, which implies the first two bounds in  \eqref{eq:w-pointwise}. 
The last bound in \eqref{eq:w-pointwise} follows from (\ref{eq:def_omega_ell}) and
$\supp (\chi'') \subset \{\lambda/2 \leq |x| \leq \lambda\}$. 
Finally, since $\omega_{\ell,\lambda}$ is compactly supported, from \eqref{eq:scattering_equation_truncated} and \eqref{eq:scattering_length_int} we have 
\begin{align*}
0=2 \int_{\mathbb{R}^{3}} \Delta \omega_{\ell,\lambda} = \int_{\mathbb{R}^{3}}  V_\ell (1-\omega_\ell) - \int_{\mathbb{R}^{3}} \epsilon_{\ell, \lambda}= 8\pi \ao \ell^{-1} - \int_{\mathbb{R}^{3}} \epsilon_{\ell, \lambda}.
\end{align*}
This implies (\ref{item:prop:prop_epsilon_1-Vwe}).
\end{proof}
 
\subsection{Fock Space Formalism} \label{sec:Fockspace}
For $m \in \pi \mathbb{N}_0^3=\pi \{0,1,2,...\}^3$, let us denote 
\begin{equation} \label{eq:up_def}
u_m(x) = \prod_{i=1}^3 u_{m_i}(x_i),\quad u_{m_i}(x) = \left\{ \begin{array}{cc} 1 , & m_i = 0 \\
\sqrt{2}\cos(m_i(x_i+1/2)), & m_i \neq 0 \end{array} 
\right. .
\end{equation}
The family $\{u_m\}_{m \in \pi\mathbb{N}_0^3}$ is an orthonormal basis of $L^2(\Lambda)$ satisfying  Neumann boundary conditions. A special role is played by the condensate function
$u_0 = \mathds{1}_\Lambda.$

Given a Hilbert space $\mathfrak{H}$, we consider 
$$\cF(\mathfrak{H}) = \bigoplus_{n \geq 0}  \mathfrak{H}^{\otimes_s n},\quad \mathcal F^{\leq k}(\mathfrak{H}) = \bigoplus_{n = 0}^k  \mathfrak{H}^{\otimes_s n} ,$$
the bosonic Fock space over $\mathfrak{H}$ and its truncated version, respectively. In our application, we focus on the cases where $\mathfrak{H} =L^2(\Lambda)$ or the subspace $u_0^\perp \subset L^2(\Lambda)$, and  we will denote respectively the Fock spaces 
$$\mathcal F = \mathcal F(L^2(\Lambda)),\quad \mathcal F_+ = \mathcal F(u_0^\perp).$$
The bosonic creation and annihilation operators are given by 
\begin{align*}
	(a^* (g) \Psi )(x_1,\dots,x_{n+1})&= \frac{1}{\sqrt{n+1}} \sum_{j=1}^{n+1} g(x_j)\Psi(x_1,\dots,x_{j-1},x_{j+1},\dots, x_{n+1}), \\
	(a(g) \Psi )(x_1,\dots,x_{n-1}) &= \sqrt{n} \int_{\R^3} \overline{g(x_n)}\Psi(x_1,\dots,x_n) \dd x_n,
\end{align*}
for any $g\in L^2(\Lambda), \Psi \in L^2_s(\Lambda^n)$ and $n\geq 0$.
We will also use the short-hand notations $a_p = a(u_p), a^*_p = a^*(u_p)$ as well as the operator-valued distributions $a_x^*$ and $a_x$, with $x\in \Lambda$, which satisfy
\begin{align*}
a^*(g)=\int_{\Lambda}   g(x) a_x^* d x, \quad a(g)=\int_{\Lambda}  \overline{g(x)} a_x d x
\end{align*}
for all $g\in L^2(\Lambda)$.
These operators satisfy the canonical commutation relations
\begin{align*}
 [a(g_1),a(g_2)]=[a^*(g_1),a^*(g_2)]=0,\quad [a(g_1), a^* (g_2)]= \langle g_1, g_2 \rangle
\end{align*}
for all $g_1,g_2$ in $L^2(\Lambda)$, and 
\begin{align*}
 [a^*_x,a^*_y]=[a_x,a_y]=0, \quad [a_x,a^*_y]=\delta_{x,y},
\end{align*}
for all $x,y$ in $\Lambda$. For any one body operator $A$ with coefficients $A_{p,q} = \braket{u_p, A u_q}$ and kernel $A(x,y)$, we define its second quantized form 
\begin{align*}
\dd\Gamma(A) = \sum_{p,q \in \pi\mathbb{N}_0^3} A_{p,q} a_p^*a_q = \int_{\Lambda^2} A(x,y) a^*_x a_y \dx\dy.
\end{align*}
In particular, the particle number and the excitation number are denoted
\begin{align*}
\mathcal N = \dd\Gamma(\1) = \sum_{p \in \pi\mathbb{N}_0^3} a_p^*a_p, \qquad \mathcal N_+ = \dd\Gamma(Q) = \sum_{p \neq 0}  a_p^*a_p,
\end{align*}
where we introduced the notation $\sum_{p \neq 0} := \sum_{p \in \pi\mathbb{N}_0^3\setminus\{0\}} $ and recall $Q = 1 -|u_0\rangle \langle u_0|$. Additionally, we denote the orthogonal projections onto the excitation Fock space $\mathcal F_+$ and onto the truncated Fock space $\mathcal F_+^{\leq n}$ by $\1_+ = \1^{\{\mathcal N = \mathcal N_+ \}}$ and $\1_+^{\leq n} = \1^{\{\mathcal N = \mathcal N_+ \}} \1^{\{\mathcal N_+ \leq n \}}$ for $n\geq 0$, respectively.

With this formalism, the $n-$particle Hamiltonian in \eqref{eq:def_Hnl} can be written as 
\begin{align}
H_{n,\ell} 
	&= \sum_{p \in \pi\mathbb{N}_0^3} p^2 a_p^*a_p + \frac{1}{2}\sum_{p,q,r,s \in \pi \mathbb{N}_0^3} V_\ell^{pqrs} a_p^* a_q^* a_r a_s \label{eq:2nd-Q_1} \\
	&= \int_{\Lambda} \nabla_x a^*_x \nabla_x a_x \dx + \frac{1}{2} \int_{\Lambda^2} V_\ell(x-y) a^*_xa^*_y a_x a_y \dx\dy, \label{eq:2nd-Q}
\end{align}
where we have denoted 
$$V_\ell^{pqrs} = \braket{u_p \otimes u_q, V_\ell u_r \otimes u_s}_{L^2(\Lambda^2)}.$$
The right-hand side \eqref{eq:2nd-Q} is an operator on Fock space $\cF=\mathcal F(L^2(\Lambda))$ but we will always consider its restriction to the $n$-particle sector which coincides with the expression in \eqref{eq:def_Hnl}.

\subsection{The Excitation Hamiltonian}
	\label{sec:excitation}

In this section, we will rewrite the Hamiltonian $H_{n,\ell}$ in the Fock space of excitations $\cF_+ = \cF(u_0^\perp)$. We do so by using the unitary transformation $U: L^2_s(\Lambda^n) \to \cF_+^{\leq n}$ introduced in \cite{LewNamSerSol-15} 
\begin{equation}\label{def:U}
U (\Psi) = \bigoplus_{j=0}^n \frac{1}{\sqrt{(n-j)!}} Q^{\otimes j} a_0^{n-j} \Psi.
\end{equation}
On $\cF_+^{\leq n}$ and for all $p,q \neq 0$, it satisfies
\begin{align}
U a_0^*a_0 U^*& = n - \cN_+, \qquad\qquad U a_p^*a_q U^* = a_p^*a_q, \\
U a_p^*a_0 U^* &=  a^*_p\sqrt{n-\cN_+}, \qquad  U a_0^*a_q U^* = \sqrt{n-\cN_+} a_q.
 \label{eq_U_action}
\end{align}
Implementing these transformations on $U H_{n,\ell} U^*$, we obtain the following lemma.
\begin{lemma} \label{lem_excitation_Hamil} Let $H_{n,\ell}$ be as  in \eqref{eq:def_Hnl}. We have the following operator identity on $\mathcal{F}_+^{\leq n}$
\begin{equation*}
U H_{n,\ell} U^* = \mathds{1}_+^{\leq n} \mathcal{H} \mathds{1}_+^{\leq n},
\end{equation*}
where 
\begin{align}
	\label{eq_excitation_Hamil}
\mathcal{H} = \frac{n^2}{2} V_{\ell}^{0000} + Q_1 + \dG(-\Delta) + H_2^{(U)} + Q_2 + Q_3^{(U)} + Q_4 + \mathcal{E}^{(U)}
\end{align}
is an operator on the full Fock space $\cF$, $V_{\ell}^{0000}$, $Q_1$, $Q_2$, $Q_3^{(U)}$, $Q_4$, $H_2^{(U)}$ are given in \eqref{ed1}--\eqref{ed2}  and the error term 
$\mathcal{E}^{(U)}$, given by \eqref{eq:def_EU}, satisfies
\begin{align} \label{eq:E_error_excitation}
\pm \mathcal{E}^{(U)} \leq C\frac{n^\frac{1}{2}(\cN+1)^\frac{3}{2}}{\ell} + \varepsilon n^{-1} Q_4 + \varepsilon^{-1} C \frac{n}{\ell},\quad \forall \varepsilon >0,
\end{align}
 on $\mathcal{F}$.
\end{lemma}

\begin{proof}


The computation of $\mathcal{E}^{(U)}$ is standard, see for instance \cite[Section 4]{LewNamSerSol-15}: conjugating the Hamiltonian $H_{n,\ell} $ in  \eqref{eq:2nd-Q} with $U$ and applying the rules \eqref{eq_U_action}, we obtain $U H_{n,\ell} U^* = \mathds{1}_+^{\leq n} \mathcal{H} \mathds{1}_+^{\leq n}$ with $\mathcal{H}$ given by (\ref{eq_excitation_Hamil}) and $\mathcal{E}^{(U)}$ by 
\begin{align}
\mathcal{E}^{(U)} &= \mathds{1}_+^{\leq n} \bigg( \int_{\Lambda^2} V_{\ell}(x-y)\dx\dy \frac{\cN^2 - n + \cN}{2} - \int_{\Lambda^2} V_\ell(x-y) \left(a_x^*a_x + a_x^*a_y\right) \dx \dy \; \cN  \nn \\
&\qquad\qquad + \int_{\Lambda^2} V_\ell(x-y) a_x^* \dx \dy \left((n-\cN - 1)\sqrt{n-\cN} - n^\frac{3}{2} \right) +  \hc \nn \\
&\qquad\qquad + \frac{1}{2} \int_{\Lambda^2} V_\ell(x-y) a_x^*a_y^* \dx \dy \left(\sqrt{n-\cN-1}\sqrt{n-\cN} - n + \cN \right) +  \hc \bigg) \mathds{1}_+^{\leq n} \nn 
\\
&=: \mathds{1}_+^{\leq n} \Big( \mathcal{E}^{(U,0)} + \mathcal{E}^{(U,1)} + \mathcal{E}^{(U,2)} \Big) \mathds{1}_+^{\leq n}. \label{eq:def_EU}
\end{align}
We shall estimate the right-hand side of \eqref{eq:def_EU} term by term. Due to the projections it is enough to estimate $\mathcal{E}^{(U,0)}, \mathcal{E}^{(U,1)}$ and $ \mathcal{E}^{(U,2)}$ on $\cF_+^{\leq n}$.
Using the Cauchy--Schwarz inequality we obtain
\begin{align*}
\pm \mathcal{E}^{(U,0)} 
	& \leq \int_{\Lambda^2} V_{\ell}(x-y)\dx\dy \left(\mathcal N^2 + \mathcal N + n\right) + 2\int_{\Lambda^2} V_\ell(x-y) a_x^*a_x \dx \dy \; \cN \\
	& \leq C \frac{(\cN+1)^2 + n}{\ell} 
\leq C \frac{n^{1/2}(\cN+1)^{3/2} + n}{\ell}
\end{align*}
on $\mathcal F_+^{\leq n}$, where we used that $\|V_\ell\|_1 \leq C \ell^{-1}$. 

For $\mathcal{E}^{(U,1)}$ we use the elementary inequality $|\sqrt{1-t} - 1| \leq Ct$ for $ 0 \leq t \leq 1$ to obtain
\begin{align*}
\pm \mathcal{E}^{(U,1)} 
	&= \pm a^*(V_\ell \ast u_0^2) \left[ n^\frac{3}{2} \bigg(\sqrt{1 - \frac{\cN}{n}} - 1\bigg) - (\cN+1) \sqrt{n-\cN} \right]  +  \hc
\\
&\leq \varepsilon_1 a^*(V_\ell \ast u_0^2)\left(\mathcal N+1\right)^{1/2}a(V_\ell \ast u_0^2) + \varepsilon_1^{-1}n^3\left(\sqrt{1 - \frac{\cN}{n}} - 1\right)^2 (\cN+1)^{-\frac{1}{2}}
\\
&\quad + \varepsilon_1^{-1} \left((\cN+1)\sqrt{n-\cN}\right)^2 (\cN+1)^{-\frac{1}{2}}
\\
&\leq  C\varepsilon_1 \ell^{-2}  (\cN+1)^\frac{3}{2} +\varepsilon_1^{-1} C n (\cN+1)^{3/2} \\
&\leq C n^\frac{1}{2}\ell^{-1}(\cN+1)^\frac{3}{2},
\end{align*}
on $\mathcal F_+^{\leq n}$, where we used that $\|V_\ell \ast u_0^2\|_{2} \leq \|V_\ell\|_1 \leq C \ell^{-1}$ and optimized over $\varepsilon_1 > 0$.
We proceed similarly for $\mathcal{E}^{(U,2)}$, using the elementary bound 
${ |\sqrt{1-t-n^{-1}}\sqrt{1-t} - 1 + t | \leq n^{-1}}$
for all $0 \leq t \leq 1-n^{-1}$. We obtain for all $\varepsilon >0$
$$
\pm \mathcal{E}^{(U,2)} 
 \leq \varepsilon n^{-1} Q_4 + \varepsilon^{-1} n^3 \|V_{\ell}\|_1 \left(\sqrt{1-\frac{\cN+1}{n}}\sqrt{1-\frac{\cN}{n}} - 1 + \frac{\cN}{n}\right)^2
	\leq \varepsilon n^{-1} Q_4 + \varepsilon^{-1} C \frac{n}{\ell}
$$
on $\mathcal F_+^{\leq n}$. The proof of Lemma \ref{lem_excitation_Hamil} is complete. 
\end{proof}

\section{Symmetrization and Neumann Boundary Conditions} 
	\label{sec:transf_kernel}

As already explained in the introduction, we cover $\R^3$ with copies of the box $\Lambda = [-\frac{1}{2},\frac{1}{2}]^3$ and label them canonically by $z+\Lambda$ with $z \in \mathbb{Z}^3$. Then we define $P_z: \Lambda \to \Lambda+z$ as in \eqref{eq:Pz-def-intro} and define $K,\widetilde{K}: \Lambda^2 \to \R$ as in \eqref{eq:tK-intro} and \eqref{eq:K-intro},
\begin{align} \label{eq:K-tK}
K = Q^{\otimes 2}\widetilde{K},\quad Q=1-|u_0\rangle \langle u_0|,\quad \widetilde{K}(x,y) = - \sum_{z \in \mathbb{Z}^3} n \omega_{\ell,\lambda}(P_z(x)-y). 
\end{align}

%
%
%
%
%


Let us now collect some useful bounds and properties of the function $K$. In particular we show that its $L^2-$norm is small if $\lambda$ is, and that it is diagonal in the basis of Neumann eigenfunctions.

\begin{lemma} [Properties of $K$] \label{lem_K_properties} Assume $2 R/\ell < \lambda < 1/4$. Then we have
\begin{align}\label{eq:K-lemma-property}
K(x,y) =  \widetilde{K}(x,y) + n\widehat{\omega}_{\ell,\lambda}(0) = - \sum_{p\in \pi \mathbb{N}_0^3\backslash\{0\}}  n\widehat \omega_{\ell,\lambda}(p) u_p(x) u_p(y).
\end{align}
Moreover, there is a constant $C > 0$ such that
$$
\|K\|_\infty \leq C n,\quad 
\|K\|_2^2 \leq C \lambda \left(\frac{n}{\ell}\right)^2,\quad \sup_{x\in \Lambda} \|K_x\|_2^2 \leq C\lambda \left(\frac{n}{\ell}\right)^2,
$$
where $K_x(y) := K(x,y) = K(y,x)$. 
\end{lemma}

Let us start with the following useful identity.


\begin{lemma}(Coefficients in the Neumann basis) \label{prop_symmetric_momentum}
Let $f:\R^3 \to \R$ be radial and integrable with $\supp(f) \subset \Lambda$. Then for all $p,q \in \mathbb{N}_0^3$ we have
\begin{align*}
\int_{\Lambda^2}  \sum_{z \in \mathbb{Z}^3} f(P_z(x)-y) u_p(x) u_q(y)\dx \dy =\delta_{p,q} \widehat{f}(p)
\end{align*}
with the Fourier transform defined in \eqref{eq:def_FT}. 
\end{lemma}

\begin{proof} First note that terms with $\max\{|z_i|,i=1,2,3\} \geq 2$ are zero due to our assumption on the support of $f$. Using simple coordinate transformations and $u_p(P_z(x)) = u_p(x)$ for all $p \in \pi \mathbb{N}_0^3$
we obtain for $y \in \Lambda$
$$
\int_\Lambda u_p(x) \sum_{z \in \mathbb{Z}^3} f(P_z(x)-y) \dx = \int_{[-\frac{3}{2},\frac{3}{2}]^3} u_p(x) f(x-y) \dx = \int_\Lambda u_p(x+y) f(x)\dx,
$$
where we used that $\Lambda \subset y + [-\frac{3}{2},\frac{3}{2}]^3$.
Thus, with the definition of $u_p$ \eqref{eq:up_def} 
we arrive at
\begin{align*}
& \int_{\Lambda^2}  \sum_{z \in \mathbb{Z}^3} f(P_z(x)-y) u_p(x) u_q(y)\dx\dy  = \int_\Lambda f(x) \int_\Lambda u_p(x+y) u_q(y) \dy \dx \\
&= \int_\Lambda f(x) \prod_{i=1}^3 \int_{-\frac{1}{2}}^{\frac{1}{2}} \left[\cos(p_ix_i)u_{p_i}(y_i) - \sqrt{2} \sin(p_ix_i)\sin(p_i(y_i+\frac{1}{2}))\right] u_{q_i}(y_i) \dy \dx.
\end{align*}
Observe that this formula trivially holds true for $p_i = 0$.
The second term vanishes if we integrate over $x$ since $f$ is radial. We conclude that
$$
\int_{\Lambda^2}  \sum_{z \in \mathbb{Z}^3} f(P_z(x)-y) u_p(x) u_q(y) \dx\dy = \delta_{p,q} \int_\Lambda f(x) \prod_{i=1}^3 \cos(p_ix_i) \dx = \delta_{p,q}  \int_{\R^3} f(x) e^{ip\cdot x} \dx 
$$
where in the second identity we used the fact that $f$ is radial and supported in $\Lambda$.
\end{proof}

\begin{proof}[Proof of Lemma \ref{lem_K_properties}] From \Cref{prop_symmetric_momentum} we immediately have
\begin{equation}
\widetilde{K}(x,y) = - \sum_{p \in \pi\mathbb{N}_0^3} n\widehat{\omega}_{\ell,\lambda}(p) u_p(x)u_p(y) \ , \quad 
K(x,y) = - \sum_{p \in \pi\mathbb{N}_0^3 \setminus \{0\}} n\widehat{\omega}_{\ell,\lambda}(p) u_p(x)u_p(y). \label{eq:K_momentum}
\end{equation}
Moreover, \Cref{lem_vep} yields
\[ \label{eq:K2_bound}
0\le  K_2:= K(x,y) -\widetilde{K}(x,y) = n\widehat{\omega}_{\ell,\lambda}(0)\le C \lambda^2 \frac{n}{\ell}.
\]
Let us emphasize that $K_2$ is a constant. Next, note that $|P_z(x)-y| \geq |x-y|$ for all $x,y \in \Lambda$. 
Together with the first bound in \eqref{eq:w-pointwise} and the finiteness of the sum this yields
\[ \label{eq:Kbound_mirror}
|\widetilde{K}(x,y)| \leq  \frac{Cn \mathds{1}_{|x-y| \leq \lambda} }{1 + \ell |x-y|}.
\]
From
\eqref{eq:Kbound_mirror} and \eqref{eq:K2_bound} we obtain
\begin{align}\label{eq:Kxy-point-wise}
|K(x,y)| \leq \frac{Cn \mathds{1}_{|x-y| \leq \lambda}}{1+\ell|x-y|} + C \lambda^2 \frac{n}{\ell}.
\end{align}
The uniform bound $\|K\|_\infty \leq C n$ follows immediately. For the $L^2-$norm we have
\begin{align*}
\|K\|_2^2
&\leq C n^2 \int_{\Lambda^2} \frac{\mathds{1}_{|x-y| \leq \lambda}\dx \dy}{\left(1 + \ell|x-y|\right)^2}  + C \lambda^4 \left(\frac{n}{\ell}\right)^2\\
&\leq C n^2 \int_0^\lambda \frac{r^2}{(1+\ell r)^2} \dr + C \lambda^4 \left(\frac{n}{\ell}\right)^2
\leq C \lambda \left(\frac{n}{\ell}\right)^2.
\end{align*}
Similarly, $
\| K_x \|_2^2 \leq C \lambda \left(n/\ell\right)^2 $ independently of $x$.
\end{proof}

We may think of $K$ as a modified scattering solution. Therefore, it is interesting to compare the scattering length $\ao$ with the one obtained by $K$.  
The following Lemma  quantifies  their difference. This will be used to extract the scattering length $\ao$ in computations in the upcoming sections.

\begin{lemma}(Boundary effects) \label{lem_scatlength}
The function
\[ \label{eq_scattering_convolution}
h(x)=\int_\Lambda V_{\ell}(x-y) \left(n + K(x,y)\right) \dy - 8\pi\mathfrak{a} \frac{n}{\ell},\quad x\in \Lambda 
\]
satisfies
$$\|h\|_1 \lesssim \frac{n}{\ell} \frac{\log(\ell)}{\ell} \quad \textrm{ and } \quad   \|h\|_{p} \lesssim \frac{n}{\ell} \ell^{-1/p},\quad \forall  p\in (1,\infty].  $$
Consequently,
\begin{align*}
\left| n V_{\ell}^{0000} + \int_{\Lambda^2} V_{\ell}(x-y) K(x,y) \dx\dy - 8\pi\mathfrak{a}\frac{n}{\ell} \right| \leq C n \ell^{-2} \log(\ell).
\end{align*}
\end{lemma}

%
%

\begin{proof} We choose $R>0$ such that ${\rm supp}(V) \subset B_{R}(0)$.  Using the uniform bound $|K(x,y)| \leq C n$ from Lemma \ref{lem_K_properties} and the obvious bound 
$$\int_\Lambda V_\ell (x-y) \dd y \le \frac{1}{\ell} \int_{\R^3} V $$
we have 
\begin{equation} \label{eq:h2_bound}
|h(x)| \leq C \frac{n}{\ell},\quad \forall x\in \Lambda.
\end{equation}
Moreover, this bound can be improved if $x$ lies well within the interior of $\Lambda$. Indeed, for $x\in \Lambda$ satisfying ${\rm dist}(x,\partial \Lambda) > R\ell^{-1}$,  we claim that
\begin{equation} \label{eq:h1_bound}
|h(x)| \leq C \frac{n}{\ell} \frac{1}{1+\ell {\rm d}(x, \partial \Lambda)} + C \frac{n}{\ell^2}.
\end{equation}
From \eqref{eq:h2_bound} and \eqref{eq:h1_bound}, it is straightforward to deduce the desired $L^p$ bounds of $h$. 

It remains to verify \eqref{eq:h1_bound}.
Using the definition \eqref{eq:K-tK} and $V_{\ell}(1-\omega_{\ell,\lambda}) = V_{\ell} f_{\ell}$, we obtain
\begin{align} \label{eq:h1_calculation}
\int_\Lambda V_\ell(x-y) (n + K(x,y))\dy &= \int_\Lambda n (V_{\ell} f_{\ell})(x-y) \dy 
 - \sum_{z \neq 0} \int_\Lambda V_{\ell}(x-y) n \omega_{\ell,\lambda}(P_z(x)-y) \dy \nn
 \\
 & \qquad + K_2 \int_\Lambda V_{\ell}(x-y) \dy
\end{align}
with $K_2=K-\widetilde{K}$. The last term of \eqref{eq:h1_calculation} is bounded easily by \eqref{eq:K2_bound}, 
\begin{equation*} \label{eq:h12_bound}
0\le K_2\int_\Lambda V_\ell(x-y) \dd y \leq C \frac{n}{\ell} \lambda^2 \|V_{\ell}\|_1 \leq C \frac{n}{\ell^2}.
\end{equation*}
For the first term on the right-hand side of \eqref{eq:h1_calculation},  from the assumption ${\rm dist}(x,\partial \Lambda) > R\ell^{-1}$ we have 
$(\Lambda - x) \cap B_{R\ell^{-1}}(0) = B_{R\ell^{-1}}(0)$, 
which can be used together with the fact ${\rm supp}(V_{\ell}) \subset B_{R\ell^{-1}}(0)$ to deduce that
$$
\int_\Lambda n (V_\ell f_\ell)(x-y)\dy = \int_{\Lambda-x} n (V_{\ell}f_{\ell})(y) \dy = \int_{\R^3} n (V_\ell f_\ell)(y) \dy = 8\pi\mathfrak{a}\frac{n}{\ell}.
$$
For the second term on the right-hand side of \eqref{eq:h1_calculation} we note that $P_z(x) \notin \Lambda$ for $x \in \Lambda$ and $z \neq 0$. In particular $|P_z(x)-y| \geq \dd(x,\partial \Lambda)$. Together with the first bound in \eqref{eq:w-pointwise} this implies 
\[ \label{eq_w_z_bound}
\omega_{\ell,\lambda}(P_z(x)-y) \leq \frac{C}{1 + \ell |P_z(x)-y|} \leq \frac{C}{1 + \ell {\rm d}(x,\partial \Lambda)} \quad \forall z\in\mathbb{Z}^3\setminus \{0\}. 
\]
Therefore, as only finitely many summands in this term are non-zero,
\begin{align*}
\sum_{z \neq 0} \int_\Lambda V_{\ell}(x-y) n \omega_{\ell,\lambda}(P_z(x)-y) \dy &\leq C \int_\Lambda n V_{\ell}(x-y) \frac{1}{1+\ell {\rm d}(x, \partial \Lambda)}\dy
\\
&\leq C \frac{n}{\ell} \frac{1}{1+\ell {\rm d}(x, \partial \Lambda)}.
\end{align*}
This concludes \eqref{eq:h1_bound} as well as the proof of Lemma \ref{lem_scatlength}. 
\end{proof}


\section{The First Quadratic Transformation}
\label{sec:firsttransform}

In this section we apply the first transformation $e^{\mathcal B_1}$ to the excitation Hamiltonian $\mathcal H$ in \eqref{eq_excitation_Hamil}, where 
\[ \label{eq_B1}
&\cB_1 = \frac{1}{2} \int_{\Lambda^2} K(x,y) a_x^* a_y^* \dd x \dd y - \hc\, 
\]
with $K$ given in (\ref{eq:K-tK}). The role of this transformation is to replace the quadratic term $Q_2$ in (\ref{eq_excitation_Hamil}) by $\widetilde{Q}_2$, defined in \eqref{eq:B1-comm-heu-intro}, which is less singular, and extracts a scalar contribution leading to the full leading order energy $4\pi \ao n^2 \ell^{-1}$. The main result of this section is the following lemma.

\begin{lemma} \label{lem_firsttransform}
Assume that $\lambda \left(\frac{n}{\ell}\right)^2 \leq 1$, that $2 R/\ell < \lambda < 1/4$, and that $\ell$ is large enough. Then we have
\begin{align}
e^{-\cB_1} \cH e^{\cB_1} 
	&= 4\pi \ao n^2 \ell^{-1} +  \sum_{p \in \pi \mathbb{N}_0^3 \backslash\{0\}} \frac{|n \widehat{\epsilon}_{\ell,\lambda}(p)|^2}{2p^2}  + \dG(-\Delta)  + Q_4  \nn \\
	& \quad + n \int_{\Lambda^2} V_\ell(x-y) (a_x^*a_x + a_x^*a_y) \dx\dy - 8\pi\mathfrak{a} n \ell^{-1} \mathcal N    \nn \\
	&\quad + \widetilde{Q}_2 + Q_3  + \mathcal{E}_1 \label{eq:def_H_cal}
\end{align}
on $\cF_+$, with 
\begin{align*}
Q_3 &= \sqrt{n} \int_{\Lambda^2} V_\ell(x-y) a_x^*a_y^*a_x \dx\dy +  \hc,
\\
\widetilde{Q}_2 &= \int_{\Lambda^2} \widetilde{Q}_2(x,y) a_x^* a_y^* \dx\dy +  \hc,\\
\widetilde{Q}_2(x,y) &= \frac{n}{2} \sum_{z \in \mathbb{Z}^3} \epsilon_{\ell,\lambda}(P_z(x)-y) \\
&\quad+  \frac{n}{2} \sum_{z \neq 0} \Big(V_\ell(\omega_\ell - 1)(P_z(x)-y) - V_{\ell}(x-y) \omega_{\ell,\lambda}(P_z(x)-y)\Big) 
\end{align*}
and 
\begin{align}
\pm \mathcal{E}_1 &\leq C (\delta+\ell^{-1} \lambda^2 + \varepsilon n^{-1}) Q_4 + C  \delta^{-1} \left( \frac{(\cN + 1)}{\ell} + \frac{(\cN+1)^2}{n\ell} +   \lambda \left(\frac{n}{\ell}\right)^3 \right) (\cN + 1) 
\nn \\
& \quad
+ C \lambda^\frac{1}{2} \left( \left(\frac{n}{\ell}\right)^2 + \frac{n}{\ell} \right) (\cN + 1)
+ C n^\frac{1}{2} \frac{(\cN+1)^\frac{3}{2}}{\ell}
\nn \\
&\quad + C \frac{n^{1/2}}{\ell^{5/6}} \left(\dG(-\Delta) + \frac{n^2\log(\ell)}{\ell^2}\right) 
+ C \varepsilon^{-1} \frac{n}{\ell} + C \frac{n^2}{\ell^2} \log \ell \label{eq:E_1}
\end{align}
for all $0 < \delta, \varepsilon \leq 1.$
\end{lemma}

The remainder of this section is dedicated to proving this lemma. Let us first make some  remarks.

\begin{itemize}

\item[{\bf 1.}] The condition $2 R/\ell < \lambda$ was already introduced in \eqref{eq:lambda-choice}, it ensures that $\omega_{\ell,\lambda} \equiv \omega_\ell$ on the support of $V_\ell$. The condition $\lambda \left(\frac{n}{\ell}\right)^2 \leq 1$ ensures that $\|K\|_2$ remains bounded, so that the first quadratic transform preserves powers of the particle number, see Lemma \ref{lem_excitation_conserved} below.  We will keep those constraints throughout the paper. Eventually, we will additionally ask $\lambda \ll 1$ so that only the first terms in the Duhamel expansion of $e^{\mathcal B_1}$ will contribute to the LHY order.

\item[{\bf 2.}] We will eventually choose $\delta = o(1)_{\rho \ao^3 \to 0}$. On the other hand, $\varepsilon$ will be independent of $\rho \ao^3$.

\item[{\bf 3.}] We  have chosen the transformation kernel of $\cB_1$ in such a way that 
\begin{align} \label{eq:B1-comm-heu}
[\dG(-\Delta) + Q_4,\cB_1] + Q_2 \approx \widetilde{Q}_2.
\end{align}
The renormalized quadratic term $\widetilde{Q}_2$ defined in \Cref{lem_firsttransform} consists of two parts, 
$$
\widetilde{Q}_2 =\widetilde{Q}_2^{(\epsilon)} +\widetilde{Q}_2^{(bc)} =  \int_{\Lambda^2} \widetilde{Q}_2^{(\epsilon)}(x,y) a_x^* a_y^* \dx\dy +\hc + \int_{\Lambda^2} \widetilde{Q}_2^{(bc)}(x,y) a_x^* a_y^* \dx\dy + \hc 
$$
with 
\begin{align}
\widetilde{Q}_2^{(\epsilon)}(x,y) &= \frac{n}{2} \sum_{z \in \mathbb{Z}^3} \epsilon_{\ell,\lambda}(P_z(x)-y),\label{eq:Q2-eps-def}\\
\widetilde{Q}_2^{(bc)}(x,y)&= \frac{n}{2} \sum_{z \neq 0} \bigg[  \Big(V_\ell(\omega_{\ell}-1)(P_z(x)-y)\Big) - V_{\ell}(x-y) \omega_{\ell,\lambda}(P_z(x)-y)\bigg] .\nn
\end{align}
The part $\widetilde{Q}_2^{(\epsilon)}$ comes from the cutoff we introduced in $\omega_{\ell,\lambda}$, see (\ref{eq:scattering_equation_truncated}), it is essentially the desired renormalized form of $Q_2$.  The additional part $\widetilde{Q}_2^{(bc)}$ is a  boundary effect that arises from the symmetrization of the kernel $K$. This is an error term, but for technical reasons we have to keep the boundary contribution $\widetilde{Q}_2^{(bc)}$ in $\widetilde{Q}_2$ in \Cref{lem_firsttransform}, and will eliminate it later after conjugating with the cubic transformation.  

\end{itemize}

We shall need the following standard estimate. 

\begin{lemma} \label{lem_excitation_conserved}
Assume that $\lambda \left(\frac{n}{\ell}\right)^2 \leq 1$ and that $2 R/\ell < \lambda < 1/4$. For all $k\in \mathbb{N}$ there is a constant $C_k > 0$ independent of $n,\ell$ and $\lambda$ such that on $\cF$ we have
$$
e^{-t\cB_1} (\cN+1)^k e^{t\cB_1} \leq C_k (\cN + 1)^k,\quad \forall t \in [-1,1].
$$
\end{lemma}
\begin{proof}
The proof that quadratic transformations preserve powers of the particle number is well known, see for instance \cite[Lemma 4]{NamTri-21} from which we have
\begin{equation*}
e^{-t\cB_1} (\cN+1)^k e^{t\cB_1} \leq C_k\left( 1 + \|K\|_2^2\right)^k (\cN + 1)^k.
\end{equation*}
Thanks to \Cref{lem_K_properties} and the assumption on $\lambda$, we have $\|K\|_2^2 \leq C \lambda (n/\ell)^2 \leq C$. 
\end{proof}

From \eqref{eq_excitation_Hamil} 
we have
\begin{align} \label{eq_H_firsttransform}
e^{-\cB_1} \cH e^{\cB_1} 
	&= \left( \frac{n^2}{2} V_{\ell}^{0000} + e^{-\cB_1} (\dG(-\Delta) + Q_2 + Q_4)e^{\cB_1} \right) + e^{-\cB_1} H_2^{(U)} e^{\cB_1} \nn \\ 
	&\qquad + e^{-\cB_1} (Q_1 + Q_3^{(U)}) e^{\cB_1} + e^{-\cB_1} \mathcal{E}^{(U)} e^{\cB_1} \nn\\
	&= {\rm (I)_1 + (II)_1 + (III)_1} +  e^{-\cB_1} \mathcal{E}^{(U)} e^{\cB_1}\,.
\end{align}
We will estimate ${\rm (I)_1}, {\rm (II)_1}$ and ${\rm (III)_1}$ in Sections \ref{sec:T1_I},  \ref{sec:T1_II} and \ref{sec:T1_III}, respectively. Finally, in Section \ref{sec:proof_lem_firsttransform}, we gather all previous estimates and complete the proof of Lemma \ref{lem_firsttransform}.

\subsection{Analysis of {\rm (I)\textsubscript{1}}}
	\label{sec:T1_I}
	
	In this subsection we estimate the term ${\rm (I)_1}$ appearing in (\ref{eq_H_firsttransform}).  

\begin{lemma} \label{prop_leading_firsttransform}
Assume that $\lambda \left(\frac{n}{\ell}\right)^2 \leq 1$, that $2 R/\ell < \lambda < 1/4$ and that $\ell$ is large enough. Then we have
\begin{align*}
{\rm (I)_1} &= \frac{n^2}{2}  V_{\ell}^{0000} + e^{-\cB_1} (\dG\left( -\Delta \right) + Q_2 + Q_4)e^{\cB_1} \\
& =  4\pi \mathfrak{a} n^2 \ell^{-1} + \sum_{p \in \pi \mathbb{N}_0^3 \backslash\{0\}} \frac{|n \widehat{\epsilon}_{\ell,\lambda}(p)|^2}{2p^2} +  \dG(-\Delta) + Q_4 + \widetilde{Q}_2 + \mathcal{E}_1^{(Q_2)} 
\end{align*}
on $\cF_+$, with
\begin{align*}
\pm \mathcal{E}_1^{(Q_2)} &\leq (\delta + C \ell^{-1}\lambda^2 ) Q_4 + C \delta^{-1}  \lambda \frac{n^2}{\ell^3}(\cN+1)^2 
+  C \lambda^\frac{1}{2} \frac{n^2}{\ell^2}(\cN+1) + C \frac{n^2}{\ell^2} \log \ell,
\end{align*}
for all $0 < \delta \leq 1$.
\end{lemma} 

To prove \Cref{prop_leading_firsttransform}, we use the  Duhamel-type identity
\begin{align} \label{eq_cancellation_trick}
&  e^{-\cB_1}  (\dG\left( -\Delta \right) + Q_2 + Q_4)e^{\cB_1}- \dG(-\Delta) - Q_4 \nn\\
 &= \int_0^1 e^{-t\cB_1} \Big([\dG\left( -\Delta \right) + Q_4,\cB_1]+Q_2\Big) e^{t\cB_1} \dt + \int_0^1 \int_t^1 e^{-s\cB_1} [Q_2,\cB_1] e^{s\cB_1} \ds \dt\,.
\end{align}
We deal with the  two terms on the right-hand side of \eqref{eq_cancellation_trick} by using Lemmata~\ref{lem:comm_K_Q4_T1} and \ref{prop_Q2} below.

\begin{lemma}\label{lem:comm_K_Q4_T1}
Assume that $\lambda \left(\frac{n}{\ell}\right)^2 \leq 1$, that $2 R/\ell < \lambda < 1/4$ and that $\ell$ is large enough. Then we have
\begin{equation}
\label{eq:comm_K_Q4_T1}
[\dG(-\Delta)+Q_4,\cB_1] + Q_2 = \widetilde{Q}_2 + \mathcal{E}
\end{equation}
on $\cF_+$, where $\widetilde{Q}_2$ is given in \Cref{lem_firsttransform} and $\mathcal{E}$ satisfies
\begin{align}
\pm \int_0^1 e^{-t\cB_1} \mathcal{E} e^{t\cB_1} \dt 
\leq (\delta + C \ell^{-1}\lambda^2 ) Q_4 + C \delta^{-1}  \lambda \frac{n^2}{\ell^3} (\cN+1)^2 + C \lambda^\frac{1}{2} \frac{n^2}{\ell^2} (\cN+1)\label{eq:E_1_comm_K_Q4_T1}
\end{align}
for all $0 < \delta \leq 1$.
Moreover, it holds that
\[ \label{eq_Q4_conservation}
e^{-t\cB_1} Q_4 e^{t\cB_1} \leq C \left(Q_4 + \frac{n^2}{\ell} + \lambda \frac{n^2}{\ell^3} (\cN+1)^2\right),\quad \forall t\in [-1,1].
\]
\end{lemma}

\begin{proof}
We use the momentum space representation from \eqref{eq:K_momentum} and \Cref{prop_symmetric_momentum} and calculate
\begin{align}\label{eq:kinetic_commutator_first}
[\dG(-\Delta) , \cB_1] &= - n\sum_{p \neq 0} p^2 \widehat{w}_{\ell,\lambda}(p) a_p^*a_p^* + \hc \nn
\\
&= n \sum_{z \in \mathbb{Z}^3} \int_{\Lambda^2} \Delta \omega_{\ell,\lambda}(P_z(x)-y) a_x^*a_y^* +\hc
\end{align}
Moreover,
\begin{equation} \label{eq:Q4_commutator_first_expanded} 
[Q_4,\cB_1]  = \frac{1}{2}\int_{\Lambda^2} V_{\ell}(x-y) \widetilde K(x,y) a_x^* a_y^* + \hc + \mathcal E 
\end{equation}
where 
\begin{align} \label{eq:Lemma4.5-eps}
\mathcal E &= 
\frac{1}{2} K_2 \int_{\Lambda^2} V_{\ell}(x-y) a_x^*a_y^*   + \int_{\Lambda^2} V_{\ell}(x-y) a_x^* a_y^* a^*(K_y) a_x +  \hc 
\end{align}
with $K_2=K-\widetilde K$ the constant given in \eqref{eq:K2_bound}. 
With the definition of $\widetilde K$ in \eqref{eq:tK-intro} and the scattering equation \eqref{eq:scattering_equation_truncated} we  compute  
\begin{align*}
&n \sum_{z \in \mathbb{Z}^3} \Delta \omega_{\ell,\lambda}(P_z(x)-y) + \frac{1}{2} V_{\ell}(x-y) \widetilde K(x,y)\\
&=\frac{n}{2} \sum_{z \in \mathbb{Z}^3} \Big( \epsilon_{\ell, \lambda} - \big( V_\ell (1-\omega_{\ell})  \big) (P_z(x)-y) -V_\ell(x-y) \omega_{\ell,\lambda} (P_z(x)-y) \Big),
\end{align*}
which is exactly equal to $\widetilde{Q}_2(x,y)$ defined in \Cref{lem_firsttransform} plus
$$\frac{n}{2} \Big(  - \big( V_\ell (1-\omega_{\ell}) \big) (P_z(x)-y)  - V_\ell(x-y) \omega_{\ell,\lambda} (P_z(x)-y) \Big)|_{z=0} = - \frac{n}{2} V_\ell(x-y). $$
Therefore, we deduce from \eqref{eq:kinetic_commutator_first} and \eqref{eq:Q4_commutator_first_expanded} that
$$
[\dG(-\Delta) + Q_4,\cB_1] = \int_{\Lambda^2} \Big(\widetilde{Q}_2(x,y) - \frac{n}{2} V_\ell (x-y)  \Big) a_x^* a_y^*  + \hc + \mathcal E =  \widetilde{Q}_2 - Q_2 + \mathcal E .
$$

We shall now prove the estimate \eqref{eq_Q4_conservation}. Recall the uniform bound $|K(x,y)| \leq C n$ from \Cref{lem_K_properties}. The Cauchy--Schwarz inequality yields
\begin{align*}
 & \pm \int_{\Lambda^2} V_{\ell}(x-y) K(x,y) a_x^*a_y^* +  \hc \\
&\quad\leq \int_{\Lambda^2} V_\ell(x-y) a_x^*a_y^*a_x a_y +  \int_{\Lambda^2} V_\ell(x-y) K(x,y)^2 \leq Q_4 + C n^2 \ell^{-1}.
\end{align*}
Moreover, the bound $\sup_{y \in \Lambda} \|K_y\|_2^2 \leq C \lambda n^2/\ell^2$ again from \Cref{lem_K_properties} gives
\begin{align} 
& \pm \int_{\Lambda^2} V_{\ell}(x-y)a_x^*a_y^*a^*(K_y)a_x +  \hc \nn \\
&\qquad \leq \int_{\Lambda^2} V_\ell(x-y) a_x^*a_y^*a_x a_y +  \int_{\Lambda^2} V_\ell(x-y) a_x^* a(K_y)a^*(K_y)a_x \nn \\
&\qquad   \leq Q_4 + C \lambda \frac{n^2}{\ell^3} \cN^2.\label{eq_estimate_Q4_error}
\end{align}
Therefore, from \eqref{eq:Q4_commutator_first_expanded} and \eqref{eq:Lemma4.5-eps} we have $$ \pm [\mathcal B_1,Q_4] \leq  2 Q_4 + C n^2 \ell^{-1} + C \lambda \frac{n^2}{\ell^3} \cN^2 .$$
From this and Lemma \ref{lem_excitation_conserved}, a standard Gr\"onwall argument gives
$$
e^{-t\cB_1} Q_4 e^{t\cB_1} \leq C \left(Q_4 + n^2 \ell^{-1} + \lambda \frac{n^2}{\ell^3} (\cN+1)^2 \right), \quad \forall t\in [-1,1].
$$

Now we have the tools to show the estimate \eqref{eq:E_1_comm_K_Q4_T1} with $\mathcal{E}$ given in \eqref{eq:Lemma4.5-eps} as the sum of the two terms. For the first one, we use the Cauchy--Schwarz inequality and the inequality  (\ref{eq:K2_bound}) to obtain
\begin{align*}
\pm & \frac{1}{2} K_2 \int_{\Lambda^2} V_{\ell} (x-y) a_x^*a_y^* +  \hc 
 \leq (\ell^{-1} \lambda^2) Q_4  + K_2^2  (\ell^{-1} \lambda^2)^{-1} \|V_\ell\|_1  \leq \ell^{-1} \lambda^2 Q_4  + C \lambda^2 \frac{n^2}{\ell^2}
\end{align*}
and by \eqref{eq_Q4_conservation}
\begin{align}
\pm & e^{-t\cB_1} \left( \frac{1}{2} K_2 \int_{\Lambda^2} V_{\ell} (x-y) a_x^*a_y^* +  \hc \right) e^{t\cB_1} 
 \leq C \ell^{-1}\lambda^2 Q_4 + C \lambda^3 \frac{n^2}{\ell^4} (\cN+1)^2 + C \lambda^2 \frac{n^2}{\ell^2}.
\end{align}
For the second term in $\mathcal{E}$ we write
\begin{align} \label{eq:Duhamel-B1-aa}
e^{-t\cB_1} a_x^*a_y^* e^{t\cB_1} = a_x^*a_y^* + t K(x,y) + \int_0^t e^{-s\cB_1} \big(a_x^*a(K_y) + a_y^*a(K_x)\big) e^{s\cB_1}\ds.
\end{align}
We find
\begin{align*}
& e^{-t\cB_1}\left( \int_{\Lambda^2} V_{\ell}(x-y) a_x^* a_y^* a^*(K_y) a_x +  \hc \right)  e^{t\cB_1}
\\
&= \int_{\Lambda^2} V_{\ell}(x-y) \left(a_x^*a_y^* + t K(x,y) + \int_0^t e^{-s\cB_1} \big(a_x^*a(K_y) + a_y^*a(K_x)\big) e^{s\cB_1}\ds \right) \times
\\
&  \hfillcell{ \times \, e^{-t\cB_1}  a^*(K_y) a_x e^{t\cB_1} + \hc}
\end{align*}
We may now bound each term with the Cauchy--Schwarz inequality while making use of \Cref{lem_excitation_conserved}. For $\delta > 0$ we obtain
\begin{align*}
& \pm \int_{\Lambda^2} V_{\ell}(x-y) a_x^*a_y^* e^{-t\cB_1}  a^*(K_y) a_x e^{t\cB_1}  + \hc
\leq \delta Q_4 + \delta^{-1} C \|V_{\ell}\|_1 \sup_y \|K_y\|_2^2 (\cN+1)^2,
\\
& \pm \int_{\Lambda^2} V_{\ell}(x-y) K(x,y)e^{-t\cB_1}  a^*(K_y) a_x e^{t\cB_1} + \hc
\leq C \|V_{\ell}\|_1 n \sup_y \|K_y\|_2 (\cN+1)
\intertext{and}
& \pm \int_{\Lambda^2} V_{\ell}(x-y) \int_0^t e^{-s\cB_1} \big(a_x^*a(K_y) + a_y^*a(K_x)\big) e^{s\cB_1}\ds \, e^{-t\cB_1}  a^*(K_y) a_x e^{t\cB_1} + \hc
\\
& \quad \leq C \|V_{\ell}\|_1 \sup_y \|K_y\|_2^2 (\cN+1)^2.
\end{align*}
Then, $\|K\|_\infty \leq C n$ and $\sup_y \|K_y\|_2 \leq C \lambda n^2 \ell^{-2}$ from \Cref{lem_K_properties} yield
\begin{align*}
&\pm e^{-t\cB_1} \left( \int_{\Lambda^2} V_{\ell}(x-y) a_x^* a_y^* a^*(K_y) a_x + \hc \right) e^{t\cB_1} 
\\
& \quad \leq \delta Q_4 + \delta^{-1} C \lambda \frac{n^2}{\ell^3} (\cN+1)^2 + C \lambda^\frac{1}{2} \frac{n^2}{\ell^2} (\cN+1) + C \lambda \frac{n^2}{\ell^3}(\cN+1)^2.
\end{align*}
Combining this with \eqref{eq_estimate_Q4_error}, and $\lambda,\ell^{-1},\delta \leq 1$ to simplify some error terms, yields the desired error bound \eqref{eq:E_1_comm_K_Q4_T1}.
\end{proof}

From \eqref{eq:comm_K_Q4_T1} and \eqref{eq_cancellation_trick} we obtain the identity
\begin{align} \label{eq_Q2_cancel}
&e^{-\cB_1} (\dG(-\Delta) + Q_2 +Q_4) e^{\cB_1} - \dG(-\Delta) - Q_4 \nn\\
&\quad=  \widetilde{Q}_2 + \int_0^1 \int_0^t e^{-s\cB_1} [\widetilde{Q}_2,\cB_1] e^{s\cB_1} \ds \dt \nn\\
&\qquad  +\int_0^1 \int_t^1 e^{-s\cB_1} [Q_2 ,\cB_1]e^{s\cB_1} \ds \dt + \int_0^1 e^{-t\cB_1} \mathcal{E} e^{t\cB_1} \dt ,
\end{align} 
where we again used the Duhamel formula. The last  term in \eqref{eq_Q2_cancel} is an error term that is controlled by \eqref{eq:E_1_comm_K_Q4_T1}.

As $Q_2,\widetilde{Q}_2$ and $\cB_1$ are all quadratic, there are constant contributions in $[Q_2,\cB_1]$ and $[\widetilde{Q}_2,\cB_1]$. They are extracted in the following lemma.

\begin{lemma} \label{prop_Q2} Assume that $\lambda \left(\frac{n}{\ell}\right)^2 \leq 1$, that $2 R/\ell < \lambda < 1/4$ and that $\ell$ is large enough. Then we have
\begin{align}
\frac{n^2}{2} V_{\ell}^{0000} + \frac{1}{2}[Q_2,\cB_1]  = 4\pi\mathfrak{a}n^2 \ell^{-1}  + \Xi_1^{(a)}, \label{eq:new-Xi-1}\\
\frac{1}{2}[\widetilde{Q}_2,\cB_1]=  \sum_{p \in \pi \mathbb{N}_0^3 \backslash\{0\}} \frac{|n \widehat{\epsilon}_{\ell,\lambda}(p)|^2}{2p^2}  + \Xi_1^{(b)}\label{eq:new-Xi-2}
\end{align}
on $\cF_+$ where both of the error terms $\Xi_1^{(a)}$ and $\Xi_1^{(b)}$ are bounded by
$$
\pm \Xi_1^{(a)}, \pm \Xi_1^{(b)} \le C \lambda^\frac{1}{2} \left(\frac{n}{\ell}\right)^2 \cN + C \frac{n^2}{\ell^2} \log \ell .
$$
Consequently, 
\begin{align} \label{eq:new-Xi}
&\frac{n^2}{2} V_{\ell}^{0000} +  \int_0^1 \int_t^1 e^{-s\cB_1} [Q_2,\cB_1] e^{s\cB_1} \ds \dt + \int_0^1 \int_0^t e^{-s\cB_1} [\widetilde{Q}_2,\cB_1] e^{s\cB_1} \ds \dt
\nn\\
& \qquad = 4\pi\mathfrak{a}n^2 \ell^{-1}  +  \sum_{p \in \pi \mathbb{N}_0^3 \backslash\{0\}} \frac{|n \widehat{\epsilon}_{\ell,\lambda}(p)|^2}{2p^2}  + \Xi_1
\end{align}
on $\cF_+$, with 
$$
\pm \Xi_1 \leq C \lambda^\frac{1}{2} \left(\frac{n}{\ell}\right)^2 (\cN + 1) + C \frac{n^2}{\ell^2} \log \ell .
$$
\end{lemma}


\begin{proof} Let us start with \eqref{eq:new-Xi-1}. 
A simple calculation shows that
\begin{equation}
[Q_2,\cB_1] = \int_{\Lambda^2} nV_{\ell}(x-y) K(x,y) 
 +  n\int_{\Lambda^2} V_{\ell}(x-y) a^*(K_x) a_y +  \hc \label{eq:Q2_comm_B1}
\end{equation}
For the first term on the right hand-side of (\ref{eq:Q2_comm_B1}), which will be multiplied by a factor $1/2 = \int_0^1 \int_0 ^t \dd s \dd t$, we have by Lemma \ref{lem_scatlength} 
$$
\left| \frac{n^2}{2} V_{\ell}^{0000}  + \frac{n}{2}  \int_{\Lambda^2} V_{\ell}(x-y) K(x,y) - 4\pi\mathfrak{a}n^2 \ell^{-1}\right| \leq C n^2 \ell^{-2} \log \ell.$$
For the second term in \eqref{eq:Q2_comm_B1} we use the Cauchy--Schwarz inequality and obtain
\begin{align*}
\pm n\int_{\Lambda^2} V_{\ell}(x-y) a^*(K_x) a_y +  \hc  &\leq C n \|V_{\ell}\|_1 \sup_y \|K_y\|_2 \, \cN
\leq C \lambda^\frac{1}{2} \left(\frac{n}{\ell}\right)^2 \cN.
\end{align*}
Thus \eqref{eq:new-Xi-1} holds.

The proof of \eqref{eq:new-Xi-2}  is more involved. To estimate 
\begin{align}\label{eq:Q2vep_comm_B1}
[\widetilde{Q}_2,\cB_1] = 2 \int_{\Lambda^2} \widetilde{Q}_2(x,y) K(x,y) 
+ 2 \int_{\Lambda^2} \widetilde{Q}_2(x,y) a^*(K_x) a_y +  \hc 
\end{align}
let us  decompose  
$$
\widetilde{Q}_2(x,y)=\widetilde{Q}_2^{(\epsilon)}(x,y)+\widetilde{Q}_2^{(bc)}(x,y)
$$
as in \eqref{eq:Q2-eps-def} and derive pointwise estimates for $\widetilde{Q}_2^{(\epsilon)}(x,y)$ and $\widetilde{Q}_2^{(bc)}(x,y)$. Note that in \eqref{eq:Q2-eps-def} both sums are finite (each sum contains at most $3^3=27$ non-zero summands). From \Cref{lem_vep} and the bound $|P_z(x)-y| \geq |x-y|$ for all $x,y\in \Lambda$ 
we have
\begin{align} \label{eq:Q2-eps-def-bound-1}
|\widetilde{Q}_2^{(\epsilon)}(x,y)| \dy 
	&\le C n \ell^{-1}   \lambda^{-3} \sum_{z\in \mathbb{Z}^3} \mathds{1}_{|P_z(x)-y|\leq \lambda} \nn\\
	& \le C n  \ell^{-1}   \lambda^{-3} \sup_{z\in \mathbb{Z}^3} \mathds{1}_{|P_z(x)-y|\leq \lambda}  
	\le C n \ell^{-1}\lambda^{-3}  \mathds{1}_{|x-y|\leq \lambda} \quad \forall x,y\in \Lambda.
\end{align}
Moreover, for all $0\ne z\in \mathbb{Z}^3$ we have
$$
V_\ell(P_z(x)-y) = V_\ell(P_z(x)-y) \mathds{1}_{{\rm d}(x,\partial \Lambda) \leq R\ell^{-1}} \le \frac{C  V_\ell(P_z(x)-y) }{1 + \ell {\rm d}(x,\partial \Lambda)}
$$
since $\text{supp}(V_{\ell}) \subset B_{R\ell^{-1}}(0)$. 
In combination with  \eqref{eq_w_z_bound} and
 $|P_z(x)-y| \geq |x-y|$, as well as   the non-increasing assumption on $|x| \mapsto V(x)$, we arrive at the bound
\begin{align}\label{eq_ptwbound_Q2^E}
|\widetilde{Q}_2^{(bc)}(x,y)| 
	&\leq C n \sum_{z \in \mathbb{Z}^3}  \frac{V_{\ell}(P_z(x)-y)}{1 + \ell {\rm d}(x,\partial \Lambda)} 
	\le C n \frac{V_{\ell}(x-y)}{1 + \ell {\rm d}(x,\partial \Lambda)}.
\end{align}

Now we are ready to bound the last term in \eqref{eq:Q2vep_comm_B1}. From \eqref{eq:Q2-eps-def-bound-1} and \eqref{eq_ptwbound_Q2^E}, we obtain the pointwise estimate
$$
|\widetilde{Q}_2(x,y)|
	\le C n \left( \ell^{-1}\lambda^{-3}\1_{|x-y| \le \lambda}+ V_\ell(x-y) \right),\quad \forall x,y\in \Lambda.
$$
Combining this with the uniform bound $\sup_{x \in \Lambda} \|K_x\|_2 \leq C \lambda^{1/2} n\ell^{-1}$ from \Cref{lem_K_properties} and the Cauchy--Schwarz inequality with $\delta = \ell / (\lambda^{1/2} n) $, we have 
\begin{align*}
&\pm  2 \int_{\Lambda^2} \widetilde{Q}_2(x,y) a^*(K_x)a_y +  \hc \\
&\leq \delta  \int_{\Lambda^2} |\widetilde{Q}_2(x,y)| a^*(K_x) a(K_x) + \delta^{-1}  \int_{\Lambda^2} |\widetilde{Q}_2(x,y)| a_y^* a_y  \\
&\leq C \delta n \ell^{-1} \int_{\Lambda}  a^*(K_x) a(K_x) +  C \delta^{-1} n \ell^{-1}  \int_{\Lambda} a_y^* a_y \\
&\leq C \delta n \ell^{-1} \sup_{x\in \Lambda} \|K_x\|_2^2 \, \cN +  C \delta^{-1}n \ell^{-1} \cN \le C \lambda^\frac{1}{2} \left(\frac{n}{\ell}\right)^2 \cN.
\end{align*}
It remains to consider the constant term in \eqref{eq:Q2vep_comm_B1}. The contribution involving $\widetilde{Q}_2^{(bc)}$ is negligible and can be estimated using (\ref{eq_ptwbound_Q2^E}) and the uniform bound $\|K\|_\infty \leq Cn$ from \Cref{lem_K_properties} as 
\begin{align*}
\left| 2 \int_{\Lambda^2} \widetilde{Q}_2^{(bc)}(x,y)K(x,y)\right| &\leq C n^2 \int_{\Lambda^2}  \frac{V_{\ell}(x-y)}{1 + \ell {\rm d}(x,\partial \Lambda)} \\
&\le C n^2 \ell^{-1} \int_\Lambda \frac{1}{1 + \ell {\rm d}(x,\partial \Lambda)}\leq C  \frac{n^2}{\ell^2}\log \ell.
\end{align*}
Finally we consider the important constant contribution involving $\widetilde{Q}_2^{(\epsilon)}$. Combining \eqref{eq:K-lemma-property} and \eqref{eq:Q2-eps-def} with \Cref{prop_symmetric_momentum} we obtain that
\begin{align} \label{eq:prop-to-momentun-2term}
 2 \int_{\Lambda^2} \widetilde{Q}_2^{(\epsilon)}(x,y)K(x,y) 
&=-n^2 \sum_{p\in \pi \mathbb{N}_0^3 \setminus \{0\}} \widehat{\epsilon}_{\ell,\lambda}(p) \widehat{\omega}_{\ell,\lambda}(p)
\end{align}
Covering $\mathbb{Z}^{3}\setminus \{0\}$ with $8$ rotations of $\mathbb{N}^{3}\setminus \{0\}$, we overcount the points in the hyperplanes $\{p_j = 0\}$ for $j \in \{1,2,3\}$ at most $7$ times. Recall from \eqref{eq:def_omega_ell_simple} that we can write
\begin{align}
\label{eq:intro_f_eps}
 \epsilon_{\ell,\lambda}(x) = \frac{\ao}{\ell} \lambda^{-3} f(\lambda^{-1}x)\quad \text{with}\quad f= \frac{\chi''}{|\cdot|} \in C_c^\infty(\R^3).
\end{align}
Using that both $\epsilon_{\ell,\lambda}$ and $\omega_{\ell,\lambda}$ are radial (therefore so is $f$), we obtain from (\ref{eq:prop-to-momentun-2term}) that
\begin{align*}
&\left| n^2 \sum_{p\in \pi \mathbb{N}_0^3 \setminus \{0\}} \widehat{\epsilon}_{\ell,\lambda}(p) \widehat{\omega}_{\ell,\lambda}(p) - \frac{n^2}{8} \sum_{p\in \pi \mathbb{Z}^3} \widehat{\epsilon}_{\ell,\lambda}(p) \widehat{\omega}_{\ell,\lambda}(p) \right| \\
& \leq C n^2 \sum_{q\in \pi \mathbb{Z}^2} |\widehat{\epsilon}_{\ell,\lambda}(q,0) \widehat{\omega}_{\ell,\lambda}(q,0)| \\
& \leq C \frac{n^2}{\ell^2} \lambda^2 \sum_{q\in \pi \mathbb{Z}^2} |\widehat{f}(\lambda q,0)| \leq C \left(\frac{n}{\ell}\right)^2,
\end{align*}
where we used that $\|\widehat{\omega}_{\ell,\lambda}\|_{\infty} \leq \|\omega_{\ell,\lambda}\|_1 \leq C \lambda^2 \ell^{-1}$, which follows from (\ref{eq:w-pointwise}).
Using that $\{\frac{1}{\sqrt{8}} e^{i p \cdot x }\}_{p \in \pi \mathbb{Z}^3} $ is an orthonormal basis of $L^2(2\Lambda)$, we obtain that
\begin{align*}
\left|  2 \int_{\Lambda^2} \widetilde{Q}_2^{(\epsilon)}(x,y)K(x,y)  + n^2 \int_{\R^3} \epsilon_{\ell,\lambda}(x) \omega_{\ell,\lambda}(x) \dx \right| \leq C \left(\frac{n}{\ell}\right)^2.
\end{align*}
We shall now suitably rewrite the second term. Recall the scattering equation \eqref{eq:scattering_equation_truncated},
\begin{align*}
 \frac{1}{2}\epsilon_{\ell, \lambda} = \Delta \left( \omega_{\ell,\lambda} - \omega_\ell \right)\quad \text{ on }\R^3,
\end{align*}
and note that $\omega_{\ell,\lambda} - \omega_\ell$ and $\Delta \omega_\ell$ have disjoint supports. Therefore,
$$
0 = 2 \braket{ \left( \omega_{\ell,\lambda} - \omega_\ell \right) , \Delta \omega_{\ell}}_{L^2(\R^3)} = 2 \braket{\Delta \left( \omega_{\ell,\lambda} - \omega_\ell \right) , \omega_{\ell}}_{L^2(\R^3)} = \braket{\epsilon_{\ell,\lambda},\omega_{\ell} }_{L^2(\R^3)},
$$
and hence
\begin{align*}
-n^2 \int_{\R^3} \epsilon_{\ell,\lambda}(x)\omega_{\ell,\lambda}(x) \dx &= -n^2 \langle \epsilon_{\ell, \lambda}, \omega_{\ell,\lambda} - \omega_\ell \rangle_ {L^2(\mathbb{R}^{3})}\\
 &= n^2\left\langle \epsilon_{\ell, \lambda}, \frac{1}{-2\Delta} \epsilon_{\ell, \lambda} \right\rangle_{L^2(\mathbb{R}^{3})} = n^2 (2\pi)^{-3} \int_{\R^3} \frac{|\widehat{\epsilon}_{\ell,\lambda}(p)|^2}{2p^2} \dd p. 
\end{align*}
To conclude, we claim that the integral can be replaced by a corresponding Riemann sum, namely
\begin{align} \label{eq:Riemann-error-0}
\left| n^2 (2\pi)^{-3} \int_{\R^3} \frac{|\widehat{\epsilon}_{\ell,\lambda}( p)|^2}{2 p^2} \dd p - \sum_{p\in \pi \mathbb{N}_0^3 \backslash\{0\}}  \frac{|n\widehat{\epsilon}_{\ell,\lambda}(p)|^2}{2p^2} \right| \leq C \left(\frac{n}{\ell}\right)^2 .
\end{align}
Indeed, using \eqref{eq:intro_f_eps} we see that
\eqref{eq:Riemann-error-0} is equivalent to 
\begin{align} \label{eq:Riemann-error-1}
\left| \int_{\R_{\geq 0}^3} g(\lambda z) \dz - \sum_{z \in \mathbb{N}^3 \setminus\{0\}} g(\lambda z) \right| \leq C \lambda^{-2}\quad \text{with}\quad g(z)= \frac{|\widehat f (\pi z)|^2}{|\pi z|^2}.
\end{align}
Since $f \in C_c^\infty(\R^3)$, it is straightforward to check that all second derivatives $D^\alpha g$, $|\alpha|=2$, are bounded as  
$$
|D_z^\alpha g(z) | \le C |z|^{-4},\quad \forall z\in \mathbb{R}^3\backslash\{0\}.
$$
Therefore, for every $z\in \mathbb{N}^3\backslash \{0\}$ and $\xi\in \Lambda+z$, we have the Taylor expansion
$$
g(\lambda \xi)= g(\lambda z) + \lambda (\xi-z) \cdot (\nabla g)(\lambda z) + \mathcal{O} (\lambda^{-2}) |z|^{-4}.
$$ 
Integrating over $\xi \in \Lambda+z$ and using
$
\int_{\Lambda+z} \xi \dd \xi  = z
$
we find that
$$
\int_{\Lambda+z} g(\lambda \xi)\dd \xi = g(\lambda z) + \mathcal{O} (\lambda^{-2}) |z|^{-4},\quad \forall z\in \mathbb{N}^3\backslash \{0\}.
$$
Summing up these bounds over $z\in \mathbb{N}^3\backslash \{0\}$ and combining with 
$$
\int_{\Lambda} g(\lambda z) \dd z \le \int_{\Lambda} \frac{C}{|\lambda z|^2} \dd z  \le C \lambda^{-2}
$$
we obtain \eqref{eq:Riemann-error-1}. Thus the proof of \eqref{eq:new-Xi-2} is complete.

The last statement \eqref{eq:new-Xi} in \Cref{prop_Q2} follows from \eqref{eq:new-Xi-1}, \eqref{eq:new-Xi-2}  and \Cref{lem_excitation_conserved}.
\end{proof}

\begin{proof} [Conclusion of \Cref{prop_leading_firsttransform}] 
Inserting \eqref{eq:new-Xi}  in \eqref{eq_Q2_cancel} we obtain the claim.
\end{proof}

\subsection{Analysis of {\rm (II)\textsubscript{1}}}
		\label{sec:T1_II}
Here we estimate the term $${\rm (II)_1} =  e^{-\cB_1} H_2^{(U)} e^{\cB_1}$$  appearing in (\ref{eq_H_firsttransform}). We recall from \eqref{ed2} that 
\begin{equation}\label{def:h2}
H_2^{(U)} = n\, \dG(V_{\ell}\ast u_0^2 + \widehat V_{\ell}- V_{\ell}^{0000}) - \left(\frac{1}{2} \int_{\Lambda^2} V_{\ell}(x-y) a_x^* a_y^* \dx \dy\, \cN +  \hc\right),
\end{equation}
where $nV_{\ell} \ast u_0^2$ is a multiplication operator and $\widehat V_\ell$ denotes the operator with integral kernel $V_{\ell}(x-y)$. 

\begin{lemma} \label{prop_H2_quadratictrafo}
Assume that $\lambda \left(\frac{n}{\ell}\right)^2 \leq 1$, that $2 R/\ell < \lambda < 1/4$ and that $\ell$ is large enough. Then we have
\begin{align*}
e^{-\cB_1} H_2^{(U)} e^{\cB_1} = \dG\left(nV_{\ell} \ast u_0^2 + n \widehat V_{\ell}- 8\pi\mathfrak{a} \frac{n}{\ell}\right) + \mathcal{E}_1^{(H_2^{(U)})}
\end{align*}
on $\cF_+$, with 
\begin{align*}
\pm \mathcal{E}_1^{(H_2^{(U)})} \leq \delta Q_4 + \delta^{-1} C \frac{(\cN + 1)^2}{\ell} + C \lambda^\frac{1}{2} \left( \left(\frac{n}{\ell}\right)^2 + \frac{n}{\ell}\right) (\cN+1)
\end{align*}
for all $0 < \delta \leq 1$.
\end{lemma}

To control the diagonal terms in \eqref{def:h2} we use the following lemma.
\begin{lemma} \label{lem_bounded_quadratictrafo}
Assume that $\lambda \left(\frac{n}{\ell}\right)^2 \leq 1$ and that $2 R/\ell < \lambda < 1/4$. Let $A: L^2(\Lambda) \to L^2(\Lambda)$ be a bounded, self-adjoint linear operator. Then on $\cF$ we have
$$
\pm \left(e^{-\cB_1} \dG(A) e^{\cB_1} - \dG(A)\right) \leq C \|A\|_{\rm op} \|K\|_2 \left(\cN + 1\right) \leq C \|A\|_{\rm op} \lambda^{1/2} \frac{n}{\ell} \left(\cN + 1\right).
$$
\end{lemma}
\begin{proof}
Applying the Duhamel formula yields
\begin{align*}
e^{-\cB_1}\dG(A)e^{\cB_1} - \dG(A) &= \int_0^1 e^{-t\cB_1} [\dG(A),\cB_1]e^{t\cB_1} \dt.
\end{align*}
Let us denote by $K$ the operator with kernel $K(x,y)$ so that ${\cB_1 = \frac{1}{2}\sum_{m \neq 0}a_m^* a^*(K u_m) - \hc}$ 
From the Cauchy--Schwarz inequality we find
\begin{align*}
\pm [ \dG(A), \cB_1] &= \pm \sum_{m \neq 0} a^*( A K u_m)a_m^* + \hc
\leq C \|A K\|_2 (\cN+1)
\leq C \|A\|_{\rm op} \|K\|_2 (\cN+1).
\end{align*}
The first inequality in \Cref{lem_bounded_quadratictrafo} then follows from \Cref{lem_excitation_conserved} and the second one from \Cref{lem_K_properties}.
\end{proof}

\begin{proof}[Proof of \Cref{prop_H2_quadratictrafo}]
For the last term of $H_2^{(U)}$ we find with the aid of \eqref{eq:Duhamel-B1-aa} 
\begin{align*}
&e^{-\cB_1} \left(\frac{1}{2} \int_{\Lambda^2} V_{\ell}(x-y) a_x^* a_y^*\cN +  \hc \right) e^{\cB_1} 
\\
&\quad=\frac{1}{2} \int_{\Lambda^2} V_{\ell}(x-y) a_x^* a_y^* e^{-\cB_1}\cN e^{\cB_1}  +  \hc
\\
&\qquad +  \int_{\Lambda^2} V_{\ell}(x-y) K(x,y) e^{-\cB_1}\cN e^{\cB_1}  
\\
&\qquad + \int_{\Lambda^2} V_{\ell}(x-y) \int_0^1 e^{-s\cB_1} a_x^*a(K_y) e^{s\cB_1} \ds\, e^{-\cB_1}\cN e^{\cB_1}  +  \hc
\end{align*}
The second term on the right-hand side is the main term.  The first term is controlled by $Q_4$ with the Cauchy--Schwarz inequality and \Cref{lem_excitation_conserved}, 
\begin{align*}
\pm \int_{\Lambda^2} V_{\ell}(x-y) a_x^*a_y^* e^{-\cB_1}\cN e^{\cB_1} +  \hc
&\leq \delta Q_4 + C \delta^{-1}  \int_{\Lambda^2}V_{\ell}(x-y) e^{-\cB_1} \cN^2 e^{\cB_1}
\\
&\leq \delta Q_4 + C \delta^{-1}  \frac{(\cN + 1)^2}{\ell}.
\end{align*}
For the third term we  find similarly
\begin{align*}
& \pm \int_{\Lambda^2} V_{\ell}(x-y) \int_0^1 e^{-s\cB_1} a_x^*a(K_y) e^{s\cB_1} \ds\, e^{-\cB_1}\cN e^{\cB_1} +  \hc
\\
&\quad\leq \int_{\Lambda^2} V_{\ell}(x-y) \int_0^1 \left(\|K\|_2^{-1}  e^{-s\cB_1} a_x^* a(K_y)a^*(K_y) a_x e^{s\cB_1} + \|K\|_2 e^{-\cB_1} \cN^2 e^{\cB_1}\right) \ds
\\
&\quad\leq C \lambda^\frac{1}{2} \frac{n}{\ell^2} (\cN+1)^2 \leq C \frac{(\cN+1)^2}{\ell},
\end{align*}
where we used $\lambda^{\frac 1 2}(n/\ell)\le 1$ in the last estimate. The above bounds show that
\begin{align}\label{eq:cE-H2-t}
e^{-\cB_1} H_2^{(U)} e^{\cB_1} &=  e^{-\cB_1} \dG \left(n V_{\ell}\ast u_0^2 + n \widehat V_{\ell} - \int_{\Lambda^2} V_\ell(x-y) (n+K(x,y) \right) e^{\cB_1} + \widetilde{\mathcal{E}}_1^{(H_2^{(U)})}
\end{align}
with
$$
\pm \widetilde{\mathcal{E}}_1^{(H_2^{(U)})} \leq C \lambda^\frac{1}{2} \left(\frac{n}{\ell}\right)^2 (\cN + 1) + \delta Q_4 + C\delta^{-1} \frac{(\cN+1)^2}{\ell}.
$$
From \Cref{lem_scatlength} we find 
$$
\pm \left( \int_{\Lambda^2}V_{\ell}(x-y)(n + K(x,y)) - 8\pi\mathfrak{a}\frac{n}{\ell} \right) \le C n \ell^{-2} \log(\ell) \le C \lambda^{\frac 1 2} \frac{n}{\ell},
$$
where we used $\lambda \geq C \ell^{-1} \geq \left( \log(\ell)/\ell \right)^2$ in the last estimate. Note that the operator in the bracket in the first line of \eqref{eq:cE-H2-t} is bounded by $C n \ell^{-1}$. We  now apply the previous estimate together with \Cref{lem_bounded_quadratictrafo} to the first line of \eqref{eq:cE-H2-t} and find the statement from \Cref{prop_H2_quadratictrafo}.
\end{proof}

\subsection{Analysis of {\rm (III)\textsubscript{1}}}
		\label{sec:T1_III}
Here we analyze 
$$
{\rm (III)_1} = e^{-\cB_1} (Q_1+Q_3^{(U)})e^{\cB_1}
$$ 
appearing in (\ref{eq_H_firsttransform}), where we recall that
\begin{align*}
Q_1 &= {n}^{3/2} \int_{\Lambda^2} V_{\ell}(x-y) a_x^* \dx\dy +  \hc\, ,
\\
Q_3^{(U)} &= \sqrt{(n-\cN+1)_+} \int_{\Lambda^2} V_{\ell}(x-y) a_x^* a_y^* a_x \dx\dy +  \hc, \\
Q_3 &= n^\frac{1}{2} \int_{\Lambda^2} V_{\ell}(x-y) a_x^*a_y^*a_x \dx\dy +  \hc
\end{align*}

\begin{lemma} \label{prop_Q3_quadratictrafo} Assume that $\lambda \left(\frac{n}{\ell}\right)^2 \leq 1$, that $2 R/\ell < \lambda < 1/4$ and that $\ell$ is large enough. Then we have
$$
{\rm (III)_1} = Q_3 + \mathcal{E}_1^{(Q_3)}
$$ 
on $\cF_+$ where the error term satisfies
\begin{align*}
\pm \mathcal{E}_1^{(Q_3)} 
	&\leq C \delta Q_4 + \delta^{-1} C \left(\frac{(\cN+1)^2}{n\ell} +   \lambda \left(\frac{n}{\ell}\right)^3 \right)  (\cN + 1) + C n^\frac{1}{2} \frac{(\cN+1)^\frac{3}{2}}{\ell}
\\
	&\quad + C \frac{n^{1/2}}{\ell^{5/6}} \left(\dG(-\Delta) + \frac{n^2\log(\ell)}{\ell^2}\right)
\end{align*}
for any $0 < \delta \leq 1$. 
\end{lemma}

\begin{proof}

We expand $e^{-\cB_1}a_x^*a_y^*e^{\cB_1}$ as in \eqref{eq:Duhamel-B1-aa} to obtain
\begin{align} \label{eq:Q3_commutator_B1}
e^{-\cB_1} Q_3^{(U)} e^{\cB_1} 
	&= \int_{\Lambda^2} V_{\ell}(x-y) a_x^*a_y^* e^{-\cB_1} a_x \sqrt{(n-\cN)_+} e^{\cB_1} +  \hc \nn
\\
&\quad+ \int_{\Lambda^2} V_{\ell}(x-y) K(x,y) e^{-\cB_1} a_x \sqrt{(n-\cN)_+} e^{\cB_1} +  \hc \nn
\\
&\quad  + \int_{\Lambda^2} V_{\ell}(x-y)  \int_0^1 e^{-t\cB_1} \big(a^*_x a(K_y) + a_y^*a(K_x)\big) e^{t\cB_1} \dt \times 
\\
& \hfillcell{\times e^{-\cB_1} a_x \sqrt{(n-\cN)_+} e^{\cB_1} +  \hc} \nn
\end{align}
The last term is an error term. In fact, using the Cauchy--Schwarz inequality (for an appropriate choice of $\eta>0$) we obtain
\begin{align*}
& \pm \int_{\Lambda^2} V_{\ell}(x-y)  \int_0^1 e^{-t\cB_1} a^*_x a(K_y) e^{t\cB_1} \dt \, e^{-\cB_1} a_x \sqrt{(n-\cN)_+} e^{\cB_1} +  \hc 
\\
&\quad\leq\eta  \int_{\Lambda^2}  V_{\ell}(x-y)  \int_0^1 \biggl[ e^{-t\cB_1} a^*_x a(K_y) e^{(t-1)\cB_1}  (\cN+1)^{-\frac{1}{2}} e^{-(t-1)\cB_1} a^*(K_y)a_x e^{t\cB_1}  
\\
&\qquad + \eta^{-1} C e^{-\cB_1} \sqrt{(n-\cN)_+}  a_x^*  (\cN+1)^\frac{1}{2} a_x \sqrt{(n-\cN)_+} e^{\cB_1} \biggl] \dt
\\
&\quad\leq \ell^{-1} \left(\eta \sup_y \|K_y\|_2^2 + \eta^{-1}C n  \right) (\cN+1)^\frac{3}{2}
\\
&\quad \leq C \lambda^\frac{1}{2} \frac{n^\frac{3}{2}}{\ell^2} (\cN+1)^\frac{3}{2}
\leq C \frac{n^\frac{1}{2}}{\ell} (\cN+1)^\frac{3}{2}.
\end{align*}

In the first two lines of \eqref{eq:Q3_commutator_B1} we may replace $\sqrt{(n-\cN)_+}$ by $\sqrt{n}$, using 
$$\sqrt{n} - \sqrt{(n-\cN)_+} \leq \frac{\cN}{ \sqrt{n}},$$ 
which follows from the elementary inequality 
$ 1 - \sqrt{(1-x)_+} \leq x$ for all $x\geq 0$. With the aid of the Cauchy--Schwarz inequality we obtain
\begin{align*}
& \pm \int_{\Lambda^2} V_{\ell}(x-y) K(x,y) e^{-\cB_1} a_x \bigg(\sqrt{n} - \sqrt{(n-\cN)_+} \bigg)e^{\cB_1} +  \hc
\\
&\leq n \int_{\Lambda^2} V_{\ell}(x-y) e^{-\cB_1} \bigg( n^{-\frac{1}{2}} a_x^* (\cN+1)^\frac{1}{2} a_x
  + n^\frac{1}{2} \left(\sqrt{n} - \sqrt{(n-\cN-1)_+}\right)^2 (\cN+1)^{-\frac{1}{2}}  \bigg)e^{\cB_1}
\\
&\leq C \frac{\sqrt{n}}{\ell}  (\cN+1)^\frac{3}{2}.
\end{align*}
In a similar way one shows that
\begin{align*}
& \pm \int_{\Lambda^2} V_{\ell}(x-y) a_x^*a_y^* e^{-\cB_1} a_x \left(\sqrt{n}-\sqrt{(n-\cN)_+}\right)e^{\cB_1} +  \hc
\\
& \leq \delta Q_4 
 + \delta^{-1} C \int_{\Lambda^2} V_{\ell}(x-y) e^{-\cB_1} (\sqrt{n} - \sqrt{(n-\cN+1)_+}) a_x^*a_x (\sqrt{n} - \sqrt{(n-\cN+1)_+}) e^{\cB_1}
\\
& \leq \delta Q_4 + \delta^{-1} C \ell^{-1} \frac{(\cN+1)^3}{n}
\end{align*}
for all $ \delta > 0$.
In particular, we have
\begin{align} \label{eq_Q3_intermediate}
e^{-\cB_1} Q_3^{(U)} e^{\cB_1} &= \sqrt{n} \int_{\Lambda^2}  V_{\ell}(x-y) a_x^*a_y^* e^{-\cB_1} a_x e^{\cB_1} +  \hc \nn
\\
&\quad + \sqrt{n}  \int_{\Lambda^2} V_{\ell}(x-y) K(x,y) e^{-\cB_1} a_x e^{\cB_1} +  \hc +  \widetilde{\mathcal{E}}_1^{(Q_3)}
\end{align}
with
\begin{align*}
\pm \widetilde{\mathcal{E}}_1^{(Q_3)} \leq C n^\frac{1}{2} \frac{(\cN+1)^\frac{3}{2}}{\ell} + \delta Q_4 + \delta^{-1} C \frac{(\cN+1)^3}{n \ell}.
\end{align*}
The second term of \eqref{eq_Q3_intermediate} cancels with $e^{-\cB_1} Q_1 e^{\cB_1}$ as we will see below, whereas the first line equals $Q_3$ up to negligible errors. Indeed,
\begin{align*}
& \sqrt{n} \int_{\Lambda^2}  V_{\ell}(x-y) a_x^*a_y^* e^{-\cB_1} a_x e^{\cB_1} +  \hc - Q_3
\\
&=  \sqrt{n} \int_{\Lambda^2} V_{\ell}(x-y) a_x^*a_y^* \int_0^1 e^{-t\cB_1} [a_x,\cB_1] e^{t\cB_1} \dt +  \hc
\\
&=  \sqrt{n} \int_{\Lambda^2}  V_{\ell}(x-y) a_x^*a_y^* \int_0^1 e^{-t\cB_1} a^*(K_x) e^{t\cB_1} \dt +  \hc
\end{align*}
so that with the Cauchy--Schwarz inequality for all $\delta > 0$
\begin{align}
& \pm \left( \int_{\Lambda^2} n^\frac{1}{2} V_{\ell}(x-y) a_x^*a_y^* e^{-\cB_1} a_x e^{\cB_1} +  \hc - Q_3 \right) \nn
\\
&  \leq\delta Q_4 + \delta^{-1} C n \sup_x \|K_x\|_2^2 \|V_{\ell}\|_1 (\cN + 1) \nn
\\
&  \leq\delta Q_4 +  \delta^{-1} C \lambda \left(\frac{n}{\ell}\right)^3 (\cN + 1). \label{eq:Q3_B1_a}
\end{align}

In order to control the second term in \eqref{eq_Q3_intermediate} we recall the definition of $h$ in \Cref{lem_scatlength} and use $ \braket{\xi, \int_\Lambda a_x  \xi} = 0$ for $\xi \in \cF_+$ as well as that $e^{-\cB_1}$ leaves $\cF_+$ invariant 
to obtain
\begin{align}
&n^\frac{1}{2}  \int_{\Lambda^2} V_{\ell}(x-y)  K(x,y) e^{-\cB_1} a_x e^{\cB_1} +  \hc + e^{-\cB_1} Q_1 e^{\cB_1} \nn
\\
&= n^\frac{1}{2}\int_\Lambda h(x) e^{-\cB_1} a_x e^{\cB_1} +  \hc \nn
\\
&= n^\frac{1}{2}\int_\Lambda h(x) a_x +  \hc + n^\frac{1}{2} \int_\Lambda h(x) \int_0^1 e^{-t\cB_1} a^*(K_x) e^{t\cB_1}\dt +  \hc \label{eq:Q3_B1_b}
\end{align}
on $\cF_+$. Both terms on the right hand side are small. Indeed, the Cauchy--Schwarz inequality and the estimates on $h$ in \Cref{lem_scatlength} yield
\begin{align}
& \pm n^\frac{1}{2} \int_\Lambda h(x) \int_0^1 e^{-t\cB_1} a^*(K_x) e^{t\cB_1} \dt +  \hc
	\leq C n^\frac{1}{2} \|h\|_1 \sup_x \|K_x\|_2 (\cN+1)^\frac{1} {2} \nn \\
	& \quad \leq C n^\frac{1}{2} \frac{n \log(\ell)}{\ell^2} \lambda^\frac{1}{2} \frac{n}{\ell} (\cN+1)^\frac{1}{2} 
	 \leq \lambda \left(\frac{n}{\ell}\right)^3 (\cN+1) + C \frac{n^2 \log(\ell)^2}{\ell^3}. \label{eq:Q3_B1_c}
\end{align}
Moreover, by the Cauchy--Schwarz inequality we have for all $\eta >0$
\begin{align}
\pm n^\frac{1}{2} \int_\Lambda h(x)a_x +  \hc &\leq C \eta n^{1/2} \int_\Lambda |h(x)| a_x^*a_x + C \eta^{-1} n^{1/2} \int_\Lambda |h(x)| \nn
\\
& \leq C \eta n^{1/2} \|h\|_{3/2} \dG(-\Delta) +C \eta^{-1} n^{1/2}  \|h\|_1 \nn \\
& \leq C \eta n^{1/2} \ell^{-2/3} \frac{n}{\ell} \dG(-\Delta) +C \eta^{-1} n^{1/2} \ell^{-1} \frac{n}{\ell} \log \ell \nn \\
& \leq C \frac{n^{1/2}}{\ell^{5/6}}\left( \dG(-\Delta) + \frac{n^2\log(\ell)}{\ell^2}\right), \label{eq:Q3_B1_d}
\end{align}
where we chose $\eta = \ell^{5/6}/n$ in the last step. For the second inequality we used that on $\cF_+$  
\[\label{eq_sobolev_trick}
\dG(\Phi) \leq C \|\Phi\|_{3/2} \dG(-\Delta),
\]
for $\Phi \in L^{3/2}(\Lambda)$. This follows from the H\"older and Sobolev inequalities as 
$$
\braket{f, \Phi f} = \int_\Lambda \Phi(x) |f(x)|^2 \leq C \|\Phi\|_{3/2} \|f\|_6^2 \leq C \|\Phi\|_{3/2} \left(\|f\|_2^2 + \|\nabla f\|_2^2\right), \; f \in H^1(\Lambda)
$$
and $\cN \leq \pi^{-2} \dG(-\Delta)$ on $\cF_+$.

Inserting \eqref{eq:Q3_B1_c} and \eqref{eq:Q3_B1_d} into \eqref{eq:Q3_B1_b}, and subsequently \eqref{eq:Q3_B1_a} and \eqref{eq:Q3_B1_b} into \eqref{eq_Q3_intermediate}, yields 
$$
e^{-\cB_1} \left( Q_1 + Q_3^{(U)} \right)e^{\cB_1} = Q_3 + \mathcal{E}_1^{(Q_3)}
$$ 
with 
\begin{align*}
\pm \mathcal{E}_1^{(Q_3)} &\leq C n^\frac{1}{2} \frac{(\cN+1)^\frac{3}{2}}{\ell} + \delta Q_4 + \delta^{-1} C \frac{(\cN+1)^3}{n\ell} + \delta Q_4 + \delta^{-1} C \lambda \left(\frac{n}{\ell}\right)^3 (\cN + 1)
\\
&\quad + \lambda \left(\frac{n}{\ell}\right)^3 (\cN+1) + C \frac{n^2 \log(\ell)^2}{\ell^3} + C \frac{n^{1/2}}{\ell^{5/6}} \left( \dG(-\Delta) + \frac{n^2\log(\ell)}{\ell^2}\right).
\end{align*}
We readily deduce \Cref{prop_Q3_quadratictrafo} via simplifications due to $\delta \leq 1$ and $\log(\ell) \ell^{-1} \lesssim \ell^{-5/6}$.
\end{proof}

\subsection{Proof of Lemma \ref{lem_firsttransform}}
	\label{sec:proof_lem_firsttransform}
	\begin{proof}
Inserting Lemmata \ref{prop_leading_firsttransform}, \ref{prop_H2_quadratictrafo} and \ref{prop_Q3_quadratictrafo} in \eqref{eq_H_firsttransform} we obtain
\begin{align*}
e^{-\cB_1} \cH e^{\cB_1} &= 4\pi\mathfrak{a} n^2 \ell^{-1} +  \sum_{p \in \pi \mathbb{N}_0^3 \backslash\{0\}} \frac{|n \widehat{\epsilon}_{\ell,\lambda}(p)|^2}{2p^2}
 + \dG(-\Delta) + Q_4 + \widetilde{Q}_2 + \mathcal{E}_1^{(Q_2)} 
\\
& \quad + \dG\left(nV_{\ell} \ast u_0^2 + n \widehat V_{\ell}- 8\pi\mathfrak{a} \frac{n}{\ell}\right) + \mathcal{E}_1^{(H_2^{(U)})}
 + Q_3 + \mathcal{E}_1^{(Q_3)} + e^{-\cB_1} \mathcal{E}^{(U)} e^{\cB_1}.
\end{align*}
We apply \Cref{lem_excitation_conserved} and \eqref{eq_Q4_conservation} to \eqref{eq:E_error_excitation} and find 
$$
\pm e^{-\cB_1} \mathcal{E}^{(U)} e^{\cB_1} \leq \varepsilon n^{-1} C Q_4  + \varepsilon^{-1} C \frac{n}{\ell} + C \lambda \frac{n}{\ell^3}(\cN+1)^2 + Cn^\frac{1}{2} \ell^{-1} (\cN+1)^\frac{3}{2}\,,
$$
were some terms simplified due to the condition $\vep \leq 1$.
Collecting all the error terms and using $\lambda n^2\ell^{-2} \leq 1$ and $\delta \leq 1$
yields \Cref{lem_firsttransform}.
\end{proof}


\section{The Cubic Transformation} \label{sec:cubic_trafo}

In this section we apply the cubic transformation $e^{\mathcal B_c}$ to $e^{\mathcal B_1} \mathcal H e^{-\mathcal B_1}$ with
\begin{equation} \label{eq:def_Bc}
\cB_c=\frac{ \theta_M(\mathcal N)}{\sqrt{n}}  \int_{\Lambda^2}  q_x^* a^*(K_x) q_x  \dd x \dd y -  \hc
\end{equation}
Recall that
\begin{equation*}
q_x = \int_{\Lambda} Q(x,y) a_y \dd y = a(Q_x).
\end{equation*}
where $Q=1-|u_0\rangle \langle u_0|$ and $Q(x,y)$ is its integral kernel. In particular 
$$q(f) = a(Qf)=a(f)-\langle f,u_0\rangle a(u_0),\quad \forall f\in L^2(\Lambda).$$
The use of $q_x$ instead of $a_x$ in the definition of $\mathcal B_c$ in (\ref{eq:def_Bc}) ensures that $\cB_c$ leaves $\cF_+$ invariant. Note that the commutation relations of these  operators with the usual creation and annihilation operators are given by
\begin{align*}
[q_x,a_y^*] = \delta_{x,y} - u_0(y) = \delta_{x,y} - 1,\quad \forall x,y\in \Lambda.
\end{align*}
In all normal ordered expressions, $q_x$ may be replaced by $a_x$ on $\cF_+$ since  
$q_x |\xi\rangle =  a_x |\xi\rangle $ for all $\xi \in \cF_+$. 

For $1 \leq M \leq n$ we define
\begin{align*}
\theta_M(\mathcal N) := \theta(\mathcal N/M),
\end{align*}
where $\theta\in C^{\infty}(\mathbb{R}_{\geq 0},[0,1])$  satisfies $\theta(x) = 1$ for $x \leq \frac{1}{2}$ and $\theta(x) = 0$ for $x \geq 1$. The cut-off $\theta_M$ in $\mathcal B_c$ ensures that $\mathcal{B}_c$ does not create too many excitations, thereby allowing us to close some Gr\"onwall estimates in the computation of $e^{-\mathcal B_c} (\dd\Gamma(-\Delta) + Q_3 + Q_4) e^{\mathcal B_c}$. Effectively, with this we only renormalize the $Q_3$ term on the sector with particle number  $\mathcal N \lesssim M$. In \Cref{lem:ana_H_loc}, $M$ is chosen of the order of $n^{1-68 \kappa}$, which is sufficient to compute the free energy of the system up to the second order for small $\kappa$. We will write $\theta_M$ instead of $\theta_M(\cN)$ in the following.

The main purpose of the transformation $e^{\mathcal B_c}$ is to remove the cubic term $Q_3$ in (\ref{eq:def_H_cal}). This also renormalizes the second line in (\ref{eq:def_H_cal}), which is essentially $(2\widehat{V}(0) - 8\pi \ao) n \ell^{-1}\mathcal N$, into $ 8\pi \ao n \ell^{-1}\mathcal N$. Eventually, we obtain the Bogoliubov Hamiltonian $\mathbb{H}_{\rm Bog}$ on $\cF_+$ defined in \eqref{eq:H_mom-intro}.

The following lemma is the main result of this section.

\begin{lemma} \label{lem_cubictrafo}
Assume that $\lambda \left(\frac{n}{\ell}\right)^2 \leq 1$, that $2 R/\ell < \lambda < 1/4$, that $\ell$ is large enough and that 
$$\sigma := \max \{n^{1/2} \ell^{-5/6}, n^{1/2}M \ell^{-3/2}, \lambda^{-1/2} n^{1/2} M^{1/2} \ell^{-1} \}  \leq 1.$$
Then we have
\begin{align} \label{eq:H_Tc}
 e^{-\cB_c} e^{-\cB_1} \cH e^{\cB_1} e^{\cB_c} 
 	&=  4\pi \ao n^2 \ell^{-1}  + \mathbb{H}_{\rm Bog} +Q_4 + \mathcal{E}_c
\end{align}
on $\cF_+$ where the error term satisfies
\begin{align*}
\pm \mathcal{E}_c 
&\leq C \sigma \left(Q_4 + \dG(-\Delta) + \frac{n}{\ell}(\cN+1) \right) + \frac{1}{2} Q_4
\\
& \quad+ C \delta \left(Q_4 + \frac{n}{\ell}(\cN+1) \right)
+ \delta^{-1} C \left[\frac{n}{\ell}\frac{\cN+1}{M} + \frac{(\cN + 1)}{\ell} + \frac{(\cN+1)^2}{n\ell} +   \lambda  \left(\frac{n}{\ell}\right)^3  \right]  (\cN + 1)
\\
&\quad
+ C \left[\lambda^\frac{1}{2} \left( \left(\frac{n}{\ell}\right)^2 + \frac{n}{\ell}\right) + \lambda^{-1/2} \frac{n^{3/2}}{\ell^2} + \frac{n}{\ell}\frac{(\cN+1)^{1/2}}{n^{1/2}} \right] (\cN+1)
+ C \left( \left(\frac{n}{\ell}\right)^2 \log \ell + \frac{n}{\ell} \right)
\end{align*}
on $\cF_+$, for all $0 < \delta \leq 1$.
\end{lemma}

To prove \Cref{lem_cubictrafo}, we start with Lemma \ref{lem_firsttransform} and  the quadratic form identity 
\begin{align}
 e^{-\cB_c} &e^{-\cB_1} \cH e^{\cB_1} e^{\cB_c} -  4\pi\mathfrak{a} n^2 \ell^{-1} - \frac{1}{2} \sum_{p \neq 0}
\frac{|n\widehat{\epsilon}_{\ell,\lambda}(p)|^2}{2p^2}	=  e^{-\cB_c} \Big(\dG(-\Delta) + Q_4 + Q_3\Big) e^{\cB_c}  \nn\\
	&\qquad +e^{-\cB_c} \dG\left(nV_{\ell} \ast u_0^2 + n \widehat V_{\ell} - 8\pi\mathfrak{a} \frac{n}{\ell}\right) e^{\cB_c}  + e^{-\cB_c} \widetilde{Q}_2 e^{\cB_c}  + e^{-\cB_c} \mathcal{E}_1  e^{\cB_c} \label{eq_HNcubic_raw}
\end{align}
on $\cF_+$. Recall that  $V_{\ell} \ast u_0^2$ is a multiplication operator and that $\widehat V_\ell$ is the operator with integral kernel $V_{\ell}(x-y)$. Using the Duhamel formula, we can expand the term on the third line above as
\begin{align*}
 e^{-\cB_c}  \Big(\dG(-\Delta) + Q_4 + Q_3\Big) e^{\cB_c} 
 	&= \dG(-\Delta) + Q_4 + \int_0^1 \int_t^1 e^{-s\cB_c}[Q_3,\cB_c]e^{s\cB_c}\ds\dt \nn \\
&\quad  + \int_0^1 e^{-t\cB_c} \Big(Q_3 + [\dG(-\Delta)+Q_4,\cB_c]\Big)e^{t\cB_c} \dt.
\end{align*}
Plugging the above equation into (\ref{eq_HNcubic_raw}), we obtain
\begin{align}
 & e^{-\cB_c} e^{-\cB_1} \cH e^{\cB_1} e^{\cB_c} 
	-  4\pi\mathfrak{a} n^2 \ell^{-1} - \frac{1}{2} \sum_{p \neq 0} \frac{|n\widehat{\epsilon}_{\ell,\lambda}(p)|^2}{2p^2}  - \dd\Gamma(-\Delta) - Q_4 
\nn \\
	&=  	\Big\{ e^{-\cB_c} \dG\left(nV_{\ell} \ast u_0^2 + n \widehat V_{\ell} - 8\pi\mathfrak{a} \frac{n}{\ell}\right) e^{\cB_c} + \int_0^1 \int_t^1 e^{-s\cB_c}[Q_3,\cB_c]e^{s\cB_c}\ds\dt \Big\} \nn \\
	& \quad + \Big\{ e^{-\cB_c} \widetilde{Q}_2 e^{\cB_c} \Big\} + \Big\{ \int_0^1 e^{-t\cB_c} \Big(Q_3 + [\dG(-\Delta)+Q_4,\cB_c]\Big)e^{t\cB_c} \dt + e^{-\cB_c} \mathcal{E}_1  e^{\cB_c} \Big\} \nn \\
	& = {\rm (I)}_c + {\rm (II)}_c + {\rm (III)}_c .	\label{eq_HNcubic}
\end{align}

In Section \ref{sec:Tc_kin_Q4} we compute the action of $e^{\mathcal B_c}$  on $\dd\Gamma(-\Delta)$ and $Q_4$. Then we show that ${\rm (I)}_c$ is essentially $8\pi\ao n\ell^{-1} \mathcal N$, while ${\rm (II)}_c$ gives the pairing term involving $(a_p^* a_p^*+\hc)$ in the Bogoliubov Hamiltonian in Section \ref{sec:Q1-Bc} and Section \ref{sec:Q2-Bc}, respectively. Finally, in Section \ref{sec:proof_lem_cubictrafo}, we estimate the error term  ${\rm (III)}_c$ and conclude the proof of Lemma \ref{lem_cubictrafo}.

We end this subsection with an estimate of the action of $e^{\mathcal B_c}$ on $\mathcal N^k$ analogous to \Cref{lem_excitation_conserved}.

\begin{lemma} \label{lem_number_preserved_cubic}
Assume that $\lambda \left(\frac{n}{\ell}\right)^2 \leq 1$ and that $2 R/\ell < \lambda < 1/4$. For all $k \in \mathbb{N}$ there is a constant $C_k > 0$ such that on $\cF_+$
$$
e^{-t\cB_c} \cN^k e^{t\cB_c} \leq C_k (\cN + 1)^k,\quad \forall t\in [-1,1]. 
$$
\end{lemma}

\begin{proof} 
From the Duhamel formula we have 
 \begin{align*}
e^{-t\cB_c} (\cN+1)^k e^{t\cB_c} - (\cN+1)^k = \int_0^t e^{-s\cB_c} [(\cN+1)^k,\cB_c] e^{s\cB_c} \ds 
\end{align*}
Let us estimate the commutator. Recall that in normal ordered expressions we can replace $q_x$ by $a_x$  on $\cF_+$. Therefore, using the Cauchy--Schwarz inequality, we have
\begin{align*}
&[(\cN+1)^k,\cB_c] = n^{- \frac 1 2}  \int_{\Lambda^2} [(\cN+1)^k, \theta_M K(x,y) q_x^*a_y^* q_x] +  \hc\\
&= n^{- \frac 1 2} \int_{\Lambda^2}  ((\cN+1)^k- \cN^k) \theta_M  K(x,y) a_x^*a_y^* a_x +  \hc\\
&= n^{- \frac 1 2} \int_{\Lambda^2}  ((\cN+1)^k- \cN^k) (\cN+1)^{- \frac {k-1} 2} \theta_M  a_x^*a^*(K_x) a_x (\cN+2)^{\frac {k-1} 2} +  \hc\\
&\le n^{-1} ((\cN+1)^k- \cN^k)(\cN+1)^{-\frac {k-1} 2}  \theta_M  \sup_x \|K_x\|_2^2 \, \cN^2 \theta_M  (\cN+1)^{-\frac {k-1} 2} ((\cN+1)^k- \cN^k)\\
&\quad + (\cN+2)^{\frac {k-1} 2} \cN (\cN+2)^{\frac {k-1} 2}\\
&\le C_k n^{-1} (\cN+1)^{k+1}\theta_M + C  (\cN+1)^k  \le C_k(\cN+1)^k.
\end{align*}
Here we used that $\sup_x \|K_x\|_2^2 \le \lambda n^2\ell^2 \leq 1$ and that $n^{-1} \theta_M (\mathcal N+1) \leq n^{-1} (M+1) \leq 2$. Therefore, we obtain
\begin{align*}
e^{-t\cB_c} (\cN+1)^k e^{t\cB_c} - (\cN+1)^k \leq C_k \int_0^t e^{-s\cB_c} (\cN+1)^k e^{s\cB_c} \ds.
\end{align*}
Applying the Grönwall lemma concludes the proof.
\end{proof}

\subsection{Actions on $\dd\Gamma(-\Delta)$ and $Q_4$}
	\label{sec:Tc_kin_Q4}
	
In this section, we estimate the actions of $e^{\mathcal B_c}$ on $\dd\Gamma(-\Delta)$ and $Q_4$. 

\begin{lemma} \label{lem_kinetic_preserved}
Assume that $\lambda \left(\frac{n}{\ell}\right)^2 \leq 1$, that $2 R/\ell < \lambda < 1/4$ and that $\ell$ is large enough. Let $\sigma\le 1$ be as in \Cref{lem_cubictrafo}. For all $t \in [-1,1]$ we have on $\cF_+$
\begin{align*} 
e^{- t \cB_c} \dG(-\Delta) e^{t\cB_c} &\leq  C \left(\dG(-\Delta) + Q_4 + \frac{n}{\ell}(\cN+1)\right),
\\
e^{-t\cB_c} Q_4 e^{t\cB_c} &\leq C \left(Q_4 + \frac{n}{\ell}(\cN + 1) + \sigma \dG(-\Delta)\right).
\end{align*}
\end{lemma}

As an intermediate step, we need to compute accurately the commutators $[\dd\Gamma(-\Delta), \mathcal B_c]$ and  $[Q_4, \mathcal B_c]$, which is done in Lemmata \ref{lem_kinetic_commutator_cubic} and \ref{lem_Q4_commutator_cubic}.
This will further be useful in order to bound the term ${\rm (II)}_c$ in the proof of Lemma \ref{lem_cubictrafo} in Section \ref{sec:proof_lem_cubictrafo}.

\begin{lemma} \label{lem_kinetic_commutator_cubic}
Assume that $\lambda \left(\frac{n}{\ell}\right)^2 \leq 1$, that $2 R/\ell < \lambda < 1/4$ and that $\ell$ is large enough. Then, on $\cF_+$ we have
$$
[\dG(-\Delta),\cB_c] = {\theta_M} n^\frac{1}{2} \int_{\Lambda^2} \left(V_{\ell} (\omega_{\ell}-1)\right)(x-y) a_x^*a_y^*a_x \dx\dy +  \hc + \mathcal{E}_c^{(\dG(-\Delta))}
$$
with
\begin{align*}
\pm \mathcal{E}_c^{(\dG(-\Delta))} & \leq C n^\frac{1}{2}\ell^{-\frac{5}{6}} \left( Q_4 + \dG(-\Delta)\right) +  C \lambda ^{-\frac{1}{2}} \frac{n^\frac{1}{2}M^\frac{1}{2}}{\ell}  \dG(-\Delta).
\end{align*}
\end{lemma}

\begin{lemma} \label{lem_Q4_commutator_cubic}
Assume that $\lambda \left(\frac{n}{\ell}\right)^2 \leq 1$, that $2 R/\ell < \lambda < 1/4$ and that $\ell$ is large enough. Then, on $\cF_+$ we have
\begin{equation*}
[Q_4,\cB_c] = - {\theta_M} n^\frac{1}{2} \int_{\Lambda^2}(V_{\ell} \omega_{\ell})(x-y) a_x^*a_y^*a_x \dx\dy +  \hc + \mathcal{E}_c^{(Q_4)}
\end{equation*}
with
\begin{align*}
\pm \mathcal{E}_c^{(Q_4)} & \leq  C \left(n^\frac{1}{2}\ell^{-\frac{5}{6}} + \frac{n^\frac{1}{2} M}{\ell^\frac{3}{2}} \right) \left( Q_4 +  \dG(-\Delta) \right).
\end{align*}
\end{lemma}

We have defined $\sigma$ in \Cref{lem_cubictrafo} in such a way that the error terms in Lemmata \ref{lem_kinetic_commutator_cubic} and \ref{lem_Q4_commutator_cubic} are bounded by
\begin{align} \label{rem:kin_Q4_Tc}
\pm \left(\mathcal{E}_c^{(\dG(-\Delta))} + \mathcal{E}_c^{(Q_4)}\right) &\leq C\sigma (Q_4 + \dG(-\Delta)).
\end{align}


\begin{proof}[Proof of \Cref{lem_kinetic_commutator_cubic}]

Clearly $[\dG(-\Delta),\theta_M]=0$. 
From \eqref{eq:K-lemma-property} we find
$$
\cB_c = - n \frac{\theta_M}{\sqrt{n}} \sum_{m,p,q \neq 0} \widehat{\omega}_{\ell,\lambda}(p) a_m^*a_p^*a_q \int_\Lambda u_m(x)u_p(x)u_q(x) - \hc
$$
We compute
\begin{align*}
&[\dG(-\Delta), \cB_c] = -n^{1/2} \theta_M \sum_{m,p,q \neq 0} \widehat{\omega}_{\ell,\lambda}(p) a_m^*a_p^*a_q (m^2+p^2-q^2) \int_\Lambda u_m(x)u_p(x)u_q(x) + \hc
\\
&\quad = - 2 n^{1/2} \theta_M \sum_{m,p,q \neq 0} p^2 \widehat{\omega}_{\ell,\lambda}(p) a_m^*a_p^*a_q  \int_\Lambda u_m(x)u_p(x)u_q(x) + \hc 
\\
&\qquad - 2  n^{1/2} \theta_M \sum_{m,p,q \neq 0} \widehat{\omega}_{\ell,\lambda}(p) a_m^*a_p^*a_q  \int_\Lambda (\nabla u_m)(x) u_p(x) (\nabla u_q) (x) + \hc
\\
&\quad = 2 \frac{\theta_M}{\sqrt{n}} \int_{\Lambda^2} (-\Delta_2 K)(x,y) q_x^*a_y^*q_x + \hc - 2 \frac{\theta_M}{n^{1/2}} \int_{\Lambda^2} (\nabla_1 K)(x,y) q_x^*a_y^* \nabla_x q_x + \hc,
\end{align*}
where for the last equality we used \eqref{eq:K-lemma-property}. 
%
%
%
Moreover,
$$
-\Delta_2 K(x,y) =
 n\sum_{z \in \mathbb{Z}^3} (\Delta \omega_{\ell})(P_z(x)-y) + \frac{n}{2}\sum_{z \in \mathbb{Z}^3} \epsilon_{\ell,\lambda }(P_z(x)-y).
$$
On $\mathcal F_+$, we can replace $q_x$ by $a_x$, and we have
\begin{align*}
 [\dG(-\Delta), \cB_c] &= \theta_M  n^\frac{1}{2} \int_{\Lambda^2}  (2\Delta \omega_{\ell})(x-y)a_x^*a_y^*a_x +  \hc
\\
&\quad + \theta_M  n^\frac{1}{2} \sum_{z \neq 0} \int_{\Lambda^2}  (2\Delta \omega_{\ell})(P_z(x)-y)a_x^*a_y^*a_x +  \hc
\\
& \quad + \theta_M n^\frac{1}{2} \sum_{z \in \mathbb{Z}^3} \int_{\Lambda^2}  \epsilon_{\ell,\lambda }(P_z(x)-y) a_x^*a_y^*a_x +  \hc
\\
&\quad  - \left(2n^{-\frac{1}{2}}\theta_M \int_{\Lambda^2} \nabla_x K(x,y) \cdot a_x^*a_y^* \nabla_x a_x +  \hc\right)
\\
&=: \theta_M  n^\frac{1}{2}\int_{\Lambda^2}  (2\Delta \omega_{\ell})(x-y)a_x^*a_y^*a_x +  \hc + \sum_{i=1}^3 G_i.
\end{align*}

The scattering equation
$2\Delta \omega_{\ell} + V_{\ell} - V_{\ell}\omega_{\ell} = 0$
now gives the correct main term in \Cref{lem_kinetic_commutator_cubic} and we conclude the proof by estimating the error $\mathcal{E}_c^{(\dG(-\Delta))} := \sum_{i=1}^3 G_i$ term by term.

For $G_1$, we use $|1-\omega_{\ell}| \leq 1$ and that for $z\neq 0$ we have
$$V_{\ell}(P_z(x)-y) \leq V_{\ell}(x-y) \mathds{1}_{{\rm d}(x,\partial \Lambda) \leq R \ell^{-1}} $$
since $|x-y| \leq |P_z(x)-y|$ and $V$ is decreasing and  supported on a ball of radius $R$. Together with $\theta_M \leq 1$ and the Cauchy--Schwarz inequality we obtain on $\mathcal F_+$
\begin{align*}
\pm G_1 &= \pm n^\frac{1}{2} \theta_M \sum_{z \neq 0} \int_{\Lambda^2} \big(V_{\ell}(\omega_{\ell}-1)\big)(P_z(x)-y) a_x^*a_y^*a_x +  \hc
\\
&  \leq \delta Q_4 + C \delta^{-1}n \int_{\Lambda^2} V_{\ell}(x-y) \mathds{1}_{{\rm d}(x,\partial \Lambda) \leq R \ell^{-1}} a_x^*a_x
\\
& \leq \delta Q_4 + C\delta^{-1} {n}  \ell^{-\frac{5}{3}} \dG(-\Delta)
\end{align*}
for all $\delta > 0$.
In the last inequality we applied \eqref{eq_sobolev_trick} to $\Phi = \mathds{1}_{{\rm d}(x,\partial \Lambda) \leq R \ell^{-1}}$ with $\|\Phi\|_{3/2} \leq C \ell^{-2/3}$. The choice $\delta = n^{1/2} \ell^{-5/6}$ gives the first error term in \Cref{lem_kinetic_commutator_cubic}.

In order to bound $G_2$, we again use \eqref{eq_sobolev_trick}, this time applied to $\Phi(y) = \epsilon_{\ell,\lambda }(P_z(x)-y)$ for fixed $x$. We may estimate the $L^p-$norms of $\epsilon_{\ell,\lambda }$ using 
$$|\epsilon_{\ell,\lambda }(P_z(x)-y)| \leq C \ell^{-1} \lambda ^{-3} \mathds{1}_{|x-y|\leq  \lambda },$$
which follows from \Cref{lem_vep} and $|x-y| \leq |P_z(x)-y|$. The Cauchy--Schwarz inequality, $\theta_M (\cN+1)^k \theta_M \leq C M^k$ for $k \geq 0$, $\cN \leq \pi^{-2} \dG(-\Delta)$ on $\cF_+$ and an appropriate choice for $\delta_M > 0$ yield on $\mathcal F_+$
\begin{align*}
\pm G_2 & 
\leq  n^\frac{1}{2} \sum_{z \in \mathbb{Z}^3} \Big( \int_{\Lambda^2} \delta_M |\epsilon_{\ell,\lambda }(P_z(x)-y)| \theta_M a_x^*a_y^* a_y a_x \theta_M
+ \delta_M^{-1} |\epsilon_{\ell,\lambda }(P_z(x)-y)| a_x^* a_x \Big) 
\\
&\leq
C n^\frac{1}{2} \left(\delta_M \|\epsilon_{\ell,\lambda }\|_{3/2} \theta_M (\cN+1)  \dG(-\Delta) \theta_M +  \delta_M^{-1} \|\epsilon_{\ell,\lambda }\|_1 \cN \right)
\\
&\leq C n^\frac{1}{2} \ell^{-1} \lambda ^{-3} \left(\delta_M M \lambda^2 + \delta_M^{-1} \lambda ^3  \right) \dG(-\Delta)
\\
& \leq C \frac{n^\frac{1}{2} M^\frac{1}{2}}{\ell} \lambda ^{-\frac{1}{2}} \dG(-\Delta).
\end{align*}
For $G_3$ we readily check that
$$|\nabla_x K(x,y)| \leq C \frac{n}{\ell} \frac{1}{|x-y|^2}$$
from \Cref{lem_vep} and \eqref{eq:K-tK}, so that on $\mathcal F_+$
\begin{align*}
\pm \frac{1}{2} G_3 
& \leq n^{-\frac{1}{2}}  \left(\delta_M \int_{\Lambda^2} \theta_M a_x^*a_y^*a_ya_x \theta_M |\nabla_x K(x,y)| + \delta_M^{-1} \int_{\Lambda^2} \nabla_x a_x^* \nabla_x a_x |\nabla_x K(x,y)| \right) 
\\
& \leq C n^{-\frac{1}{2}} \left(\delta_M \int_{\Lambda^2} \theta_M a_x^*a_y^*a_ya_x \theta_M \frac{n}{\ell |x-y|^2} + \delta_M^{-1} \int_{\Lambda} \nabla_x a_x^* \nabla_x a_x \frac{n}{\ell} \right).
\end{align*} 
We now use the Hardy inequality on $\Lambda$
\begin{equation} \label{eq_Hardy}
\int_{\Lambda} \frac{|f(x,y)|^2}{|x-y|^2} \dx \leq C \|f(\cdot,y)\|_{H^1(\Lambda)}^2, \, \forall f \in H^1(\Lambda^2), y \in \Lambda .
\end{equation}
Consequently, together with $\cN \leq \pi^{-2} \dG(-\Delta)$ on $\cF_+$, we have
 $$
\int_{\Lambda^2} a_x^*a_y^* a_xa_y \frac{1}{|x-y|^2} \leq C (\cN-1)\dG(-\Delta).
$$
With the aid of this bound we obtain
\begin{align*}
\pm G_3 
\leq C n^{-\frac{1}{2}} \frac{n}{\ell} \left(\delta_M \theta_M (\cN-1) \dG(-\Delta) \theta_M + \delta_M^{-1} \dG(-\Delta) \right)
 \leq C \frac{n^\frac{1}{2}M^\frac{1}{2}}{\ell} \dG(-\Delta),
\end{align*}
which completes the proof of \Cref{lem_kinetic_commutator_cubic} since $\lambda \leq 1$.
\end{proof}

\begin{proof}[Proof of \Cref{lem_Q4_commutator_cubic}]
With the definition of $K$ in \eqref{eq:K-tK}, and recalling that  $q_x = a_x$ on $\cF_+$,  we compute
\begin{align*}
n^\frac{1}{2}[Q_4, \cB_c] 
&= \theta_M \int_{\Lambda^2} V_{\ell}(x-y) \big(-n\omega_{\ell}(x-y) - \sum_{z \neq 0}n\omega_{\ell,\lambda }(P_z(x)-y) + K_2 \big) a_x^* a_y^* a_x +  \hc
\\
& \quad - \left(\theta_M \int_{\Lambda^2} V_{\ell}(x-y) a_x^* a_y^* a(K_x) +  \hc \right)
\\
& \quad   + \theta_M \int_{\Lambda^2} V_{\ell}(x-y) a_x^*a_y^* \int_\Lambda a_v^* K(v,x)a_v a_y  +  \hc
\\
& \quad + \theta_M \int_{\Lambda^3} V_{\ell}(x-y) a_x^* a_v^* a^*(K_v) a_x a_y +  \hc 
\\
&\quad - \left(\theta_M \int_{\Lambda^2} V_{\ell}(x-y) a_x^*a_y^* \int_\Lambda a^*(K_v) a_x a_v +  \hc \right)
\\
&=: - \theta_M \int_{\Lambda^2} n(V_{\ell}\omega_{\ell})(x-y) a_x^* a_y^* a_x +  \hc + n^\frac{1}{2} \sum_{i=1}^6 I_i.
\end{align*}
Here we extracted the main term and it remains to estimate all the $I_i$.  
From \eqref{eq_w_z_bound} we have
$$\sum_{z \neq 0} |\omega_{\ell,\lambda }(P_z(x)-y)| \leq \frac{C}{1+\ell{\rm d}(x,\partial \Lambda)} .$$
With this and \eqref{eq_sobolev_trick} we find 
\begin{align*}
\pm I_1 &= \mp \theta_M n^\frac{1}{2} \sum_{z \neq 0} \int_{\Lambda^2} V_{\ell}(x-y) \omega_{\ell,\lambda }(P_z(x)-y) a_x^*a_y^*a_x + \hc 
\\
&\leq  \delta Q_4 + C\delta^{-1} n \|V_{\ell}\|_1 \dG\left( (1+\ell{\rm d}(\,\cdot\,,\partial \Lambda))^{-2} \right)
\\
&\leq \delta Q_4 + C\delta^{-1} \frac{n}{\ell} \| (1+\ell{\rm d}(\,\cdot\, ,\partial \Lambda))^{-2} \|_{3/2} \dG(-\Delta)
\\
&\leq \delta Q_4 + C\delta^{-1} {n}\ell^{-\frac{5}{3}} \dG(-\Delta)
\end{align*}
for all $\delta > 0$. Recalling the bound \eqref{eq:K2_bound} we obtain
\begin{align*}
\pm I_2 &= \pm \theta_M n^{-\frac{1}{2}} \int_{\Lambda^2} V_\ell(x-y) K_2 a_x^*a_y^*a_x + \hc \leq \delta Q_4 + C \delta^{-1} \lambda^4 \frac{n}{\ell^3} \cN.
\end{align*}
The term $I_3$ can be bounded by a simple  Cauchy--Schwarz estimate as
\begin{align*}
\pm I_3 &\leq \delta Q_4 + C \delta^{-1} \lambda \frac{n}{\ell^3} \cN.
\end{align*}
The choice $\delta = n^{1/2}\ell^{-5/6}$ and $\cN \leq  \pi^{-2} \dG(-\Delta)$ on $ \cF_+$ yields that $I_1+I_2+I_3$ may be bounded as stated in the lemma. For the remaining terms we set $\eta = n^{1/2} M \ell^{-3/2} $. The Cauchy--Schwarz inequality gives
\begin{align*}
\pm I_4 &\leq  \eta M^{-2} \int_{\Lambda^3} V_{\ell}(x-y) \theta_M a_x^*a_y^* a_v^* (\cN+1) a_va_x a_y \theta_M
\\
&\qquad \qquad \qquad + n^{-1} \eta^{-1} M^2 \int_{\Lambda^3} V_{\ell}(x-y) K(v,x)^2 a_y^* a_v^*  (\cN+1)^{-1} a_v a_y.
\end{align*}
From \Cref{lem_vep} and \eqref{eq:K-tK} we find that $|K(x,v)| \leq C \frac{n}{\ell} |x-v|^{-1}$. 
Combining this with the Hardy inequality as in \eqref{eq_Hardy} we obtain
\begin{equation*} 
\int_{\Lambda^3} f(y,v)^2 V_{\ell}(x-y) K(v,x)^2 {\rm d}v \dx \dy \leq C \frac{n^2}{\ell^3} \|f\|_{H^1(\Lambda^2)}^2
\end{equation*}
for all $f \in H^1(\Lambda^2)$ 
so that
\begin{equation}\label{eq:K^2_sobolev}
\int_{\Lambda^3} V_\ell(x-y)K(v,x)^2 a_y^*a_v^*a_va_y \leq C \frac{n^2}{\ell^3}(\cN-1)\dG(-\Delta).
\end{equation}
We hence obtain
$$
\pm I_4 \leq \eta M^{-2} \theta_M Q_4 \cN^2 \theta_M + C n^{-1} \eta^{-1} M^2 \frac{n^2}{\ell^3} \dG(-\Delta) 
\leq C \frac{n^\frac{1}{2} M}{\ell^\frac{3}{2}} \left(Q_4 + \dG(-\Delta) \right).
$$
The remaining terms are bounded via simple Cauchy--Schwarz  estimates, as
\begin{align*}
\pm I_5 &= \pm \theta_M n^{-\frac{1}{2}} \int_{\Lambda^3} V_{\ell}(x-y) a_x^* a_v^* a^*(K_v) a_x a_y +  \hc
\\
&\leq n^{-1} \eta^{-1} \|V_{\ell}\|_1 \sup_v \|K_v\|_2^2 \theta_M \cN^3 \theta_M + \eta Q_4
\\
& \leq \eta Q_4 + C n^{-1} \eta^{-1} \frac{n^2 M^2}{\ell^3} \lambda \, \cN \leq C \frac{n^\frac{1}{2} M}{\ell^\frac{3}{2}} \left(Q_4 + \dG(-\Delta) \right).
\\
\pm I_6 &\leq \eta Q_4 + C \eta^{-1} \lambda \frac{n^2}{\ell^2} \frac{M^2}{n\ell} \cN \leq C \frac{n^\frac{1}{2} M}{\ell^\frac{3}{2}} \left( Q_4 + \dG(-\Delta) \right).
\end{align*}
The proof of \Cref{lem_Q4_commutator_cubic} is complete.
\end{proof}

Now we are ready to give the proof of \Cref{lem_kinetic_preserved}.

\begin{proof}[Proof of \Cref{lem_kinetic_preserved}] Let us start by showing that
\[ \label{eq_Q4+kinetic_preserved}
e^{-t\cB_c} (Q_4+\dG(-\Delta))e^{t\cB_c} \leq C(Q_4+\dG(-\Delta) + \frac{n}{\ell}(\cN+1))
\]
We shall do this for $0\leq t \leq 1$, the proof in the case $-1\leq t<0$ works the same.  
We first use the Duhamel formula as well as Lemmata  \ref{lem_kinetic_commutator_cubic} and \ref{lem_Q4_commutator_cubic}, and subsequently the Cauchy--Schwarz inequality together with \Cref{lem_number_preserved_cubic}, $\theta_M \leq 1$ and \eqref{rem:kin_Q4_Tc} with $\sigma \leq 1$ to obtain
\begin{align*}
& e^{-t\cB_c} \big(Q_4+\dG(-\Delta)\big)e^{t\cB_c} - Q_4-\dG(-\Delta) = \int_0^t e^{-s\cB_c} [Q_4+\dG(-\Delta),\cB_c] e^{s\cB_c} \ds 
\\
&= - \int_0^t e^{-s\cB_c} \left(\theta_M \int_{\Lambda^2} n^\frac{1}{2} V_{\ell}(x-y)a_x^*a_y^*a_x +  \hc - \mathcal{E}_c^{(Q_4)} - \mathcal{E}_c^{(\dG(-\Delta))}\right) e^{s\cB_c} \ds
\\
&\leq \int_0^t e^{-s\cB_c} \left(Q_4 + C \frac{n}{\ell} \cN\right) e^{s\cB_c} \ds + C \int_0^t e^{-s\cB_c} \Big( Q_4 + \dG(-\Delta) \Big) e^{s\cB_c} \ds
\\
&\leq C \int_0^t e^{-s\cB_c} \big(Q_4 + \dG(-\Delta)\big) e^{s\cB_c} \ds + C \frac{n}{\ell}(\cN + 1).
\end{align*}
The Gr\"onwall lemma then yields \eqref{eq_Q4+kinetic_preserved}.
This immediately implies the result for the kinetic operator $\dG(-\Delta)$ in \Cref{lem_kinetic_preserved}. For the quartic operator $Q_4$ we repeat the previous argument by solely considering $Q_4$. The Duhamel formula,  \Cref{lem_Q4_commutator_cubic} and the Cauchy--Schwarz inequality yield
\begin{align*}
e^{-t\cB_c}  Q_4 e^{t\cB_c} - Q_4 = \int_0^t e^{-s\cB_c} [Q_4,\cB_c] e^{s\cB_c} \ds
\leq \int_0^t e^{-s\cB_c} \left(Q_4 + C\frac{n}{\ell}\cN + \mathcal{E}_c^{(Q_4)} \right) e^{s\cB_c} \ds.
\end{align*}
We then use the estimate of the error term in \Cref{lem_Q4_commutator_cubic} together with \Cref{lem_number_preserved_cubic} and \eqref{eq_Q4+kinetic_preserved} to find
\begin{align*}
e^{-t\cB_c}  Q_4 e^{t\cB_c} - Q_4 &\leq C \int_0^t e^{-s\cB_c} \left(Q_4  + \sigma \dG(-\Delta)  \right) e^{s\cB_c} \ds + C\frac{n}{\ell}(\cN + 1)
\\
&\leq C \int_0^t e^{-s\cB_c} Q_4 e^{s\cB_c} \ds + C \sigma \left(Q_4 + \dG(-\Delta) + \frac{n}{\ell}(\cN+1)\right)
 + C\frac{n}{\ell}(\cN + 1)
\\
&\leq C \int_0^t e^{-s\cB_c} Q_4 e^{s\cB_c} \ds + C \big(Q_4 + \frac{n}{\ell}(\cN+1) + \sigma \dG(-\Delta)\big).
\end{align*}
In the last inequality we gathered some terms due to $\sigma \leq 1$.
Now \Cref{lem_kinetic_preserved} follows again by the Gr\"onwall lemma.
\end{proof}

\subsection{Analysis of {\rm (I)}\textsubscript{$c$}} \label{sec:Q1-Bc}

In this section we analyze the term 
\begin{equation}\label{defI}
{\rm (I)}_c
	= e^{-\cB_c} \dG\left(nV_{\ell} \ast u_0^2 + n \widehat V_{\ell} - 8\pi\mathfrak{a} \frac{n}{\ell} \right) e^{\cB_c} + \int_0^1 \int_t^1 e^{-s\cB_c}[Q_3,\cB_c]e^{s\cB_c}\ds\dt
\end{equation}
appearing in \eqref{eq_HNcubic}. We prove that the main contribution of  this term is $8\pi\mathfrak{a} \cN n \ell^{-1}$. Prior to this let us show a lemma that will be used to handle the diagonal terms in \eqref{defI}. It is the cubic analogue of \Cref{lem_bounded_quadratictrafo}.

\begin{lemma} \label{lem_bounded_cubictrafo}
Assume that $\lambda \left(\frac{n}{\ell}\right)^2 \leq 1$ and that $2 R/\ell < \lambda < 1/4$. Let $A: L^2(\Lambda) \to L^2(\Lambda)$ be a bounded hermitian operator. For all $t \in [-1,1]$ we have on $\cF$
\begin{align*}
\pm \left(e^{-t\cB_c} \dG(A) e^{t\cB_c} - \dG(A)\right) \leq C \|A\|_{\rm op} \|K\|_2 (\cN+1) \leq C\|A\|_{\rm op} \lambda ^\frac{1}{2} \frac{n}{\ell} (\cN+1).
\end{align*}
\end{lemma}

\begin{proof} 
One easily checks that for a bounded operator $T : L^2(\Lambda) \to L_{\rm s}^2(\Lambda^2)$ with integral kernel $T(x_1,x_2,;y_1)$ and for any $\xi,\xi' \in \mathcal F$, the following inequality holds
\begin{align}
\pm \int_{\Lambda^3} T(x_1,x_2;y_1) \langle \xi', a^*_{x_1} a^*_{x_2} a_{y_1} \xi \rangle + \hc \leq C \|T\|_{\rm op} \|\left(\mathcal N + 1\right)^{3/4} \xi' \|  \|\left(\mathcal N + 1\right)^{3/4} \xi \|. \label{eq:op_bound}
\end{align}
Define the operator $B_c : L^2(\Lambda) \to L_{\rm s}^2(\Lambda^2)$ by 
\begin{align}\label{eq:Bc_op}
\langle g, B_c f\rangle = n^{-1/2} \int_{\Lambda^2} \overline{(Q^{\otimes 2} g)}(x_1,x_2)  K(x_1,x_2) (Q f)(x_2) \dd x_1 \dd x_2,
\end{align}
for all $f \in L^2(\Lambda)$ and $g \in L_{\rm s}^2(\Lambda^2)$, so that 
\begin{align*}
\cB_c &= \theta_M  \int_{\Lambda^3} B_c(x_1,x_2;y_1) a^*_{x_1} a^*_{x_2} a_{y_1}  \dd x_1 \dd x_2 \dd y_1 -  \hc
\end{align*}
From the Cauchy--Schwarz inequality, we have
\begin{align*}
\|B_c\|_{\rm op} \leq n^{-1/2} \sup_{x\in \Lambda} \|K_x\|_2 \leq C\lambda ^\frac{1}{2} \frac{n^{1/2}}{\ell},
\end{align*}
where the second inequality is a consequence of Lemma \ref{lem_K_properties}. 

With the aid of the  Duhamel formula we can write
\begin{align}\label{eq:duh_dG_A}
e^{-t\cB_c} & \dG(A)  e^{t\cB_c} -  \dG(A)  = \int_0^t e^{-s \cB_c} [\dG(A),\cB_c] e^{s \cB_c} \ds.
\end{align}
Applying (\ref{eq:op_bound}) and that $M \leq n$ we can bound the commutator as
\begin{align*}
\pm [\dG(A),\cB_c] 
	&= \pm \theta_M  \int_{\Lambda^2} (A_1 B_c + A_2 B_c - B_c A)(x_1,x_2;y_1) a^*_{x_1} a^*_{x_2} a_{y_1} + \hc \\
	&\leq C \|A\|_{\rm op} \|B_c\|_{\rm op} \theta_{M} \left(\mathcal N + 1\right)^{3/2} \\
	&\leq C \|A\|_{\rm op} M^{1/2}  \lambda ^\frac{1}{2} \frac{n^{1/2}}{\ell} \left(\mathcal N + 1\right) \leq C \|A\|_{\rm op} \lambda^\frac{1}{2} \frac{n}{\ell}\left(\mathcal N + 1\right)\,.
\end{align*}
Here we have denoted by $A_1B_c$, $A_2 B_c$ and $B_c A$ the composition of applications, where $A_i$ is the operator $A$ acting on the variable $i$.
Now plugging this estimate into (\ref{eq:duh_dG_A}) and applying \Cref{lem_number_preserved_cubic} we obtain the claim of the lemma.
\end{proof}


\begin{lemma} \label{prop_quadratic_cubictrafo}
Assume that $\lambda \left(\frac{n}{\ell}\right)^2 \leq 1$, that $2 R/\ell < \lambda < 1/4$ and that $\ell$ is large enough. Let $\sigma \leq 1$ as in \Cref{lem_cubictrafo}. On $\cF_+$ we have
\begin{align*}
{\rm (I)}_c 
= 8\pi\mathfrak{a} \cN \frac{n}{\ell}  + 
\Xi_c
\end{align*}
with
\begin{align*}
\pm \Xi_c & \leq C \delta \left(Q_4 + \frac{n}{\ell}(\cN+1) \right) + C\delta^{-1} \frac{n}{\ell}\frac{(\cN+1)^2}{M} + C \delta^{-1} \frac{M}{n} \lambda \left(\frac{n}{\ell}\right)^3(\cN+1) 
\\
& \quad + C \sigma \left(Q_4 + \dG(-\Delta) + \frac{n}{\ell}(\cN+1) \right) + C \lambda ^\frac{1}{2} \left(\frac{n}{\ell}\right)^2 (\cN+1)
\end{align*}
for all $0 < \delta \leq 1$.
\end{lemma}

\begin{proof} 
We start with computing the commutator $[Q_3,\cB_c]$. 
A lengthy but straightforward computation shows that
\begin{align}\nonumber
 [Q_3,\cB_c] &= \theta_M \int_{\Lambda^2} V_{\ell}(x-y)K(x,y) (a_x^*a_x + a_x^*a_y) +  \hc
\\ \nonumber
& \quad +  \Big\{ \theta_M \Big( - \int_{\Lambda^2} V_\ell(x-y) a_x^* (a(K_x) + a(K_y)) + \hc \Big) \Big\}  
\\ \nonumber
& \quad +\Big\{ \theta_M \Big[ \int_{\Lambda^2} V_{\ell}(x-y) a_x^* a^*(K_y)a_x a_y + \int_{\Lambda^3} V_{\ell}(x-y) a_v^*a^*(K_v) a_y a_x
\\ \nonumber
& \hspace{3cm} + \int_{\Lambda^3} V_{\ell}(x-y)  a_x^*a_y^* a_v^* K(v,x) a_v  
 - \int_{\Lambda^2} V_{\ell}(x-y) a_y^*a_x^*a^*(K_x)a_y
\\ \nonumber
& \hspace{3cm} - \int_{\Lambda^3} V_\ell(x-y) a_x^*a_y^* a^*(K_v)a_v \Big] + \hc \Big\}
\\ \nonumber
& \quad + \Big\{ \theta_M \Big[ - \int_{\Lambda^3} V_\ell(x-y) a_x^* a^*(K_v) a_v (a_y + a_x) \\ \nonumber
&\qquad \qquad\qquad \qquad  + \int_{\Lambda^3} V_{\ell}(x-y) a_y^*a_v^*a^*(K_v)(a_x + a_y)\Big] +\hc \Big\}
\\ \nonumber
&  \quad + \Big\{\theta_M \int_{\Lambda^3} V_{\ell}(x-y) a_x^* a_v^* (K(v,x) a_y + K(v,y) a_x) a_v  +  \hc \Big\}
\\ \nonumber
& \quad + \Big\{  \int_{\Lambda^3} V_{\ell}(x-y) \left([a_x^*a_y^*a_x,\theta_M] + [a_x^*a_y a_x, \theta_M] \right) a_v^*a^*(K_v) a_v +  \hc \Big\}
\\ \label{q3c}
&=: \int_{\Lambda^2} V_{\ell}(x-y) K(x,y) (a_x^*a_x + a_x^*a_y) +  \hc + \sum_{i=1}^6 J_i
\end{align}
where we set
$$
J_1 =  (\theta_M-1)\int_{\Lambda^2} V_{\ell}(x-y)K(x,y) (a_x^*a_x + a_x^*a_y) +  \hc 
$$
and $J_2, \dots, J_6$ denotes the remaining expressions in curly brackets. 
For the error term $J_1$ we use that 
$$\pm \left(\theta_M - 1\right) \leq \mathds{1}^{\cN \geq M/2}.$$
Hence, by the Cauchy--Schwarz inequality
\begin{align} \label{eq:E_theta_bound}
\pm J_1 & \leq \delta \|K\|_\infty \int_{\Lambda^2} V_\ell(x-y) a_x^*a_x + \delta^{-1} \|K\|_\infty \int_{\Lambda^2} V_\ell(x-y) (\theta_M - 1) a_x^*a_x (\theta_M - 1) 
\nn\\
&\leq C \delta \frac{n}{\ell}\cN + C \delta^{-1} \frac{n}{\ell} \cN \mathds{1}^{\cN \geq M/2}
\leq C \delta \frac{n}{\ell}\cN + C \delta^{-1} \frac{n}{\ell} \frac{\cN^2}{M}
\end{align}
for all $\delta > 0$.
Similarly we bound the error terms $J_2$ through $J_6$ with the aid of the Cauchy--Schwarz inequality using that $\lambda, \ell^{-1}, M n^{-1} \leq 1$. For any $\delta > 0$ we obtain
\begin{align*}
\pm J_2 &\leq C \ell^{-1} \sup_x \|K_x\|_2 \cN \leq C \lambda \frac{n^2}{\ell^3} \cN \leq C \lambda^\frac{1}{2} \left(\frac{n}{\ell}\right)^2 \cN,
\\
\pm J_3 & \leq C \delta^{-1} \|V_{\ell}\|_1 \sup_y \|K_y\|_2^2 \theta_M (\cN+1)^2 \theta_M + \delta Q_4,
\\
& \leq \delta Q_4 + C \delta^{-1} \frac{M}{\ell} \lambda  \left(\frac{n}{\ell}\right)^2 (\cN+1)
\\
\pm J_4 & \leq C \theta_M \|V_{\ell}\|_1 \sup_v \|K_v\|_2 (\cN+1)^2 \leq C \lambda ^\frac{1}{2} \left(\frac{n}{\ell}\right)^2 (\cN+1).
\end{align*}
For $J_5$ recall the bound \eqref{eq:K^2_sobolev}.  With $\delta_5 > 0$ appropriately chosen 
\begin{align*}
\pm J_5 & \leq C \delta_5^{-1} \int_{\Lambda^3} V_{\ell}(x-y) \theta_M  a_x^*a_v^* (\cN+1) a_v a_x \theta_M 
\\
& \qquad + \delta_5 \int_{\Lambda^3} V_{\ell}(x-y) a_y^*a_v^* (\cN+1)^{-1} a_v a_y K(v,x)^2
\\
&\leq C \delta_5 \frac{n^2}{\ell^3} \dG(-\Delta) + C \delta_5^{-1} \ell^{-1} \theta_M (\cN+1)^3 \theta_M
\\
&= \eta \dG(-\Delta) + C\eta^{-1} \frac{n^2 M^2}{\ell^4} (\cN+1).
\end{align*}
for any $\eta > 0$. 
For  $J_6$ we use the pull-through formula $a_x \theta_M(\cN) = \theta_M(\cN+1) a_x$
which yields
\begin{align*}
[a_x^*a_y^* a_x,\theta_M] a_v^*a^*(K_v) a_v 
&= a_x^*a_y^* \big(\theta_M(\cN+1) - \theta_M(\cN+2)\big) a_x a_v^* a^*(K_v) a_v
\\
&= a_x^*a_y^* \big(\theta_M(\cN+1) - \theta_M(\cN+2)\big) (\delta_{x,v} + a_v^*a_x) a^*(K_v) a_v
\\
&= \delta_{x,v} a_x^*a_y^* \big(\theta_M(\cN+1) - \theta_M(\cN+2)\big) a^*(K_v) a_v 
\\
& \quad + a_x^*a_y^* a_v^* \big(\theta_M(\cN+2) - \theta_M(\cN+3)\big) a_x a^*(K_v) a_v.
\end{align*}
By the smoothness assumption on $\theta_M$  we can bound 
\begin{equation} \label{eq_theta_operatorbound}
\big \| \theta_M(\cN+i)-\theta_M(\cN+j) \big\|_{\rm op} \leq C M^{-1} |i-j|, \; i,j \in \mathbb{N}_0,
\end{equation}
for some constant $C>0$. Therefore,
\begin{align*}
& \int_{\Lambda^3} V_{\ell}(x-y) [a_x^*a_y^*a_x,\theta_M] a_v^*a^*(K_v) a_v +  \hc
\\
& = \int_{\Lambda^2} V_{\ell}(x-y) a_x^*a_y^* \big(\theta_M(\cN+1) - \theta_M(\cN+2)\big)a^*(K_x)a_x +  \hc
\\
& \quad + \int_{\Lambda^3} V_{\ell}(x-y) a_x^*a_y^* a_v^* \big(\theta_M(\cN+2) - \theta_M(\cN+3)\big) a_x a^*(K_v)a_v +  \hc
\\
& \leq \delta Q_4 + \delta^{-1} \int_{\Lambda^2} V_{\ell}(x-y) a_x^*a(K_x) \big(\theta_M(\cN+1) - \theta_M(\cN+2)\big)^2 a^*(K_x)a_x 
\\
& \quad + \delta Q_4 + \delta^{-1} \int_{\Lambda^3} V_{\ell}(x-y) a_v^*a(K_v)a_x^* \big(\theta_M(\cN+2) - \theta_M(\cN+3)\big) (\cN+1) \times  \\
&  \hfillcell{ \times \big(\theta_M(\cN+2) - \theta_M(\cN+3)\big) a_x a^*(K_v)a_v}
\\
& \leq 2 \delta Q_4 + C \delta^{-1} \|V_{\ell}\|_1  M^{-2} \lambda  \left(\frac{n}{\ell}\right)^2 \Big( \mathds{1}(\cN \leq 2M) (\cN+1)^2 \mathds{1}(\cN \leq 2 M) 
\\
& \hfillcell{ + \mathds{1}(\cN-1 \leq 2 M)  (\cN+1)^4 \mathds{1}(\cN-1 \leq 2 M)  \Big) }
\\
& \leq 2 \delta Q_4 + C \delta^{-1} \frac{M}{\ell} \lambda  \left(\frac{n}{\ell}\right)^2 (\cN+1).
\end{align*}
Let us now consider the second part of $J_6$ containing  $[a_x^*a_ya_x,\theta_M]$. Computing the commutator and  normal ordering the expression one sees three terms appearing. One may be bounded as above and the other two as 
$$
\pm \left(\theta_M(\cN+1)-\theta_M(\cN)\right) \int_{\Lambda^2} V_\ell(x-y) K(x,y) a_x^* (a_x + a_y) + \hc \leq C \frac{n}{\ell}\frac{\cN}{M}
$$
with the Cauchy--Schwarz inequality and $\|K\|_\infty \leq C n$ and
\begin{align*}
& \pm \left(\theta_M(\cN+1)-\theta_M(\cN)\right) \int_{\Lambda^3} V_\ell(x-y) a_x^*a_v^* a_v \left( a_x K(v,y) + a_y K(v,x) \right) + \hc
\\
&\leq C \delta_6^{-1} M \ell^{-1} + \delta_6 \int_{\Lambda^3} V_\ell(x-y) a_x^*a_v^* (\cN+1)^{-1} a_x a_v K(v,y)^2 
\\
&\leq C \delta_6 \frac{n^2}{\ell^3}\dG(-\Delta) + C \delta_6^{-1} M \ell^{-1}
\\
&= \eta \dG(-\Delta) + C\eta^{-1} n^2 M \ell^{-4}.
\end{align*}
for any $\eta>0$ using \eqref{eq:K^2_sobolev} as in the analysis of $J_5$.  
We conclude that 
$$
\pm J_6 \leq C \delta  Q_4  + C \delta^{-1} \frac{M}{\ell} \lambda  \left(\frac{n}{\ell}\right)^2  (\cN + 1) + \eta \dG(-\Delta) + C \eta^{-1} n^2 M\ell^{-4} + C \frac{n}{\ell} \frac{\cN}{M}.
$$ 

With \eqref{q3c} we have 
\begin{align} \label{eq:Q3_commutator_Bc_expanded}
&\int_0^1 \int_t^1  e^{-s\cB_c}[Q_3,\cB_c]e^{s\cB_c} \ds\dt \nn\\
&= 2 \int_0^1 \int_t^1 e^{-s\cB_c} \left(\int_{\Lambda^2} V_{\ell}(x-y)K(x,y) (a_x^*a_x+a_x^*a_y) \right) e^{s\cB_c} \ds\dt \nn
\\
&\qquad + \int_0^1 \int_t^1 e^{-s\cB_c} \sum_{i=1}^6 J_i e^{s\cB_c} \ds\dt.
\end{align}
Collecting the bounds on the error terms, we obtain with the aid of  \Cref{lem_number_preserved_cubic} and \Cref{lem_kinetic_preserved} 
\begin{align*}
& \pm \int_0^1 \int_t^1 e^{-s\cB_c} \sum_{i=1}^6 J_i e^{s\cB_c} \ds \dt \\
& \leq C \delta \left(Q_4 + \frac{n}{\ell}(\cN+1) + \sigma \dG(-\Delta)\right) + C \delta^{-1} \frac{n}{\ell}\frac{(\cN+1)^2}{M} + C \delta^{-1} \frac{M}{n}\lambda \left(\frac{n}{\ell}\right)^3 (\cN+1) 
\\
& \quad + C \eta \left(Q_4 + \dG(-\Delta) + \frac{n}{\ell}(\cN+1) \right) + C \eta^{-1} \frac{n^2M^2}{\ell^4}(\cN+1) + C \lambda ^\frac{1}{2} \left(\frac{n}{\ell}\right)^2 (\cN+1)
\end{align*}
where we have used that $\delta, M^{-1} \leq 1$. 
We choose $\eta = \sigma$. In particular $\eta^{-1} n^2M^2 \ell^{-4} \leq \sigma n\ell^{-1}$ so that with $\delta \leq 1$ we find
\begin{align*}
& \pm \int_0^1 \int_t^1 e^{-s\cB_c} \sum_{i=1}^6 J_i e^{s\cB_c} \ds \dt \\
&\leq C \delta \left(Q_4 + \frac{n}{\ell}(\cN+1)\right) + C \delta^{-1} \frac{n}{\ell} \frac{(\cN+1)^2}{M} +
 C \delta^{-1} \frac{M}{n} \lambda \left(\frac{n}{\ell}\right)^3(\cN+1) 
\\
& \quad + C\sigma  \left(Q_4 + \dG(-\Delta) + \frac{n}{\ell}(\cN+1) \right) + C \lambda ^\frac{1}{2} \left(\frac{n}{\ell}\right)^2 (\cN+1). 
\end{align*}

The first line of \eqref{eq:Q3_commutator_Bc_expanded} as well as the first term in \eqref{defI} can be estimated 
with the aid of Lemma \ref{lem_bounded_cubictrafo}.
One readily checks from the elementary properties of $K$ in Lemma \ref{lem_K_properties} that the various operators (taking the place of $A$ in   Lemma \ref{lem_bounded_cubictrafo}) are bounded by $C n \ell^{-1}$. 
Therefore, 
\begin{align} \label{eq:Q3_cubic}
{\rm (I)}_c & = 
 \int_{\Lambda^2} V_{\ell}(x-y)K(x,y) (a_x^*a_x+a_x^*a_y)
+ 
\dG\left( nV_{\ell} \ast u_0^2 + n \widehat V_{\ell} - 8\pi a \frac{n}{\ell}\right)\nn
\\
& \quad + \int_0^1 \int_t^1 e^{-s\cB_c} \sum_{i=1}^6 J_i e^{s\cB_c} \ds \dt + J_7
\end{align}
with
$$\pm J_7 \leq C \lambda ^\frac{1}{2} \left(\frac{n}{\ell}\right)^2 (\cN+1).$$
The first line of \eqref{eq:Q3_cubic} can equivalently be written as
\begin{align} \nn
& \int_{\Lambda^2} V_{\ell}(x-y)K(x,y) (a_x^*a_x+a_x^*a_y)
+ 
\dG\left( nV_{\ell} \ast u_0^2 + n \widehat V_{\ell} - 8\pi \ao \frac{n}{\ell}\right) \\ 
&= 8\pi\mathfrak{a}\frac{n}{\ell}\cN + 2 \dG(h) + J_8 \label{eq:Q3_diagonalized}
\end{align}
where $h$ stands for multiplication with the function $h$ in \eqref{eq_scattering_convolution} and 
\begin{align}
J_8 =  \int_{\Lambda^2} V_{\ell}(x-y)(n+K(x,y)) \left( a_x^* a_y - a_x^* a_x \right) \dx \dy  
\end{align}
%
Using the bound on $h$ from \Cref{lem_scatlength} and the Sobolev inequality \eqref{eq_sobolev_trick} we find
$$
\ \pm \dG(h) \leq C\|h\|_{3/2} \dG(-\Delta)
\leq C \sigma  \dG(-\Delta).
$$
For $J_8$ we shall show that operator inequality
$$
\pm J_8 \leq C \frac{n}{\ell^3} \dG(-\Delta).
$$
Indeed, $J_8 = \dG(S)$ and for $f \in H^1(\Lambda)$
\[ \label{eq_quadratic_convolution}
\braket{f, S f} & 
= -\frac{1}{2} \int_{\Lambda^2} V_{\ell}(x-y)(n+K(x,y)) |f(x)-f(y)|^2 \dx  \dy
\\
&  \leq C n \int_{\Lambda^2} V_{\ell}(x-y) \left| \int_0^1 \nabla f(x + t(y-x)) \cdot (y-x) \dt \right|^2 \dx\dy
\\
&  \leq C n \int_0^1 \int_{\Lambda^2} |x-y|^2V_{\ell}(x-y) |\nabla f (x+t(y-x))|^2 \dx\dy \dt
\\
&  \leq Cn \ell^{-2} \int_0^1 \int_{\R^3 \times \R^3 } \mathds{1}_{\Lambda}(x+y) \mathds{1}_{\Lambda}(y) V_{\ell}(x) |\nabla f(x+y - tx)|^2 \dx\dy \dt
\\
&  \leq C n \ell^{-3} \|\nabla f\|_2^2.
\]
Combining all the estimates, we conclude the proof of \Cref{prop_quadratic_cubictrafo}.
\end{proof}

\subsection{Analysis of {\rm(II)}\textsubscript{$c$}}
	\label{sec:Q2-Bc}
	
In this section we analyze the term 
	\begin{align*}
{\rm (II)}_c = e^{-\cB_c} \widetilde Q_2 e^{\cB_c}  =  e^{-\cB_c} \widetilde Q_2^{(\epsilon)} e^{\cB_c} +  e^{-\cB_c} \widetilde Q_2^{(bc)} e^{\cB_c} 
\end{align*}
appearing in \eqref{eq_HNcubic}, where we decompose $\widetilde{Q}_2 = \widetilde{Q}_2^{(\epsilon)} + \widetilde{Q}_2^{(bc)}$ as in \eqref{eq:Q2-eps-def}. As we will show below, the cubic transformation leaves $\widetilde{Q}_2$ essentially unchanged. However, after its action  the boundary term $\widetilde{Q}_2^{(bc)}$ can be absorbed into the error terms. More precisely, we have the following lemma.

\begin{lemma} \label{prop_Q2_commutator_cubic} Assume that $\lambda \left(\frac{n}{\ell}\right)^2 \leq 1$, that $2 R/\ell < \lambda < 1/4$ and that $\ell$ is large enough. Let $\sigma \leq 1$ as in \Cref{lem_cubictrafo}. On $\cF_+$ we have
$$
{\rm (II)}_c = \widetilde{Q}_2^{(\epsilon)} + \mathcal{E}_c^{(\widetilde{Q}_2)}
$$
with 
\begin{align*}
\pm \mathcal{E}_c^{(\widetilde{Q}_2)} &\leq \frac{1}{4} Q_4  + C \sigma \left(\dG(-\Delta) + Q_4 + \frac{n}{\ell} (\cN+1) \right) \\
&\quad + \delta \left(Q_4 + \frac{n}{\ell}(\cN+1) \right) +  \delta^{-1} C \lambda  \left(\frac{n}{\ell}\right)^3 (\cN+1)+ C \lambda^{-1} \frac{n^{3/2}}{\ell^2}(\cN+1) + C \frac{n^2}{\ell^2} 
\end{align*}
for all $0 < \delta \leq 1$.
\end{lemma}

\begin{proof} 
With the aid of the Duhamel formula we can write
\begin{align*}
e^{-\cB_c} \widetilde{Q}_2 e^{\cB_c} 
= \widetilde{Q}_2^{(\epsilon)} + \widetilde{Q}_2^{(bc)} + \int_0^1 e^{-t\cB_c} [\widetilde{Q}_2,\cB_c] e^{t\cB_c} \dt.
\end{align*}
We will bound all terms except the main one $\widetilde{Q}_2^{(\epsilon)}$.  Let us start with $\widetilde{Q}_2^{(bc)}$. With  the aid of the  pointwise bound 
  \eqref{eq_ptwbound_Q2^E} 
and the Cauchy--Schwarz inequality we obtain 
$$
\pm \widetilde{Q}_2^{(bc)} = \pm \left( \int_{\Lambda^2} \widetilde{Q}_2^{(bc)}(x,y)a_x^*a_y^* +  \hc \right) \leq \frac{1}{4} Q_4 + C \frac{n^2}{\ell^2}
$$
where we used that $ \int_{\Lambda^2} V_\ell(x-y) [1+\ell d(x,\partial \Lambda)]^{-2} \leq C \ell^{-2}$. 
The explicit prefactor $1/4$ is chosen for later convenience. 
Next, consider the commutator  
\begin{equation} \label{eq_comm_Q2_cubic}
[\widetilde{Q}_2 ,\cB_c] = n^{-\frac{1}{2}} \theta_M  \int_{\Lambda^2} K(v,w) [\widetilde{Q}_2 ,q_v^*a_w^*q_v] +  \hc
 + n^{-\frac{1}{2}} [\widetilde{Q}_2,\theta_M]\int_{\Lambda^2} K(v,w) q_v^*a_w^*q_v +  \hc 
\end{equation}
The first term of \eqref{eq_comm_Q2_cubic} equals 
$$
n^{-\frac{1}{2}} \theta_M \int_{\Lambda^4} \widetilde{Q}_2(x,y) K(v,w)  [a_x^*a_y^* + a_x a_y, q_v^*a_w^*q_v] +  \hc = \sum_{i=1}^3 R_i
$$
with
\begin{align*}
R_1 &= 
 2 n^{-\frac{1}{2}}\int_{\Lambda^2} \widetilde{Q}_2^{(\epsilon)}(x,y) \theta_M \bigg(a^*(K_y)a_x q_y - q_y^*a_x^* a^*(K_y) 
\\
&  \quad  +\int_\Lambda q_v^*a^*(K_v)a_x^* - \int_\Lambda a^*(K_v) a_x q_v + \int_\Lambda 	q_v^* K(v,x) a_y q_v + K(x,y) q_y - q(K_x) \bigg) +  \hc 
\\
R_2&= 2 n^{-\frac{1}{2}} \int_{\Lambda^2} \widetilde{Q}_2^{(bc)}(x,y) \theta_M \bigg( a^*(K_x)a_x q_y - q_y^*a^*(K_y)a_x^*  \bigg) +  \hc 
\\
R_3&= 2 n^{-\frac{1}{2}} \int_{\Lambda^2} \widetilde{Q}_2^{(bc)}(x,y) \theta_M \bigg( \int_\Lambda q_v^*a^*(K_v)a_x^*  + \int_\Lambda q_v^* K(v,x) a_y q_v - \int_\Lambda a^*(K_v)a_x q_v
 \\
&\qquad\qquad\qquad \qquad \qquad + K(x,y) q_y - q(K_x) \bigg) +  \hc 
\end{align*}
Since all expressions are normal ordered, we can again replace $q_x$ by $a_x$ on $\mathcal F_+$. 

Let us start with estimating $R_1$. 
Recall the bound \eqref{eq:Q2-eps-def-bound-1}.  
For $x \in \Lambda$ fixed we have $\dG(\1_{|x-\,\cdot\,|\leq \lambda}) \leq C \lambda^2 \dG(-\Delta)$ on $\cF_+$ by \eqref{eq_sobolev_trick}.
Using this and $\sup_x \|K_x\|_2 \leq C \lambda^{1/2} n / \ell$ from Lemma \ref{lem_K_properties}, we can bound all the cubic terms in a similar way. Let us bound for example the first one. 
Recall that $\theta_M = \theta_M(\mathcal{N})$ and note that $\theta_M(\mathcal{N})a^*(f) = a^*(f) \theta_M(\mathcal{N}+1)$. 
On $\mathcal F_+$ we have
\begin{align*}
&\pm 2 n^{-1/2} \int_{\Lambda^2} \widetilde{Q}_2^{(\epsilon)}(x,y) \theta_M a^*(K_y)a_x a_y + \hc \\
& \leq C n^{1/2} \lambda^{-3} \ell^{-1} \left(\delta \int_{\Lambda^2} a^*(K_y)a(K_y) \1_{|x-y| \leq \lambda} + \delta^{-1} \theta_M(\mathcal N+1) \int_{\Lambda^2}  \1_{|x-y|\leq \lambda} a^*_x a_y^*  a_y a_x \right) \\
&\leq C n^{1/2} \lambda^{-3} \ell^{-1} \left(\delta \lambda^4 n^2 \ell^{-2} \cN + \delta^{-1} M \lambda^2 \dG(-\Delta) \right)
 \\
 & \leq C n M^{1/2}\ell^{-3/2} \left(\dG(-\Delta) + \frac{n}{\ell}\cN\right)
 \\
&\leq C \sigma \left(\dG(-\Delta) + \frac{n}{\ell}\cN\right),
\end{align*}
where we chose $\delta = \lambda^{-1} n^{-1/2} \ell^{1/2} M^{1/2}$ and used $\lambda n^2 \ell^{-2} \leq 1$.
The other cubic terms can be bounded in a similar way. 
We proceed with the linear terms in $R_1$. We find 
\begin{align*}
& \pm 2 n^{-1/2} \int_{\Lambda^2}  \widetilde{Q}_2^{(\epsilon)}(x,y) K(x,y) \theta_M a_y + \hc 
\\
&\quad \leq C n^{1/2}\lambda^{-3} \ell^{-1} \left(\delta \int_{\Lambda^2}\1_{|x-y| \leq \lambda}a_y^*a_y + \delta^{-1} \int_{\Lambda^2} K(x,y)^2 \right) 
\leq C \lambda^{-1} n^{3/2}\ell^{-2} (\cN+1) .
\end{align*}
Finally, for the last term in $R_1$ we have by the Cauchy--Schwarz inequality
\begin{align*}
\pm 2 n^{-\frac{1}{2}} \int_{\Lambda^2} \widetilde{Q}_2^{(\vep)}(x,y) \theta_M a(K_x) + \hc 
	&\leq C n^{1/2} \ell^{-1} \lambda^{-3} \|\1_{|\,\cdot\,|\leq \lambda}\|_1 \sup_x \|K_x\|_2 \left(\mathcal N+1\right) \\
	&\leq C \lambda^{1/2} n^{3/2} \ell^{-2} \left(\mathcal N+1\right) \leq C \sigma \frac{n}{\ell} (\mathcal N+1).
\end{align*}
This in particular shows that  $\pm R_1 \leq C  \sigma (\dd\Gamma(-\Delta) + \frac{n}{\ell}(\mathcal N +1)) + C \lambda^{-1}n^{3/2}\ell^{-2}(\cN+1)$.

Let us proceed with the analysis of $R_2$. Using the pointwise bound  $|\widetilde{Q}_2^{(bc)}(x,y)| \leq C n V_\ell(x-y)$, we see that we can bound the two terms in $R_2$ using $Q_4$ because of the presence of $a_x a_y$ or $a_y^*a_x^*$. From the Cauchy--Schwarz inequality, we obtain
\begin{align*}
\pm R_2 &\leq \delta Q_4 + \delta^{-1}C n \sup_y \|K_y\|_2^2 \|V_{\ell}\|_1 (\cN+1) 
\leq \delta Q_4 + \delta^{-1} C \lambda  \frac{n^3}{\ell^3}(\cN+1).
\end{align*}
It remains to study $R_3$. The first three terms in $R_3$ are bounded similarly as their counterpart in $R_1$, that is
\begin{align*}
&\pm 2  n^{-\frac{1}{2}} \int_{\Lambda^3} \widetilde{Q}_2^{(bc)}(x,y) \theta_M (a_v^*a^*(K_v)a_x^* + a_v^* K(v,x) a_y a_v  -  a^*(K_v)a_x a_v ) + \hc \\
& \leq C M^{1/2} n^{1/2} \sup_v \|K_v\|_2 \sup_y \|V_\ell \|_1 (\mathcal N+1) \\
&\leq C \lambda^\frac{1}{2} \frac{n^\frac{3}{2} M^\frac{1}{2}}{\ell^2} (\cN+1) \leq  C \sigma \frac{n}{\ell}\left(\mathcal N + 1\right).
\end{align*}
For last term in $R_3$ the Cauchy--Schwarz inequality yields
\begin{align*}
\pm 2 n^{-\frac{1}{2}} \int_{\Lambda^2} \widetilde{Q}_2^{(bc)}(x,y) \theta_M a(K_x) + \hc 
&\leq C n^{1/2} \sup_x \|K_x\|_2 \sup_y \|V_\ell \|_1 (\cN+1)  \\
&\leq  C \lambda^{1/2} \frac{n^{3/2}}{\ell^2} \mathcal N \leq C \sigma \frac{n}{\ell} (\mathcal N+1).
\end{align*}
Finally, for the fourth term in $R_3$, we use the more elaborate pointwise estimate (\ref{eq_ptwbound_Q2^E}) on $\widetilde{Q}_2^{(bc)}(x,y)$ and the Sobolev inequality from \eqref{eq_sobolev_trick} for $\Phi(y) = (1+\ell d(y,\partial \Lambda))^{-5/6}$ so that $\|\Phi\|_{3/2} \leq C \ell^{-2/3}$. 
With $\|K\|_\infty \lesssim n$ we obtain
\begin{align*}
&\pm  2 n^{-1/2} \int_{\Lambda^2} \widetilde{Q}_2^{(bc)}(x,y) K(x,y) \theta_M a_y + \hc 
\\
&\leq C n^{3/2} \left(\delta \int_{\Lambda^2} V_\ell(x-y) \frac{a_y^*a_y}{(1+\ell d(y,\partial \Lambda))^{5/6}} + \delta^{-1} \int_{\Lambda^2} \frac{V_\ell(x-y)}{(1+\ell d(y,\partial \Lambda))^{7/6}} \right)
\\
&\leq C n^{3/2} \left(\delta \ell^{-5/3} \dG(-\Delta) + \delta^{-1} \ell^{-2} \right) 
\leq C \frac{n^{1/2}}{\ell^{5/6}} \left(\dG(-\Delta) + \left(\frac{n}{\ell}\right)^2 \right) 
\\
&\leq C \sigma \left(\dG(-\Delta) + \left(\frac{n}{\ell}\right)^2 \right).
\end{align*}
%
In combination, we obtain
$$
\pm R_3 \leq C \sigma \left(\frac{n}{\ell}(\cN+1) + \dd\Gamma(-\Delta) + \Big(\frac{n}{\ell}\Big)^2\right),
$$
which concludes the bound of the first term of \eqref{eq_comm_Q2_cubic}. 

We now bound the second term in \eqref{eq_comm_Q2_cubic}, given by 
\begin{align*}
&n^{-\frac{1}{2}}  [\widetilde{Q}_2,\theta_M] \int_{\Lambda} a_v^*a^*(K_v) a_v +  \hc
\\
& = n^{-\frac{1}{2}} \int_{\Lambda^2} \widetilde{Q}_2^{(\epsilon)} (x,y) \Big( \big(\theta_M(\cN-2) -\theta_M(\cN) \big) a_x^*a_y^*  +  \hc \Big) \int_{\Lambda} q_v^*a^*(K_v) q_v +  \hc
\\
& \quad + n^{-\frac{1}{2}} \int_{\Lambda^2} \widetilde{Q}_2^{(bc)}(x,y) \Big( \big(\theta_M(\cN-2) -\theta_M(\cN) \big) a_x^*a_y^*   +  \hc \Big) \int_{\Lambda} q_v^*a^*(K_v) q_v +  \hc
\\
&=: R_4 + R_5.
\end{align*}
For the analysis of $R_4$, recall \eqref{eq_theta_operatorbound}
as well as the pointwise bound (\ref{eq:Q2-eps-def-bound-1}) on $\widetilde{Q}_2^{(\epsilon)}$ and $\dG(\1_{|x-\,\cdot\,|\leq \lambda}) \leq C \lambda^2 \dG(-\Delta)$ on $\cF_+$, from which we obtain
\begin{align*}
&\pm n^{-\frac{1}{2}}  \int_{\Lambda^2} \widetilde{Q}_2^{(\epsilon)}(x,y) \big(\theta_M(\cN-2) -\theta_M(\cN) \big) a_x^*a_y^*
\int_{\Lambda} a_v^*a^*(K_v) a_v +  \hc
\\
& \leq C\delta \int_{\Lambda^3}  \big(\theta_M(\cN-2) -\theta_M(\cN) \big)^2 |\widetilde{Q}_2^{(\epsilon)}(x,y)| a_x^*a_y^* a_v^*(\cN+1)a_va_ya_x
\\
&\quad +  \delta^{-1} n^{-1} \int_{\Lambda^3} |\widetilde{Q}_2^{(\epsilon)}(x,y)| a_v^* a(K_v) (\cN+1)^{-1} a^*(K_v) a_v
\\
&\leq C \delta M^{-2} \mathds{1}\left(\cN-2 \leq 2M\right) n\ell^{-1}\lambda^{-1} (\cN+1)^3 \dd\Gamma(-\Delta)  + C\delta^{-1} \ell^{-1}  \sup_{x\in \Lambda} \|K_x\|_2^2 (\cN+1)m
\\
&\leq C \delta M n\ell^{-1}\lambda^{-1} \dd\Gamma(-\Delta)   + C\delta^{-1} \lambda n^2 \ell^{-3}   (\cN+1)\\ 
&\leq C M^{1/2} n \ell^{-3/2} \left(\dd\Gamma(-\Delta) + \frac{n}{\ell} (\mathcal N+1)\right) \leq C \sigma  \left(\dd\Gamma(-\Delta) + \frac{n}{\ell} (\mathcal N+1)\right).
\end{align*}
The same estimate holds for the $a_xa_y$ term in $R_4$ after normal ordering the expression. The terms generated by the commutators are similar to the ones already appearing in $R_1$ and satisfy the same bound. Therefore we obtain
$$
\pm R_4 \leq C \sigma  \left(\dd\Gamma(-\Delta) + \frac{n}{\ell} (\mathcal N+1) + \lambda^{-1} \frac{n^{3/2}}{\ell^2}M^{-1}(\cN+1) \right).
$$
For $R_5$ we use the pointwise bound $|\widetilde{Q}_2^{(bc)}(x,y)|\leq CnV_\ell(x-y)$ and the Cauchy--Schwarz inequality to obtain
\begin{align*}
& n^{-\frac{1}{2}}  \int_{\Lambda^2} \widetilde{Q}_2^{(bc)}(x,y) \big(\theta_M(\cN-2) -\theta_M(\cN) \big) a_x^*a_y^* 
\int_{\Lambda} a_v^*a^*(K_v) a_v +  \hc
\\
& \leq \delta \int_{\Lambda^3} V_{\ell}(x-y) \big(\theta_M(\cN-2) -\theta_M(\cN) \big)^2 a_x^*a_y^* a_v^*(\cN+1)a_v a_ya_x
\\
& \quad + C\delta^{-1} n \int_{\Lambda^3} V_{\ell}(x-y)
a_v^* a(K_v) (\cN+1)^{-1} a^*(K_v) a_v
\\
&\leq C \delta M^{-2} \mathds{1}(\cN-2 \leq 2M) (\cN+1)^2 Q_4 + C\delta^{-1} \lambda  \frac{n^3}{\ell^3} (\cN+1)
\\
&\leq C \delta Q_4 + \delta^{-1} C \lambda  \frac{n^3}{\ell^3}(\cN+1).
\end{align*}
As in the case of $R_4$, the $a_xa_y$ term of $R_5$ is bounded similarly as the term above after normal ordering. The terms arising from the commutators are of the of type of the ones appearing in $R_2$ and $R_3$, we therefore omit their treatment. In total, we find for all $\delta>0$
\begin{align*}
\pm R_5 &\leq \delta Q_4 + C \delta^{-1} \lambda  \frac{n^3}{\ell^3}(\cN+1) + C \sigma \left( \frac{n}{\ell} (\cN+1) + \dd\Gamma(-\Delta) + \left(\frac{n}{\ell}\right)^2 \right),
\end{align*}
which is the desired bound for the second term of \eqref{eq_comm_Q2_cubic}. 

In summary, we have shown the lemma with
$$
\mathcal{E}_c^{(\widetilde{Q}_2)} = \widetilde{Q}_2^{(bc)}+ \int_0^1 e^{-t\cB_c}  \sum_{i=1}^5 R_i e^{t\cB_c} \dt.
$$
Applying \Cref{lem_kinetic_preserved} and \Cref{lem_number_preserved_cubic} to the $R_i$ and simplifying the error terms using $\delta \leq 1$, we conclude the  bound on the error term $\mathcal{E}_c^{(\widetilde{Q}_2)}$ as stated in Lemma \ref{prop_Q2_commutator_cubic}. 
\end{proof}


\subsection{Proof of Lemma \ref{lem_cubictrafo}}
	\label{sec:proof_lem_cubictrafo}
	We are now ready to give the proof of Lemma \ref{lem_cubictrafo}.

\begin{proof}
Recall the notation introduced in  (\ref{eq_HNcubic}). The two main terms $\rm{(I)}_c$ and $\rm{(II)}_c$ were analyzed in \Cref{prop_quadratic_cubictrafo} and \Cref{prop_Q2_commutator_cubic} respectively. Recall that 
$$
\dG(-\Delta) =\sum_{p \in \pi \mathbb{N}_0^3 \setminus \{0\}} p^2 a_p^* a_p. 
$$
Moreover, by \eqref{eq:Q2-eps-def}  and \Cref{prop_symmetric_momentum} 
$$
\widetilde Q_2^{(\epsilon)}  = \int_{\Lambda^2} Q_2^{(\epsilon)} (x,y)  a_x^* a_y^* + \hc =  \frac{n}{2} \sum_{p\in \pi \mathbb{N}_0^3} \widehat \epsilon_{\ell,\lambda}(p) a_p^* a_p^* + \hc ,
$$
which  becomes  the pairing term in the Bogoliubov Hamiltonian \eqref{eq:H_mom-intro} when restricted to $\cF_+$. 
Hence we have 
$$
e^{-\cB_c} e^{-\cB_1} \cH e^{\cB_1} e^{\cB_c} = 4\pi\mathfrak{a}n^2 \ell^{-1} + \mathbb{H}_{\rm Bog} + Q_4 + \Xi_c +  \mathcal{E}_c^{(\widetilde{Q}_2)} + {\rm (III)}_c
$$
with the two error terms  $\Xi_c$ and $\mathcal{E}_c^{(\widetilde{Q}_2)}$ estimated  in \Cref{prop_quadratic_cubictrafo} and \Cref{prop_Q2_commutator_cubic}, respectively. It thus  remains to estimate
\begin{align*}
{\rm (III)}_c = \int_0^1 e^{-t\cB_c} \Big(Q_3 + [\dG(-\Delta)+Q_4,\cB_c]\Big)e^{t\cB_c} \dt + e^{-\cB_c} \mathcal{E}_1  e^{\cB_c} = {\rm (III)}_{c_1} + {\rm (III)}_{c_2},
\end{align*}
where we recall that $\mathcal{E}_1$ satisfies (\ref{eq:E_1}). We deal with the two terms  separately.

\medskip
\noindent
\textbf{Estimating $ {\rm (III)}_{c_1}$.} From Lemmata \ref{lem_kinetic_commutator_cubic} and \ref{lem_Q4_commutator_cubic}, we obtain
\begin{align*}
Q_3 + [\dG(-\Delta)+Q_4,\cB_c]
	=  (1-\theta_M) \int_{\Lambda^2} V_\ell(x-y) a_x^*a_y^*a_x + \hc + \mathcal{E}_c^{(\dG(-\Delta))}+\mathcal{E}_c^{(Q_4)}.
\end{align*}
As in \eqref{eq:E_theta_bound}, we use that 
$ (1-\theta_M) \1^{\{\mathcal N < M/2 \}} =0$
to obtain
$$
\pm (1-\theta_M) n^{1/2} \int_{\Lambda^2} V_\ell(x-y) a_x^*a_y^*a_x + \hc \leq \delta Q_4 + C \delta^{-1} \frac{n}{\ell} \frac{\cN^2}{M}.
$$
We may now combine \Cref{lem_number_preserved_cubic} and \Cref{lem_kinetic_preserved} with \eqref{rem:kin_Q4_Tc} to bound ${\rm (III)}_{c_1}$ as
\begin{align*}
\pm {\rm (III)}_{c_1} 
	&\leq   C \delta \left(Q_4 + \frac{n}{\ell}(\cN+1) + \sigma \dG(-\Delta) \right)  + C \delta^{-1} \frac{n}{\ell} \frac{(\cN+1)^2}{M}\\
	& \quad  + C \sigma \left(Q_4 + \dG(-\Delta) + \frac{n}{\ell}(\cN+1) \right) 
		\\ 
	& \leq C \sigma \left(Q_4 + \dG(-\Delta) + \frac{n}{\ell}(\cN+1) \right)  + C \delta \left(Q_4 + \frac{n}{\ell}(\cN+1) \right)   + C \delta^{-1} \frac{n}{\ell} \frac{(\cN+1)^2}{M}.
\end{align*}

\medskip
\noindent
\textbf{Estimating $ {\rm (III)}_{c_2}$.} We apply Lemmata \ref{lem_number_preserved_cubic} and \ref{lem_kinetic_preserved} to  (\ref{eq:E_1}) and find
\begin{align*}
& \pm  e^{-\cB_c} \mathcal{E}_1 e^{\cB_c} \leq C (\delta+\ell^{-1} \lambda^2 + \varepsilon n^{-1}) \left( Q_4 + \frac{n}{\ell} (\cN+1) + \sigma \dG(-\Delta) \right) \\
&\quad + C\delta^{-1}  \left( \frac{(\cN + 1)}{\ell} + \frac{(\cN+1)^2}{n\ell} +   \lambda \left(\frac{n}{\ell}\right)^3 \right) (\cN + 1) + C \lambda^\frac{1}{2} \big( (\frac{n}{\ell})^2 + \frac{n}{\ell} \big) (\cN + 1)
\\
& \quad
+ C n^\frac{1}{2} \frac{(\cN+1)^\frac{3}{2}}{\ell}
+ C \sigma \left(Q_4 + \dG(-\Delta) + \frac{n}{\ell}(\cN+1) \right)
\\
&\quad + C \varepsilon^{-1} \frac{n}{\ell} + C \left(\frac{n}{\ell}\right)^2 \log \ell
\\
&\leq \frac{1}{4}Q_4 + C \delta \left( Q_4 + \frac{n}{\ell} (\cN+1) \right) + \delta^{-1} C \left( \frac{(\cN + 1)}{\ell} + \frac{(\cN+1)^2}{n\ell} +   \lambda \left(\frac{n}{\ell}\right)^3 \right) (\cN + 1) 
\\
& \quad
+ C \lambda^\frac{1}{2} \left( \left(\frac{n}{\ell}\right)^2 + \frac{n}{\ell} \right) (\cN + 1)
+ C n^\frac{1}{2} \frac{(\cN+1)^\frac{3}{2}}{\ell}
+ C \sigma \left(Q_4 + \dG(-\Delta) + \frac{n}{\ell}(\cN+1) \right) 
\\
&\quad + C \left(\left(\frac{n}{\ell}\right)^2 \log \ell + \frac{n}{\ell} \right).
\end{align*}
Here  we used $\delta \leq 1$ for the second inequality and we set $\varepsilon = (5C)^{-1}$ so that for $\ell$ large enough we have $C(\ell^{-1}\lambda^2 + \varepsilon n^{-1}) \leq 1/4$. 

We now have obtained the necessary bounds on ${\rm (III)}_c={\rm (III)}_{c_1}+{\rm (III)}_{c_2}$. Combining them  with the estimates of ${\rm (I)}_c$ and ${\rm (II)}_c$ given by Lemmata \ref{prop_quadratic_cubictrafo} and \ref{prop_Q2_commutator_cubic}, respectively, and using $M \leq n$, we conclude the proof of  \Cref{lem_cubictrafo}.
\end{proof}


\section{The Second Quadratic Transformation} \label{sec:last_trafo}
In this section, we diagonalize  explicitly the Bogoliubov Hamiltonian in \eqref{eq:H_mom-intro},
\begin{align*}
\mathbb{H}_{\rm Bog} = \sum_{p\neq 0 } \left(p^2+8\pi\mathfrak{a} \frac{n}{\ell}\right) a_p^*a_p + \frac{1}{2} \sum_{p \neq 0 } 
 n\widehat{\epsilon}_{\ell,\lambda}(p)(a_p^*a_p^*+a_pa_p) + \frac{1}{2} \sum_{p \neq 0}
\frac{|n\widehat{\epsilon}_{\ell,\lambda}(p)|^2}{2p^2}.
\end{align*}
We define 
\begin{align} 
\cB_2 = \frac{1}{2}\sum_{p \neq 0 } \vphi_p (a_p^*a_p^*-a_pa_p), \label{eq:B2_def}
\end{align}
where 
\begin{align}
\vphi_p &= \sinh^{-1}(\nu_p), \; \nu_p = -\sqrt{\frac{1}{2}\left(\frac{A_p}{\sqrt{A_p^2-B_p^2}}-1\right)}, \label{eq:def_nup} \\
A_p &= p^2+8\pi\mathfrak{a}n \ell^{-1},\; B_p=n\widehat{\epsilon}_{\ell,\lambda}(p). \label{eq:def_A_B}
\end{align}
Note that we have
\[ \label{eq_vep(0)}
|B_p-8\pi \ao n\ell^{-1}| = |n\widehat{\epsilon}_{\ell,\lambda}(p) - n \widehat{\epsilon}_{\ell,\lambda}(0)| \le C \min\{1,\lambda^2p^2\} n\ell^{-1},
\]
which follows from the radial symmetry of $\epsilon$ and $|\epsilon_{\ell,\lambda}(x)| \leq C \lambda^{-3} \ell^{-1} \mathds{1}_{|x|\leq \lambda} $, see \eqref{eq:w-pointwise}.
In particular for $\lambda^2 n \ell^{-1}$ small enough we have $A_p >|B_p|$ for all $p\ne 0$, and the formula of $\nu_p$ in \eqref{eq:def_nup} is therefore well-defined.


The main result of this section is the following lemma.
\begin{lemma} \label{lem_last_trafo}
Let $\lambda,n,\ell,\sigma,\mathcal{E}_c$ as in \Cref{lem_cubictrafo}, assume that  $\lambda$ is small enough and let $\mathcal{W} = e^{\cB_1}e^{\cB_c}e^{\cB_2}$. Then 
\begin{align*}
  \mathcal{W}^*&\cH \mathcal{W}= E_{n,\ell} + \dd \Gamma(E_{\rm Bog})  +  e^{-\cB_2}\left(Q_4 + \mathcal{E}_c + \mathcal{E}_2\right) e^{\cB_2}
\end{align*}
on $\cF_+$ with 
\begin{align}\label{eq:mu}
\dd \Gamma(E_{\rm Bog})  &=	 \sum_{p \neq 0} \sqrt{p^4+16\pi \ao n \ell^{-1}p^2} \, a_p^*a_p, \nn\\
 E_{n,\ell} 
 	&=  4\pi \ao n^2 \ell^{-1}  + \frac{1}{2}\sum_{p \neq 0} \left[ \sqrt{p^4+16\pi \ao n \ell^{-1}p^2} - p^2 - 8\pi\mathfrak{a}\frac{n}{\ell} + \frac{(8\pi \mathfrak{a}n\ell^{-1} )^2}{2p^2} \right] \nn\\
\end{align}
and
\begin{align*}
\pm \mathcal{E}_2 &\leq C \lambda \left( \left (\frac{n}{\ell}\right)^\frac{1}{2} + 1 \right) \cN + C \lambda \left(\frac{n}{ \ell} \right)^3 + C \lambda^2 \left(\frac{n}{\ell}\right)^{5/2}.
\end{align*}
\end{lemma}


%
%
%
%
%
%

Let us first state the following lemma, which will be proved at the end of this section.

\begin{lemma}
	\label{prop_diagonalization}
	Let $\lambda \left(\frac{n}{\ell}\right)^2 \leq 1$, $2 R/\ell < \lambda$ and $\lambda$ be small enough. Then we have
\begin{align}
e^{\pm \cB_2} (\cN+1) e^{\mp \cB_2} &\leq C \left( \left(\frac{n}{\ell}\right)^\frac{1}{2}+1\right)\cN + C\left(\frac{n}{\ell}\right)^\frac{3}{2}
\label{eq_number_conserved_last}
\intertext{and}
e^{-\cB_2} \dG(-\Delta) e^{\cB_2} &\leq C  \left(\left(\frac{n}{\ell}\right)^\frac{1}{2} +1\right) \dG(-\Delta) + C \lambda^{-1}\left(\frac{n}{\ell}\right)^2 .\label{eq:Kinetic_conserved_last}
\end{align}
\end{lemma}

\begin{proof}[Proof of Lemma \ref{lem_last_trafo}]
From the assumptions on  $\lambda (\frac{n}{\ell})^2$ and $\lambda$ we conclude that $\lambda^2 (\frac{n}{\ell})$ is sufficiently small for  $\vphi_p$ in (\ref{eq:def_nup})  to be well defined. Given  \eqref{eq:H_mom-intro} and (\ref{eq:B2_def}), the action of the transformation $e^{\mathcal B_2}$ on $\mathbb{H}_{\rm Bog}$ is standard, see e.g. \cite[Section 3]{Seiringer-11}. It gives 
\[ \label{eq_bog_diagonalized}
e^{-\cB_2} \mathbb{H}_{\rm Bog}e^{\cB_2}  = e_{\rm Bog} + \sum_{p \neq 0} e_p  a_p^*a_p 
\]
with 
\begin{align*} 
e_p = \sqrt{A_p^2-B_p^2}, \quad e_{\rm Bog} = \frac{1}{2} \sum_{p \neq 0} \left[ \sqrt{A_p^2-B_p^2}-A_p\right] + \frac{1}{2}\sum_{p \neq 0} \frac{(n\widehat{\epsilon}_{\ell,\lambda}(p))^2}{2p^2},
\end{align*}
where we recall that $A_p$ and $B_p$ are defined in (\ref{eq:def_A_B}). We will use the following lemma that we show after the proof of Lemma \ref{lem_last_trafo}.

\begin{lemma} \label{lem:bog_approx}
	Let $\lambda \left(\frac{n}{\ell}\right)^2 \leq 1$, $2 R/\ell < \lambda$ and $\lambda$ be small enough. We have the uniform bound
\begin{align} 
\sup_{p\in \R^3} \left| e_p - \sqrt{p^4 + 16\pi \ao n \ell^{-1}p^2 }\right| \leq C \lambda^2 \frac{n^2}{\ell^2} \label{eq_ep_approximation}
\end{align}
as well as
\begin{align}
\left| e_{\rm Bog} - \frac{1}{2} \sum_{p \neq 0} \left[ \sqrt{p^4 + 16\pi \ao n \ell^{-1}p^2} - p^2 - 8\pi\mathfrak{a}\frac{n}{\ell} + \frac{(8\pi\mathfrak{a}n\ell^{-1} )^2}{2p^2} \right] \right|
\leq C \lambda \frac{n^3}{\ell^3}. \label{eq_ebog_approximation}
\end{align}
\end{lemma}

Starting from (\ref{eq:H_Tc}), using the identity (\ref{eq_bog_diagonalized}), the estimates of Lemma \ref{lem:bog_approx} and the bounds of \Cref{prop_diagonalization}, readily imply the statement of \Cref{lem_last_trafo}. 
\end{proof}

We end this section with the proofs of Lemmata \ref{lem:bog_approx} and \ref{prop_diagonalization}.

\begin{proof}[Proof of Lemma \ref{lem:bog_approx}]
Let us start by showing  \eqref{eq_ep_approximation}. Using 
$|\sqrt{1+x}-1| \leq |x|$ for all $x\geq -1$ 
and \eqref{eq_vep(0)} we find
\begin{align*}
&\Big| e_p - \sqrt{p^4+16\pi \ao n \ell^{-1}p^2} \Big|
\\
&\leq  \left(p^4+16\pi \ao n \ell^{-1}p^2\right)^{-\frac{1}{2}} |8\pi\mathfrak{a}\frac{n}{\ell}+n\widehat{\epsilon}_{\ell,\lambda}(p)||8\pi\mathfrak{a}\frac{n}{\ell}-n\widehat{\epsilon}_{\ell,\lambda}(p)|
\\
&\leq C \left(p^4+16\pi \ao n \ell^{-1}p^2\right)^{-\frac{1}{2}} \lambda^2 p^2 \frac{n^2}{\ell^2}
\leq C \lambda^2 \frac{n^2}{\ell^2},
\end{align*}
which shows \eqref{eq_ep_approximation}.
To prove  (\ref{eq_ebog_approximation}) note that
\begin{align*}
\tau &:= 2e_{\rm Bog} - \left[\sum_{p \neq 0}\sqrt{p^4+16\pi \ao n \ell^{-1}p^2} - p^2 - 8\pi\mathfrak{a}\frac{n}{\ell}+ \frac{(8\pi\mathfrak{a}n/\ell )^2}{2p^2} \right]
\\
&= \sum_{p \neq 0} p^2 \Bigg(\sqrt{1+\frac{16\pi\mathfrak{a}n/\ell }{p^2} + \frac{(8\pi\mathfrak{a}n/\ell )^2-(n\widehat{\epsilon}_{\ell,\lambda}(p))^2}{p^4}} \\
&\qquad\qquad\qquad\qquad -\sqrt{1+\frac{16\pi\mathfrak{a}n/\ell }{p^2} -\frac{(8\pi\mathfrak{a}n/\ell )^2-(n\widehat{\epsilon}_{\ell,\lambda}(p))^2}{2p^4}}\Bigg).
\end{align*}
The Taylor formula readily gives
$$
|v(x)-v(y)-v'(0)(x-y)| \leq C \|v''\|_\infty |x-y|(|x|+|y|)
$$
for $v \in C^2(U), U \subset \R$ open. Applying this inequality to 
$$
v(z) = \sqrt{1+z} \text{ with } x=\frac{16\pi\mathfrak{a}n/\ell }{p^2} + \frac{(8\pi\mathfrak{a}n/\ell )^2-(n\widehat{\epsilon}_{\ell,\lambda}(p))^2}{p^4}, \; y=\frac{16\pi\mathfrak{a}n/\ell }{p^2}
$$ 
and using \eqref{eq_vep(0)} yields
\begin{align*}
|\tau| &\leq  C \sum_{p \neq 0} p^2 \frac{|(8\pi\mathfrak{a}n/\ell )^2-(n\widehat{\epsilon}_{\ell,\lambda}(p))^2|}{p^4} \left(2 \frac{16\pi\mathfrak{a}n/\ell }{p^2} + \frac{|(8\pi\mathfrak{a}n/\ell )^2-(n\widehat{\epsilon}_{\ell,\lambda}(p))^2|}{p^4}\right)
\\
&\leq C \sum_{p \neq 0} p^2 \frac{n/\ell }{p^4} |8\pi\mathfrak{a} \frac{n}{\ell} -n\widehat{\epsilon}_{\ell,\lambda}(p)| \left(\frac{n/\ell  + \lambda^2(n/\ell )^2}{p^2}\right)
\\
&\leq C (\frac{n}{\ell})^2 \sum_{p\neq 0} |8\pi\mathfrak{a} \frac{n}{\ell} -n\widehat{\epsilon}_{\ell,\lambda}(p)| p^{-4}
\\
&\leq C (\frac{n}{\ell})^2 \sum_{p\neq 0} \min\{1,\lambda^2p^2\} \frac{n}{\ell} p^{-4}
\leq C \lambda (\frac{n}{\ell})^3,
\end{align*}
which concludes the proof of \eqref{eq_ebog_approximation}.
\end{proof}
\begin{proof}[Proof of Lemma \ref{prop_diagonalization}]
The action of the quadratic transformation $\cB_2$  on the creation and annihilation operators is  given by
$$
e^{\mp \cB_2} a^*_p e^{\pm \cB_2} = \cosh(\vphi_p) a_p^* \pm  \sinh(\vphi_p) a_p.
$$ 
Thus
\begin{align}\nonumber
e^{\pm \cB_2} \cN e^{\mp\cB_2} 
	&= \sum_{p \neq 0} \left(\cosh(\vphi_p) a_p^* \pm  \sinh(\vphi_p) a_p\right)\left(\cosh(\vphi_p) a_p \pm  \sinh(\vphi_p) a_p^* \right) \\ \nonumber
	&=\sum_{p \neq 0} \left[ \left(\cosh(\vphi_p)^2 + \sinh(\vphi_p)^2\right) a_p^*a_p \pm \cosh(\vphi_p)\sinh(\vphi_p) (a_p^*a_p^*+ a_pa_p) + \sinh(\vphi_p)^2\right] \\
	&\leq 2 \sum_{p \neq 0} \left[ (\cosh(\vphi_p)^2 + \sinh(\vphi_p)^2)a_p^*a_p + \sinh(\vphi_p)^2\right]\,. \label{eab}
\end{align}
To verify  \eqref{eq_number_conserved_last} it remains to show that 
\begin{align}
\cosh(\vphi_p)^2+\sinh(\vphi_p)^2 &\leq C \left( \left(\frac{n}{\ell}\right)^{\frac{1}{2}} + 1\right), \label{eq_vphi_sup}
\\
\sum_{p \neq 0} \sinh(\vphi_p)^2 & \leq C \left(\frac{n}{\ell}\right)^\frac{3}{2}. \label{eq_vphi_HS}
\end{align}
Let us prove \eqref{eq_vphi_sup}. From (\ref{eq:def_nup}) we have
\begin{align*}
\cosh(\vphi_p)^2  + \sinh(\vphi_p)^2 &= 2 \sinh( \vphi_p)^2 + 1 
\\
& = |p|^{-1} \frac{p^2 + 8\pi\mathfrak{a} n/\ell }{\sqrt{p^2 + 16 \pi \mathfrak{a} n/\ell  +  \frac{(8\pi\mathfrak{a}n/\ell )^2-(n\widehat{\epsilon}_{\ell,\lambda}(p))^2}{ p^2}}}
\\
&\leq 
|p|^{-1} \frac{p^2 + 8\pi\mathfrak{a} n/\ell }{\sqrt{p^2 + 8 \pi \mathfrak{a} n/\ell  }} = \sqrt{1 + 8\pi\mathfrak{a}\frac{n}{\ell}p^{-2}} \leq C \left( \left(\frac{n}{\ell}\right)^\frac{1}{2} + 1 \right).
\end{align*} 
Here we used $ p^{-2}\left((8\pi\mathfrak{a}n/\ell )^2-(n\widehat{\epsilon}_{\ell,\lambda}(p))^2\right)  \geq - 8\pi\ao n\ell^{-1}$, which follows from $\eqref{eq_vep(0)}$.

To prove \eqref{eq_vphi_HS} we use again (\ref{eq:def_nup}) and divide the sum into momenta $|p|$ less or bigger than $(n/\ell )^\frac{1}{2}$. For small momenta we have
\begin{align}\nonumber
\sum_{0<|p| \leq (\frac{n}{\ell})^\frac{1}{2}} \sinh(\vphi_p)^2  &\leq C \sum_{0<|p| \leq (\frac{n}{\ell})^\frac{1}{2}} \left( \sqrt{1+8\pi \mathfrak{a}n/\ell p^{-2}} + 1 \right)
\\ & \leq C \sum_{0<|p|\leq (\frac{n}{\ell})^\frac{1}{2}} \left( 1+8\pi \mathfrak{a}n/\ell p^{-2} \right)
\leq C \left(\frac{n}{\ell}\right)^\frac{3}{2}.
\end{align}
For large momenta we use 
$|\sqrt{1+z}-1-z|\leq z^2$ for $z \geq -1$ as well as 
and $|\sqrt{1+x}-\sqrt{1+y}| \leq C|x-y|$ for $x,y \geq -\frac{1}{2}$ to obtain
\begin{align*}
\sum_{|p| > (\frac{n}{\ell})^\frac{1}{2}} \sinh(\vphi_p)^2 
&\leq \sum_{|p| > (\frac{n}{\ell})^\frac{1}{2}} p^{-2} \left| p^2 + 8\pi\mathfrak{a} \frac{n}{\ell} - p^2 \sqrt{1 + \frac{16\pi\mathfrak{a} n/\ell }{p^2} + \frac{(8\pi\mathfrak{a} n/\ell )^2- (n\widehat{\epsilon}_{\ell,\lambda}(p))^2}{p^4} }\right|
\\
&\leq \sum_{|p| > (\frac{n}{\ell})^\frac{1}{2}} \left| 1 + \frac{8\pi\mathfrak{a}n/\ell }{p^2} - \sqrt{1+\frac{16\pi\mathfrak{a}n/\ell }{p^2}}\right| 
\\
&  \quad + \sum_{|p| > (\frac{n}{\ell})^\frac{1}{2}} \left|\sqrt{1 + \frac{16\pi\mathfrak{a} n/\ell }{p^2} + \frac{(8\pi\mathfrak{a} n/\ell )^2- (n\widehat{\epsilon}_{\ell,\lambda}(p))^2}{p^4} } - \sqrt{1+\frac{16\pi\mathfrak{a}n/\ell }{p^2}}\right|
\\
&\leq C \frac{n^2}{\ell^2} \sum_{|p| > (\frac{n}{\ell})^\frac{1}{2}} p^{-4} 
\leq C \left(\frac{n}{\ell}\right)^\frac{3}{2}.
\end{align*}
This concludes the proof of \eqref{eq_number_conserved_last}.
The bound \eqref{eq:Kinetic_conserved_last} is proved similarly. Proceeding as in \eqref{eab} we obtain
\begin{align*}
e^{-\cB_2} \dG(-\Delta)e^{\cB_2} &\leq  2 \sum_{p \neq 0 } \left[p^2 (\cosh(\vphi_p))^2+\sinh(\vphi_p))^2)a_p^*a_p + p^2 \sinh(\vphi_p)^2\right]
\\
&\leq C \left( \left(\frac{n}{\ell}\right)^\frac{1}{2} + 1\right) \dG(-\Delta) + C \left(\frac{n}{\ell}\right)^2 \lambda^{-1}
\end{align*}
given that 
\begin{equation} \label{eq:p^2sinh_bound}
\sum_{p \neq 0} p^2 \sinh(\vphi_p)^2 \leq C \left(\frac{n}{\ell}\right)^2 \lambda^{-1}\,,
\end{equation}
which remains to be shown.   We have
\begin{align*}
2 p^2\sinh(\vphi_p)^2 = p^2 \frac{1 + 8\pi\mathfrak{a} (n/\ell)  p^{-2} - \sqrt{1 + 16\pi\mathfrak{a}(n/\ell)  p^{-2} + \dfrac{(8\pi\mathfrak{a}(n/\ell)  )^2 - (n\widehat{\epsilon}_{\ell,\lambda}(p))^2}{p^4}}}{\sqrt{1 + 16\pi\mathfrak{a}(n/\ell)  p^{-2} + \dfrac{(8\pi\mathfrak{a}(n/\ell) )^2-(n\widehat{\epsilon}_{\ell,\lambda}(p))^2}{p^4}}}.
\end{align*}
Let us define 
$$v(x) = 1 + 8\pi\mathfrak{a} (n/\ell)  x - \left(1 + 16\pi\mathfrak{a}(n/\ell)  x + ((8\pi\mathfrak{a}(n/\ell)  )^2 - (n\widehat{\epsilon}_{\ell,\lambda}(p))^2) x^2\right) ^{\frac{1}{2}}.
$$
Taylor expanding at $x=p^{-2}$, we obtain
$$v(x) = \frac{(n\widehat{\epsilon}_{\ell,\lambda}(p))^2}{2}x^2 + R(x)$$ with 
$$|R(x)| \leq \sup_{0 \leq y\leq p^{-2}} |v'''(y)| x^2 \leq C \left(\frac{n}{\ell}\right)^3 x^2.$$
Thus
\begin{align*}
2 p^2\sinh(\vphi_p)^2 \leq \frac{(n\widehat{\epsilon}_{\ell,\lambda}(p))^2}{2p^2} + C\left(\frac{n}{\ell}\right)^3 p^{-4},
\end{align*}
which immediately implies
\begin{align*}
\sum_{p\ne 0} p^2 \sinh(\vphi_p)^2 \leq \sum_{p\ne 0} \frac{(n\widehat{\epsilon}_{\ell,\lambda}(p))^2}{2p^2} + C\left(\frac{n}{\ell}\right)^3.
\end{align*}
From the assumptions $\lambda \left(\frac{n}{\ell}\right)^2 \leq 1$ and $\lambda$ small enough, we find $(n/\ell)^3\le \lambda^{-1} (n/\ell)^2$. 

As we argued in the proof of \eqref{eq:Riemann-error-0}, we have $\widehat \epsilon_{\ell,\lambda}(p)=  \ao \ell^{-1}  \widehat f(\lambda p)$ for a fixed function $f\in C_c^\infty(\R^3)$, and hence
\begin{align} \label{eq:vep_Landau}
\sum_{p\ne 0} \frac{(n\widehat{\epsilon}_{\ell,\lambda}(p))^2}{2p^2} 
&\leq \left| \sum_{p\ne 0} \frac{(n\widehat{\epsilon}_{\ell,\lambda}(p))^2}{2p^2} - \frac{1}{(2\pi)}  \int_{\R^3} \frac{(n\widehat{\epsilon}_{\ell,\lambda}(p))^2}{2p^2} \right| +  \frac{1}{(2\pi)}  \int_{\R^3} \frac{(n\widehat{\epsilon}_{\ell,\lambda}(p))^2}{2p^2} \nn
\\
&\leq C \left(\frac{n}{\ell} \right)^2 + C \lambda^{-1} \left(\frac{n}{\ell} \right)^2 \leq C \lambda^{-1} \left(\frac{n}{\ell} \right)^2.
\end{align}
This shows \eqref{eq:p^2sinh_bound} and consequently completes the proof of Lemma \ref{prop_diagonalization}.
\end{proof}


\section{A-Priori Estimates for Gibbs States} \label{sec:a-priori} 

Recall the definition (\ref{eq:def_Hnl})  of $H_{n,\ell}$, and the one of $U$ in Section \ref{sec:excitation}. 
We introduce the Gibbs state on $\mathcal{F}_+^{\leq n}$
\begin{equation} \label{eq:Gibbs-Hnl}
\Gamma = e^{-\frac{\mathbb{H}_{n,\ell}}{T\ell^2}} / \tr e^{-\frac{\mathbb{H}_{n,\ell}}{T\ell^2}}, \quad \mathbb{H}_{n,\ell} = U H_{n,\ell} U^*
\end{equation}
It  minimizes the Gibbs variational principle and yields the free energy on the box $\Lambda_\ell$ as in \eqref{eq:def_energy_small_box}.

The analysis of $\cH$ in the previous sections will be used in \Cref{sec:proof_theo2} to control $\Gamma$ on the sector of few excited particles. This section, on the other hand, provides rough a priori estimates on the kinetic and interaction operator, which are helpful to control the sector of high particle number. Moreover, we derive  complete BEC for $\Gamma$ given that the thermal contribution to the free energy is subleading. We shall see that this is the case as long as $n$ is not too small. 

\subsection{Rough Kinetic and Interaction Energy Estimates}

\begin{lemma}	\label{lem:bec_T2} 
Let $0 \leq n \leq C \rho\ell^3 $ and $\ell = \ao (\rho\ao^3)^{-1/2-\kappa}$. Then we have
\begin{align}\label{eq:K-Q4-rough}
\mathbb{H}_{n,\ell}  \lesssim \left(\dd\Gamma(-\Delta) + Q_4 + n^2 \ao \ell^{-1} \right) \lesssim\mathbb{H}_{n,\ell} + n^2 \ao \ell^{-1}
\end{align}
on $\mathcal{F}_+^{\leq n}$. 
\end{lemma}

\begin{proof} Recall from Lemma~\ref{lem_excitation_Hamil} that $\mathbb{H}_{n,\ell}= \mathds{1}_+^{\leq n} \mathcal{H} \mathds{1}_+^{\leq n}$,
where $\mathcal{H}$ is an operator on $\cF$ defined in \eqref{eq_excitation_Hamil}. 
 We shall show that all the terms  $n^2V_\ell^{0000}$, $Q_1$, $H_2^{(U)}$, $Q_2$, $Q_3^{(U)}$ and $\mathcal{E}^{(U)}$, when restricted to $\mathcal{F}_+^{\leq n}$, are bounded by $\eps Q_4 + C_\eps n^2 \ell^{-1}$  with $\eps>0$ arbitrarily small . 

First, the constant $ n^2V_{\ell}^{0000}$ is bounded by $n^2 \|V_\ell\|_{L^1} \le C n^2\ell^{-1}$.
Next, from  the bound \eqref{eq:E_error_excitation} in \Cref{lem_excitation_Hamil}, with $\eps$ replaced by $n\eps $ , we have 
\begin{align} \label{eq:rough_EU}
\pm  \mathcal{E}^{(U)}  \le C n^\frac{1}{2}(\cN+1)^\frac{3}{2} \ell^{-1}  +  \varepsilon Q_4 + C \eps^{-1} \ell^{-1}
\end{align}
 on $\cF$ for all $\eps>0$. For $H_2^{(U)}$  we have $\cN nV_{\ell}^{0000} \le C \cN n \ell^{-1}$ and 
\begin{align*}
\pm \int_{\Lambda^2} nV_\ell(x-y) a_x^*a_x &\le n \|V_\ell\|_{L^1} \cN \le C n\ell^{-1}\cN,\\
\pm  \int_{\Lambda^2} nV_\ell(x-y) a_x^*a_y &\le  \frac 12\int_{\Lambda^2} nV_\ell(x-y) (a_x^*a_x + a^*_ya_y) \le C n\ell^{-1}\cN,\\
\pm  \left( \frac{1}{2}\int_{\Lambda^2} V_{\ell}(x-y) a_x^* a_y^*\cN + \hc \right)  &\le \frac{\eps}2 \int_{\Lambda^2} V_{\ell}(x-y) a_x^* a_y^* a_x a_y + \frac{1}{2\eps} \int_{\Lambda^2} V_{\ell}(x-y) \cN^2\\
&\le \eps Q_4 + C \eps^{-1} \ell^{-1} \cN^2
\end{align*}
so that
\begin{equation} \label{eq:rough_H2U}
\pm H_2^{(U)} \leq C n\ell^{-1}\cN + \eps Q_4 + C \eps^{-1} \ell^{-1} \cN^2
\end{equation}
on $\cF$ for all $\eps>0$. Moreover, by the Cauchy--Schwarz inequality 
\begin{align}
\pm Q_1 
&\le  \int_{\Lambda^2} n V_{\ell}(x-y)  (\cN+1)^{-1/4} a_x^* a_x (\cN+1)^{-1/4} + \int_{\Lambda^2} n^2 V_{\ell}(x-y) (\cN+1)^{1/2} \nn \\
&\le C n^{3/2} \ell^{-1}(\cN+1)^{1/2} ,\label{eq:rough_Q1} \\
\pm Q_2 
&\le \frac{\eps}{2} \int_{\Lambda^2} V_{\ell}(x-y) a_x^* a_y^* a_x a_y + \frac{\eps^{-1}}{2} n^2 \int_{\Lambda^2} V_{\ell}(x-y) \nn
\\
& \le \eps Q_4 + C \eps^{-1}n^2 \ell^{-1}, \label{eq:rough_Q2} \\
\pm Q_3^{(U)} 
&\le \frac{\eps}{2} \int_{\Lambda^2} V_{\ell}(x-y) a_x^* a_y^* a_x a_y + 2\eps^{-1}  n \int_{\Lambda^2} V_{\ell}(x-y)  a_x^* a_x \nn 
\\
&\le \eps Q_4 + C \eps^{-1} n \ell^{-1}\cN \label{eq:rough_Q3}
\end{align}
on $\cF$ for all $\eps>0$. Restricting these bounds to $\mathcal{F}_+^{\leq n}$ where $\cN\le n$ we find that
$$
\pm (\mathbb{H}_{n,\ell} - \dd \Gamma(-\Delta)-Q_4) \le  \eps  Q_4 + C \eps^{-1} n^2 \ell^{-1} 
$$
on $\mathcal{F}_+^{\leq n}$ for all $0<\eps \leq 1$, which implies \eqref{eq:K-Q4-rough}.
\end{proof}


\subsection{Bose--Einstein Condensation}\label{sec:BEC-apriori}

\begin{lemma}	\label{lem:bec_T}
Let $\Gamma$ be as in \eqref{eq:Gibbs-Hnl} with $ (\rho\ao^3)^{1/4+\nu/2} (\rho\ell^3) \leq n  \leq C(\rho\ell^3) $, $\ell = \ao (\rho\ao^3)^{-1/2-\kappa}$ with $0<\kappa<1/34$ and $0\leq T \leq (\rho \ao)(\rho\ao^3)^{-\nu}$ with $0<\nu  <1/12-5\kappa/3$. Then we have
\begin{align}
\tr ( \mathcal N_+ \Gamma)  \leq C n (\rho \ao^3)^{\gamma}, \label{eq:est_N}
\end{align}
with $\gamma = \min \{1/10  - 2\kappa -6\nu/5, 1/17 - 2\kappa  \}>0$.
\end{lemma}

In the dilute limit $\rho\ao^3\to 0$, \Cref{lem:bec_T} implies that $\tr ( \mathcal N_+ \Gamma)$ is much smaller than $n$, the total number of particles. This is equivalent to complete condensation of $\Gamma$, namely all but  $o(n)$ particles occupy the zero-momentum mode.  

\begin{proof} From the method in \cite{LieSei-02} (see also \cite[Lemma 5.2]{LieSeiSolYng-05}) we have the a priori knowledge of condensation on $ L^2_s(\Lambda^n)$
\begin{align}
	\label{eq:input}
H_{n,\ell} \geq 4\pi \ao\ell^{-1} n^2 +  C^{-1} \cN_+ - C n \rho\ao \ell^2 (\rho\ao^3)^{1/17}.
\end{align}
Since $U\cN_+ U^*=\cN_+$, the same bound holds with $H_{n,\ell}$ replaced by $\mathbb{H}_{n,\ell} = U H_{n,\ell} U^*$. 
We obtain
\begin{align*} 
C^{-1} \tr (\cN_+ \Gamma) &\leq
\tr (\mathbb{H}_{n,\ell}  \Gamma ) - 4 \pi \ao \ell^{-1} n^2 +  C n \rho\ao \ell^2 (\rho\ao^3)^{1/17}. \nn
\end{align*}
Therefore, it is enough to show that 
\begin{align} \label{eq:energy_gibbs_state}
\tr (\mathbb{H}_{n,\ell}  \Gamma ) - 4 \pi \ao \ell^{-1} n^2 \leq C n (\rho\ao^3)^{1/10 - 2\kappa - 6\nu/5}.
\end{align} 
In order to prove (\ref{eq:energy_gibbs_state}), we will use the following upper bound on the ground state energy \cite[Theorem 2.2]{LieSeiSolYng-05}
\begin{equation} \label{eq:spectrum_bound}
\ell^2 F_{\ell}(n) \leq \inf \sigma( \mathbb{H}_{n,\ell} ) \leq 4 \pi \ao n^2 \ell^{-1} + C n (\rho \ao^3)^{1/3}.
\end{equation}
From the Gibbs variational principle, we have that for all $\varepsilon >0$
\begin{align*}
\varepsilon \tr \left( \dG(-\Delta)\Gamma\right) - T\ell^2  S(\Gamma) & \geq  -T\ell^2\log \tr e^{-\tfrac{\varepsilon}{T\ell^2}\dd\Gamma(-\Delta)}
	\geq - C T^{5/2} \ell^5 \varepsilon^{-3/2},
\end{align*}
which follows from the bound
\begin{align} \label{eq:unperturbed_energy_bound}
\beta^{-1} \log \tr e^{- \beta \dd\Gamma(-\Delta)} = - \beta^{-1}  \sum_{p\in \pi \mathbb{N}_0^{3} \setminus \{0\}} \log \left(1-e^{-\beta p^2} \right) \leq C \beta^{-5/2}
\end{align}
for some $C>0$ and all $\beta>0$.

We can now prove (\ref{eq:energy_gibbs_state}). From the Gibbs variational principle we have for all $0<\varepsilon \leq 1$
\begin{align*}
\tr (\mathbb{H}_{n,\ell} \Gamma) 
	&= (1+\varepsilon) \left(\tr (\mathbb{H}_{n,\ell} \Gamma) - T\ell^2S(\Gamma) \right) -   \left(\varepsilon \tr (\mathbb{H}_{n,\ell} \Gamma) - (1 + \varepsilon) T\ell^2S(\Gamma) \right) \\
	&\leq (1+\varepsilon) \ell^2 F_{\ell}(n) -  \left(\varepsilon \tr (\dd\Gamma(-\Delta) \Gamma) - 2T\ell^2S(\Gamma) \right) \\
	&\leq (1+\varepsilon) \inf \sigma(\mathbb{H}_{n,\ell}) - \left(\varepsilon \tr (\dd\Gamma(-\Delta) \Gamma) - 2T\ell^2S(\Gamma) \right) \\
	&\leq (1+\varepsilon)\left(4 \pi \ao n^2 \ell^{-1} + C n (\rho \ao^3)^{1/3} \right) + C \varepsilon^{-3/2} T^{5/2} \ell^5.
\end{align*}
Taking $\varepsilon = (\rho \ao^3)^{1/10-6\nu/5}$, using that $(\rho\ao^3)^{1/4+\nu/2} (\rho\ell^3) \leq n$ and our assumptions on $\ell$ and $T$, we obtain (\ref{eq:energy_gibbs_state}) from which \Cref{lem:bec_T} follows.
\end{proof}

\section{Proof of Theorem \ref{theo:free_energy_small_box}} \label{sec:proof_theo2}

In this section we shall give the proof of  Theorem \ref{theo:free_energy_small_box}. We denote $Y= \rho \ao^3$ and set 
\begin{align} \label{eq:def-M0} M_0 = n^{1-80\kappa}, M = n^{1-68\kappa}, \lambda = Y^{10\kappa},\delta = Y^{3\kappa}
\end{align}
with $\kappa=1/1000$. Recall that $T\le \rho \ao (\rho \ao^3)^{-\nu}$ and $\ell=\ao/(\rho \ao^3)^{1/2+\kappa}$ as in \eqref{eq:def_ell} with $\nu = \kappa/5$.
Moreover, let us first focus on the case  $ (\rho\ao^3)^{1/4+\nu/2} (\rho\ell^3) \leq n  \leq C(\rho\ell^3) $, which allows us to use condensation in the sense of \Cref{lem:bec_T}. The case of smaller $n$ will be considered afterwards.

\paragraph{Case $(\rho\ao^3)^{1/4+\nu/2} (\rho\ell^3) \leq n  \leq C(\rho\ell^3) $:}
We will combine the estimates in the previous sections with the localization method on the number of excited particles in the spirit of \cite{LieSol-01,LewNamSerSol-15}. To be precise, we fix smooth functions $f,g:\R_{\geq 0} \to [0,1]$ such that 
$$f^2 + g^2 = 1,\quad f(x) = 1 \text{ for }x<\frac{1}{2}, \quad f(x) = 0 \text { for }x>1$$
and define 
$$f_{M_0}= f(\cN_+ /M_0), \quad g_{M_0}= g(\cN_+ /M_0).$$
Recall the definition of the Hamiltonian $\mathbb{H}_{n,\ell}$ and the Gibbs state $\Gamma$ in \eqref{eq:Gibbs-Hnl}. Applying \cite[Proposition 6.1]{LewNamSerSol-15}, we can write
\begin{align} \label{eq:loc_Hnl}
\mathbb{H}_{n,\ell} 
	= f_{M_0} \mathbb{H}_{n,\ell}  f_{M_0} + g_{M_0} \mathbb{H}_{n,\ell}   g_{M_0} + \mathcal{E}_{M_0},
\end{align}
where
\begin{align} \label{eq:E-M-0-IMS-aaa}
\pm \mathcal E_{M_0}=  \pm \left( \frac{1}{2} [f_{M_0},[f_{M_0}, \mathbb{H}_{n,\ell} ]] +  \frac{1}{2} [g_{M_0},[g_{M_0}, \mathbb{H}_{n,\ell} ]] \right) \le \frac{C}{M_0^2} [H_{n,\ell}]_{{\rm diag}}. 
\end{align}
Here $[H_{n,\ell}]_{{\rm diag}}$ is the  diagonal part of $\mathbb{H}_{n,\ell}$, i.e. $[H_{n,\ell}]_{{\rm diag}}$ contains the terms in $\mathbb{H}_{n,\ell}$ that commute with the number operator $\cN$. Recall the rough estimate \eqref{eq:K-Q4-rough}
\begin{align} \label{eq:E-M-0-IMS-bbb}
\mathbb{H}_{n,\ell} \lesssim \dd \Gamma(-\Delta) + Q_4 + C n^2 \ao \ell^{-1}.
\end{align}
The diagonal part $[H_{n,\ell}]_{{\rm diag}}$ satisfies the same bound since the right-hand side of \eqref{eq:E-M-0-IMS-bbb} is  diagonal. 
Thus we deduce from \eqref{eq:E-M-0-IMS-aaa} that the localization error can be controlled by
\begin{align} \label{eq:E-M-0-IMS}
\pm \mathcal E_{M_0} \le \frac{C}{M_0^2} \left( \dd \Gamma(-\Delta) + Q_4 + C n^2 \ao \ell^{-1} \right). 
\end{align}
Applying again \eqref{eq:K-Q4-rough} together with \eqref{eq:energy_gibbs_state} yields
\begin{align}
\tr \left( \left(Q_4 + \dd \Gamma(-\Delta) \right)\Gamma \right)  \leq C \tr(\mathbb{H}_{n,\ell} \Gamma) + C n^2 \ao \ell^{-1} \leq C n^2 \ao \ell^{-1} + n Y^{1/10-2\kappa-6\nu/5} .\label{eq:est_H}
\end{align}
With the choice of $M_0$ in \eqref{eq:def-M0} and the assumption $n\le C\rho \ell^3$,  a combination of \eqref{eq:E-M-0-IMS} and \eqref{eq:est_H} gives
\begin{equation} \label{eq:E_M0_tr}
\pm \tr (\mathcal{E}_{M_0} \Gamma) \lesssim n \ao \ell^{-1} + Y^{1/10-2\kappa-6\nu/5} \lesssim Y^{-2\kappa} = \mathcal O((\rho\ao)^{5/2} \ell^5 Y^{3\kappa}).
\end{equation}

Let us introduce the notation
$$\Gamma_{f_{M_0}} = f_{M_0} \Gamma f_{M_0}, \quad \Gamma_{g_{M_0}} = g_{M_0} \Gamma g_{M_0}, \quad \alpha=\tr(\Gamma_{g_{M_0}}) =\tr(g_{M_0}^2 \Gamma).$$
We will see later that $\alpha\ge 0$ is small. If $\alpha=0$, then the analysis below can be simplified greatly. Here we focus on the case $\alpha>0$. Note that $(1-\alpha)^{-1}\Gamma_{f_{M_0}}$ and $\alpha^{-1}\Gamma_{g_{M_0}}$ are normalized states. Combining \eqref{eq:loc_Hnl} with the subadditivity of the entropy (see, e.g., \cite[Theorem 14]{BroKos-90}) 
$S(\Gamma) \le S(\Gamma_{f_{M_0}}) + S(\Gamma_{g_{M_0}})$
and \eqref{eq:E_M0_tr} we have 
\begin{align} \label{eq:loc_Hnl-01}
F_\ell(n) &= \ell^{-2}\tr (\mathbb{H}_{n,\ell} \Gamma) - TS(\Gamma)\nn\\
	&\ge \ell^{-2} \tr (\mathbb{H}_{n,\ell}  \Gamma_{f_{M_0}}) - TS(\Gamma_{f_{M_0}}) 
	+ \ell^{-2} \tr (\mathbb{H}_{n,\ell}  \Gamma_{g_{M_0}}) - TS(\Gamma_{g_{M_0}}) - C (\rho\ao)^{5/2} \ell^3 Y^{3\kappa}.
\end{align}
By the Gibbs variational principle we can bound
\begin{align} \label{eq:loc_Hnl-02}
\ell^{-2} \tr (\mathbb{H}_{n,\ell}  \Gamma_{g_{M_0}}) - TS(\Gamma_{g_{M_0}}) &= \alpha \left( \ell^{-2} \tr  (\mathbb{H}_{n,\ell} \alpha^{-1} \Gamma_{g_{M_0}}) - T \mathcal S( \alpha^{-1}\Gamma_{g_{M_0}})\right) + T \alpha \log \alpha \nn \\
&\ge \alpha F_\ell(n)+ T \alpha \log \alpha. 
\end{align}
To analyze the terms involving $\Gamma_{f_{M_0}}$ on the right-hand side of \eqref{eq:loc_Hnl-01}, we use the following formulation of the Bogoliubov approximation in the sector of few excited particles, which is a consequence of the analysis in the previous sections.

\begin{proposition}
	\label{lem:ana_H_loc} Under the choice of parameters in \eqref{eq:def-M0} we have 
	\begin{align*}
\mathbb{H}_{n,\ell} &\geq (1-CY^{\kappa}) \mathcal{W} \dd\Gamma(E_{\rm Bog}) \mathcal{W}^* +  E_{n,\ell} + \mathcal O((\rho\ao)^{5/2} \ell^5 Y^{\kappa/2})
\end{align*}
 on $\mathcal F_+^{\leq M_0}$, where $\mathcal{W}, E_{\rm Bog}$ and $E_{n,\ell}$ are defined in \Cref{lem_last_trafo}.
\end{proposition}

We postpone the proof of \Cref{lem:ana_H_loc} to the end of this section. 
Let us introduce a normalized state on $\cF_+$
$$ \widetilde \Gamma = (1-\alpha)^{-1} \mathcal{W}^* \Gamma_{f_{M_0}} \mathcal{W}.$$
With \Cref{lem:ana_H_loc} and the identity $S( \Gamma_{f_{M_0}})= S(\mathcal{W}^* \Gamma_{f_{M_0}} \mathcal{W})$ we can bound 
\begin{align*}
& \ell^{-2}\tr (\mathbb{H}_{n,\ell} \Gamma_{f_{M_0}}) -TS(\Gamma_{f_{M_0}}) \geq -  T(1-\alpha)S(\widetilde{\Gamma}) + T(1-\alpha)\log(1-\alpha)  \\
& \qquad + (1-\alpha) \left[  \ell^{-2}(1-CY^{5\nu})   \tr ( \dd\Gamma(E_{\rm Bog})\widetilde \Gamma ) + \ell^{-2} E_{n,\ell} + \mathcal O((\rho\ao)^{5/2} \ell^3 Y^{\nu}) \right].
\end{align*}
We use the Gibbs variational principle to obtain
\begin{align*}
\ell^{-2}(1-C Y^{5\nu}) & \tr (  \dd\Gamma(E_{\rm Bog}) \widetilde \Gamma )  - T S(  \widetilde \Gamma) \ge 
\ell^{-2}(1-C Y^{5\nu})  \tr (  \dd\Gamma(E_{\rm Bog}) \Gamma_\nu )  - T S( \Gamma_\nu)
\\
&\geq - T \log \tr e^{-\tfrac{1}{T\ell^2} \dd\Gamma(E_{\rm Bog})} -C \ell^{-2} Y^{5\nu}  \tr ( \dd\Gamma(E_{\rm Bog}) \Gamma_\nu ),
\end{align*}
with 
$$ \Gamma_\nu = e^{-\tfrac{1-CY^{5\nu}}{T \ell^2}\dd \Gamma (E_{\rm Bog})} / \tr e^{-\tfrac{1-CY^{5\nu}}{T \ell^2}\dd \Gamma (E_{\rm Bog})}  .$$
We find that
\begin{align*}
&\ell^{-2} (1-CY^{5\nu}) \tr (\dd\Gamma(E_{\rm Bog}) \Gamma_\nu) \\
&= 2\ell^{-2}(1-CY^{5\nu}) (\tr (\dd\Gamma(E_{\rm Bog}) \Gamma_\nu) - 2TS(\Gamma_\nu) \\
& \quad - \left(\ell^{-2} (1-CY^{5\nu}) \tr (\dd\Gamma(E_{\rm Bog}) \Gamma_\nu) - 2TS(\Gamma_\nu) \right) 
\\
&\leq 2 \inf \sigma \left(\ell^{-2} (1-CY^{5\nu}) \dd\Gamma(E_{\rm Bog})\right) + 2T \log \tr e^{-\tfrac{1-CY^{5\nu}}{2T\ell^2} \dG(E_{\rm Bog})}
\\
&\leq CT^{5/2} \ell^3,
\end{align*}
where for the last inequality we used that $\inf \sigma (\dd\Gamma(E_{\rm Bog})) = 0$ and \eqref{eq:unperturbed_energy_bound} together with $\dG(E_{\rm Bog}) \geq \dG(-\Delta)$.
Again with a calculation as in  \eqref{eq:unperturbed_energy_bound} we find
\begin{align} \label{eq:loc_Hnl-03}
& \ell^{-2} \tr (\mathbb{H}_{n,\ell} \Gamma_{f_{M_0}}) -  T S(\Gamma_{f_{M_0}}) \nn
\\ &\ge  (1-\alpha) \left(\ell^{-2} E_{n,\ell} + T \sum_{p \in \pi \mathbb{N}_0^3 \setminus \{0\}} \log \left(1 - e^{-\tfrac{1}{T\ell^2} \sqrt{p^4+16\pi\ao n\ell^{-1} p^2}} \right)\right) \nn\\
& \quad + T (1-\alpha)\log(1-\alpha)  + \mathcal O(\ell^3 (\rho\ao)^{5/2} Y^{\nu}).
\end{align}

We claim that from \eqref{eq:mu} we have
\begin{equation} \label{eq:LHY_Riemann}
\left| E_{n,\ell} - 4\pi \mathfrak{a} n^2 \ell^{-1} - 4\pi \frac{128}{15\sqrt{\pi}} \left(\mathfrak{a} \frac{n}{\ell}\right)^{5/2} \right| \leq C \Big(\frac{n\ao}{\ell}\Big)^2 .
\end{equation}
Indeed, for $p \in \R^3$ denote 
$
g(p) = \sqrt{p^4+16 p^2}-8 - p^2 + \frac{8^2}{2p^2}
$
and observe that
$$
\frac{1}{2}\int_{\R^3_{\geq 0}} g(z) \dd^3z = 4\pi \frac{128}{15}.
$$
Then with $\hbar = (\pi \ao \frac{n}{\ell})^{-1/2}$ we have
\begin{align} \label{eq:LHY_Riemann_difference}
2\hbar^{5} \left| E_{n,\ell} - 4\pi \mathfrak{a} n^2 \ell^{-1} - 4\pi \frac{128}{15\sqrt{\pi}} \left(\mathfrak{a} \frac{n}{\ell}\right)^{5/2} \right|
=  \left| \hbar^3 \sum_{p \in \pi \mathbb{N}_0^3 \setminus \{0\} } g(p \hbar) - \pi^{-3} \int_{\R^3_{\geq 0}} g(z) \dd^3z \right|.
\end{align}
Elementary calculations show that
$$
|\partial_i \partial_j g(p)| = \left| 256 p_ip_j \left(\frac{-1}{p^3(16+p^2)^{3/2}} + \frac{1}{p^6}\right) + \delta_{ij} \left(\frac{2p^2+16}{\sqrt{p^4+16p^2}}-2-\frac{64}{p^4}\right) \right|
\leq C p^{-4}.
$$
With this we may compare the sum and the integral in boxes of size $(\pi \hbar)^3$ as in \eqref{eq:Riemann-error-0} to bound the right hand side of \eqref{eq:LHY_Riemann_difference} by $C\hbar$. Multiplying with $\hbar^{-5}$ yields \eqref{eq:LHY_Riemann}.

We can now insert \eqref{eq:loc_Hnl-02}, \eqref{eq:loc_Hnl-03} and \eqref{eq:LHY_Riemann} in \eqref{eq:loc_Hnl-01} and find
\begin{align} \label{eq:loc_Hnl-04}
F_\ell(n) \ge f_{\rm Bog}(n,\ell) - (1-\alpha)^{-1} \left[ C\ell^3 (\rho\ao)^{5/2} Y^{\nu} + T \alpha |\log \alpha| + T (1-\alpha) |\log (1-\alpha)|  \right],
	\end{align}
with $f_{\rm Bog}(n,\ell)$ defined in \eqref{eq:fbog}.
	From the BEC estimate in Lemma \ref{lem:bec_T} and the assumption $T\leq (\rho \ao) Y^{-\nu}$ we obtain
$$\alpha = \tr (g_{M_0}^2 \Gamma)  \leq C M_0^{-1}  \tr (\mathcal N_+ \Gamma)   \leq C n^{80\kappa} Y^{\gamma} \le C Y^{4\nu}.$$
The last inequality follows from an elementary computation using that $n\le C \rho \ell^3 = C Y^{-1/2 -3 \kappa}$ and the definition of $\gamma= \min \{1/10 - 2\kappa - 6\nu/5 , 1/17 - 2\kappa\} = 1/17 - 2\kappa$ (for our choice of parameters). In particular $(1-\alpha)^{-1} \leq C$. Moreover, the last error term in \eqref{eq:loc_Hnl-04} can be bounded by 
\begin{align} \label{eq:loc_Hnl-04-aaa}
T \alpha |\log \alpha| + T (1-\alpha) |\log (1-\alpha)| \le CT \alpha^{1/2} \leq CT Y^{2\nu} \le C (\rho \ao)^{5/2} \ell^3  Y^{\nu},
	\end{align}
concluding the proof of Theorem \ref{theo:free_energy_small_box} in the case $ (\rho\ao^3)^{1/4+\nu/2} (\rho\ell^3) \leq n  \leq C(\rho\ell^3) $.

\paragraph{Case $n \leq (\rho\ao^3)^{1/4+\nu/2} (\rho\ell^3)$:}

For $n=0$ the statement is trivial as the thermal contribution is negative.
For $1\leq n < (\rho\ao^3)^{1/4+\nu/2} (\rho\ell^3) $ we shall show that the thermal contribution dominates.
One easily checks that
$$
f_{\rm Bog}(n,\ell) = T \sum_{p \neq 0} \log \left(1- e^{\tfrac{-1}{T\ell^2}\sqrt{p^4+16\pi \ao n\ell^{-1}p^2}  } \right) + \mathcal{O}(\ell^3 (\rho \ao)^{5/2} Y^\nu).
$$
in this case. 
Therefore, we ignore the interaction in the computation of the free energy for a lower bound, which we are allowed to do since $V \geq 0$. The Gibbs variational principle and the calculation in \eqref{eq:unperturbed_energy_bound} yield
\begin{align*}
F_\ell(n) \geq T \sum_{p \neq 0} \log \left(1- e^{-\tfrac{p^2}{T\ell^2}} \right).
\end{align*}
It remains to compare this quantity with the thermal contribution that appears in $f_{\rm Bog}$. 
For $p \in \R^3, q \in \R_{\geq 0}$ we define the function
\begin{equation} \label{eq:def_g}
g(p,q) = \log\left(1-e^{-\sqrt{p^4+qp^2} }\right).
\end{equation}
Then
\begin{equation} \label{eq:estimate_g'}
|\partial_q g(p,q)| = \left| \frac{1}{e^{\sqrt{p^4+qp^2}}-1} \frac{p^2}{2 \sqrt{p^4+qp^2}}\right| \leq \frac{1}{2} \frac{1}{e^{p^2}-1}.
\end{equation}
Note that the upper bound is independent of $q$ so that by a first order Taylor expansion
\begin{align*}
0 &\leq T \sum_{p \neq 0} \log \left(1- e^{\tfrac{-1}{T\ell^2}\sqrt{p^4+16\pi \ao n\ell^{-1}p^2}  } \right) - T \sum_{p \neq 0} \log \left(1- e^{-\tfrac{p^2}{T\ell^2}} \right)
\\
&= T \sum_{p\neq 0} g\left(\frac{p}{T^{1/2}\ell}, \frac{16\pi\ao n \ell^{-1}}{T\ell^2} \right)-g\left(\frac{p}{T^{1/2}\ell},0\right)
\\
&\leq \frac{8\pi\ao n\ell^{-1}}{T\ell^2} T \sum_{p \neq 0 } \frac{1}{e^{\tfrac{p^2}{T\ell^2}}-1} \leq C \frac{\ao n}{T\ell^3} T^{5/2}\ell^3 = \mathcal{O}((\rho \ao)^{5/2} \ell^3 Y^{1/4-\nu}).
\end{align*}
This proves the desired bound.
\qed

\medskip 

We conclude this section by giving the proof of \Cref{lem:ana_H_loc}. 

\begin{proof}[Proof of \Cref{lem:ana_H_loc}] From \Cref{lem_cubictrafo} and \Cref{lem_last_trafo} we have 
\begin{align} \label{eq:H-last-tran-a}
\mathcal{H} = E_{n,\ell} +  \mathcal{W} \dd\Gamma(E_{\rm Bog}) \mathcal{W}^* + e^{\cB_1}e^{\cB_c} ( Q_4 + \mathcal E_c + \mathcal E_2 )e^{-\cB_c}e^{-\cB_1}
\end{align}
on $\cF_+$, where 
\begin{align}\label{eq:Q4+Ec+E2}
Q_4 + \mathcal E_c + \mathcal E_2 &\ge \left( \frac 12 - C(\sigma+\delta)\right) Q_4 - C\sigma \dd\Gamma(-\Delta) \nn\\
& \quad- C \left[ (\sigma + \delta) \frac{n}{\ell}  + \lambda \Big( \left (\frac{n}{\ell}\right)^{1/2} + 1 \Big)   + \lambda^\frac{1}{2} \Big( \left(\frac{n}{\ell}\right)^2 + \frac{n}{\ell}\Big) + \lambda^{-1} \frac{n^{3/2}}{\ell^2} \right]  (\cN+1) \nn
\\
&\quad- C \delta^{-1} \left[\frac{n}{\ell}\frac{\cN+1}{M} + \frac{(\cN + 1)}{\ell} + \frac{(\cN+1)^2}{n\ell} +  \lambda  \left(\frac{n}{\ell}\right)^3  \right]  (\cN + 1)\nn
\\
&\quad
- C  \frac{n}{\ell}\frac{(\cN+1)^{1/2}}{n^{1/2}}  (\cN+1)
- C \left( \left(\frac{n}{\ell}\right)^2 \log \ell + \frac{n}{\ell} + \lambda^2 \left(\frac{n}{\ell}\right)^{5/2} \right)
\end{align}
on $\cF_+$. The condition $n\le C \rho \ell^3$ and the choice $\ell = \ao Y^{-1/2-\kappa}$ in \eqref{eq:def_ell} imply that 
$n\ao \ell^{-1} \le C\rho \ao \ell^2 = C Y^{-2\kappa}$. 
Moreover, with the choice of the parameters in \eqref{eq:def-M0}, we have
\begin{align*}
\frac{M}{\ell} \leq C Y^{32\kappa},\quad \frac{M_0}{\ell} \leq C Y^{38\kappa},  \quad \frac{n}{\ell}\frac{M_0}{M} \leq  CY^{4\kappa},
\end{align*}
from which one easily checks that
\begin{align*}
&\sigma = \max \{n^{1/2} \ell^{-5/6}, n^{1/2} M \ell^{-3/2}, \lambda^{-1/2} n^{1/2} M^{1/2}  \ell^{-1} \} \leq C Y^{10\kappa}, \\
& (\sigma + \delta) \frac{n}{\ell}  + \lambda \Big( \left (\frac{n}{\ell}\right)^{1/2} + 1 \Big)   + \lambda^\frac{1}{2} \Big( \left(\frac{n}{\ell}\right)^2 + \frac{n}{\ell}\Big) + \lambda^{-1} \frac{n^{3/2}}{\ell^2} \le CY^{\kappa},\\
&\delta^{-1} \left[\frac{n}{\ell}\frac{M_0}{M} + \frac{M_0}{\ell} + \frac{M_0^2}{n\ell} +  \lambda  \left(\frac{n}{\ell}\right)^3  \right] + \frac{n}{\ell}\frac{M_0^{1/2}}{n^{1/2}} \le CY^{\kappa}.
\end{align*}
Hence, \eqref{eq:Q4+Ec+E2} reduces to 
\begin{align} \label{eq:Q4+Ec+E2-b}
Q_4 + \mathcal E_c + \mathcal E_2 &\ge \left( \frac 12 - C Y^{3\kappa}\right) Q_4 - C Y^{10\kappa} \dd\Gamma(-\Delta) \nn\\
&\quad - C  Y^{\kappa}  \left(1+  \frac{(\cN+1)^2}{M_0^2} \right)  (\cN+1)  - C Y^{-4\kappa} |\log Y|.
\end{align}
On the right-hand side of \eqref{eq:Q4+Ec+E2-b}, the term involving $Q_4$ is positive and can be dropped for a lower bound. Now let us apply the transformation $e^{\cB_1}e^{\cB_c}(\cdots)e^{-\cB_c}e^{-\cB_1}$ and additionally restrict to $\cF_+^{\le M_0}$. We have
\begin{align*}
&\1^{\{\cN_+ \leq M_0\}} e^{\cB_1}e^{\cB_c}\left(1+  \frac{(\cN+1)^2}{M_0^2} \right)  (\cN+1)  e^{-\cB_c}e^{-\cB_1} \1^{\{\cN_+ \leq M_0\}} \\
&\le C \1^{\{\cN_+ \leq M_0\}} \left(1+  \frac{(\cN+1)^2}{M_0^2} \right)  (\cN+1) \1^{\{\cN_+ \leq M_0\}} \le C \1^{\{\cN_+ \leq M_0\}} (\cN+1) \1^{\{\cN_+ \leq M_0\}}\\
&\leq C \left( \left( \frac n \ell\right)^{1/2}+1 \right) \1^{\{\cN_+ \leq M_0\}} \mathcal{W} (\cN +1)\mathcal{W}^* \1^{\{\cN_+ \leq M_0\}} + C \left( \frac n \ell\right)^{3/2}
\\
&\le C \1^{\{\cN_+ \leq M_0\}}  \left( \mathcal{W} \dG (E_{\rm Bog}) \mathcal{W}^* + Y^{-3\kappa} \right) \1^{\{\cN_+ \leq M_0\}}
\end{align*}
by \Cref{lem_excitation_conserved}, \Cref{lem_number_preserved_cubic}, the first bound in \Cref{prop_diagonalization} and the fact that $\left( \left( \frac n \ell\right)^{1/2}+1 \right)\cN$ is bounded by $\dG(E_{\rm Bog})$. Moreover, we find 
\begin{align*}
Y^{10\kappa} e^{\cB_1}e^{\cB_c} \dG(-\Delta) e^{-\cB_c}e^{-\cB_1} &\leq C Y^{10\kappa} \left( \left( \frac n \ell\right)^{1/2}+1 \right)  \mathcal{W} \dG(-\Delta) \mathcal{W}^* + C Y^{10\kappa} \lambda^{-1} \left( \frac n \ell\right)^2\\
&\le CY^{9\kappa} \mathcal{W} \dG(E_{\rm Bog}) \mathcal{W}^* + C Y^{-4\kappa}
\end{align*}
by the second bound in \Cref{prop_diagonalization} and $\dG(-\Delta) \leq \dG(E_{\rm Bog})$. 

In combination with  \eqref{eq:Q4+Ec+E2-b} it follows that
\begin{align} \label{eq:Q4+Ec+E2-c}
e^{\cB_1}e^{\cB_c} ( Q_4 + \mathcal E_c + \mathcal E_2) e^{-\cB_c}e^{-\cB_1} \ge - C Y^{\kappa} \mathcal{W}\dd\Gamma(E_{\rm Bog})  \mathcal{W}^* - C Y^{-4\kappa} |\log Y|
\end{align}
on $\mathcal F^{\{\cN_+ \leq M_0\}}$. Inserting \eqref{eq:Q4+Ec+E2-c} in \eqref{eq:H-last-tran-a} and using $\mathbb{H}_{n,\ell}=\1_+^{\le n}\cH \1_+^{\le n}$ we obtain
\begin{align*}
\mathbb{H}_{n,\ell} \geq  E_{n,\ell} + \left(1 -  C Y^{ \kappa}\right)  \mathcal{W}^* \dd\Gamma(E_{\rm Bog}) \mathcal{W} - C Y^{-4\kappa} |\log Y|
\end{align*}
on $\mathcal F^{\{\cN_+ \leq M_0\}}$. The claim then follows from 
$$ Y^{-4\kappa} |\log Y| = \ell^5 (\rho \ao)^{5/2} Y^\kappa |\log Y| \leq C \ell^5 (\rho \ao)^{5/2} Y^{\kappa/2}. $$
\end{proof}


\section{Proof of Theorem \ref{thm:main}} \label{sec:proof-thm-1}

Let $\kappa=5\nu=1/1000$ as in the assumptions of Theorem \ref{theo:free_energy_small_box} and let $Y = \rho \ao^3$ so that $\ell=\ao Y^{-1/2-\kappa}$ as in \eqref{eq:def_ell}. Since the limit $F_L(N) / L^3 $ does not depend on the sequence of $N\to \infty$ with $N/L^3 \to \rho$,  we may assume without loss of generality that $L/\ell$ is an integer and take $N = \lfloor \rho L^3 \rfloor +1$. This is helpful since we shall divide the big box $\Lambda_L$ into $M_B$ smaller boxes $\Lambda_\ell$, where now  $M_B= (L/\ell)^3$ is an integer.

We claim that
\begin{equation}
	\label{eq:FL_fbog}
\frac{F_L(N)}{L^3} \geq \frac{1}{\ell^3}f_{\rm Bog}(\rho\ell^3,\ell) - C (\rho\ao)^{5/2} Y^{\nu},
\end{equation}
where we recall that $f_{\rm Bog}$ is defined in (\ref{eq:fbog}).
Assuming (\ref{eq:FL_fbog}), \Cref{thm:main} follows readily from approximating the sum in the definition of $f_{\rm Bog}$ in \eqref{eq:fbog} by an integral. This is done in \eqref{eq:Riem_sum} in the following lemma.
Furthermore, \Cref{lem:Riem_sum} contains a second estimate (\ref{eq:Taylor_sum}), which measures the error made in replacing $n$ by $\rho \ell^3$ in the thermal contribution of the free energy $f_{\rm Bog}$. It will be used in the proof of (\ref{eq:FL_fbog}) below.

\begin{lemma} \label{lem:Riem_sum}
Under the assumptions of Theorem \ref{theo:free_energy_small_box} consider 
\begin{equation}
	\label{eq:def_fth}
f^{\rm th}_{\rm Bog}(n,\ell) =T \sum_{p\in \pi \mathbb{N}_0^{3} \setminus \{0\}} \log \left(1-e^{\frac{-1}{T\ell^2}\sqrt{p^4 + 16 \pi \frac{\ao n}{\ell}p^2}}\right).
\end{equation}
Then we have
\begin{equation}
	\label{eq:Riem_sum}
\bigg| \frac{f^{\rm th}_{\rm Bog}(\rho \ell^3,\ell)}{\ell^3} - \frac{T^{5/2}}{(2\pi)^3} \int_{\mathbb{R}^{3}} \log\left(1-e^{-\sqrt{p^4 + 16\pi\frac{\rho \ao}{T} p^2}}\right)\dd p \bigg| 
	\leq C (\rho \ao)^{5/2} Y^{3 \nu}
\end{equation}
and, for all $0\leq n \leq C \rho \ell^3$,
\begin{equation}
\bigg| f^{\rm th}_{\rm Bog}(n,\ell) -f^{\rm th}_{\rm Bog}(\rho \ell^3,\ell)  \bigg| 
	\leq C  \frac{\ao}{\ell^{3}} (n-\rho\ell^3)^2 Y^{1/4} + C \ell^3 (\rho \ao)^{5/2} Y^{1/4-3\nu}. \label{eq:Taylor_sum}
\end{equation}
\end{lemma}
We postpone the proof of Lemma \ref{lem:Riem_sum} to the end of this section. 

\bigskip

\noindent
\textbf{Proof of  (\ref{eq:FL_fbog}).} 
Let $\Gamma$ be the Gibbs state of $H_N$,  satisfying
\begin{equation*}
 F_L(N) = \tr H_N \Gamma - T S(\Gamma),
\end{equation*}
where we recall that $S(\Gamma) = - \tr \Gamma \log \Gamma$ is the entropy of the state $\Gamma$. We want to localize $\Gamma$ in smaller boxes. For this purpose we introduce  a collection of disjoint cubes $(B_i)_{1\leq j \leq M_B}$ of side length $\ell$ forming a partition of $\Lambda$, that is 
$\Lambda = \bigcup_{j=1}^{M_B} B_j$. 
Using that $V \geq 0$ and the bosonic symmetry of $\Gamma$ we have
\begin{align*}
\tr H_N \Gamma 
	&\geq \sum_{j=1}^{M_B} \bigg[ N \tr_{L^2_s(\Lambda_L^N)} \, (i\nabla)_{x_1} \mathds{1}_{B_j}(x_1) (i\nabla)_{x_1}\Gamma  + \frac{N(N-1)}{2}\tr_{L^2_s(\Lambda_L^N)} \, V(x_1-x_2) \mathds{1}_{B_j}(x_1)\mathds{1}_{B_j}(x_2) \Gamma \bigg] \\
	&= \sum_{j=1}^{M_B} \sum_{n=0}^N \bigg[n \tr_{L^2_s(\Lambda_L)} \, (-\Delta) \Gamma_{j,n}^{(1)}  + \frac{n(n-1)}{2}\tr_{L^2_s(\Lambda_L^2)} \, V(x-y) \Gamma_{j,n}^{(2)}\bigg] \\
	&=  \sum_{j=1}^{M_B} \sum_{n=0}^N \tr_{L^2_s(\Lambda_L^n)} \, H_{n} \Gamma_{j,n},
\end{align*} 
where we have denoted, for $0 \leq n \leq N$,
\begin{equation} \label{eq:Gamma_loc}
\Gamma_{j,n} = \binom{N}{n} \tr_{n+1\to N} \bigg( \mathds{1}_{B_j}^{\otimes n}\mathds{1}_{B_j^c}^{\otimes N-n} \Gamma \mathds{1}_{B_j}^{\otimes n}\mathds{1}_{B_j^c}^{\otimes N-n} \bigg),
\end{equation}
with the notation $B_j^c = \Lambda \setminus B_j$, and where $\Gamma_{j,n}^{(k)} = \tr_{k+1 \to n}\Gamma_{j,n}$. It is understood that $H_0 = 0$ and $H_1 = -\Delta$. Here $\tr H_{n} \Gamma_{j,n}$ has to be interpreted in terms of quadratic forms. Indeed, the range of $\Gamma_{j,n}$ does not belong to the domain of the Neumann Laplacian, but it does belong to $H^1((B_j)^n)$, the domain of the associated quadratic form $Q(\Psi) = \sum_{k=0}^n \int_{(B_j)^n} |\nabla_{x_k} \Psi|^2$.

We will now use the subadditivity of the entropy \cite{LieRus-73b}. Following the notation of \cite[Proposition 7]{Lewin-11}, the state 
$\Gamma_j = \bigoplus_{n=0}^N \Gamma_{j,n}$ 
is the $\mathds{1}_{B_j}$-localization of $\Gamma$.  Since ${\sum_{j=1}^{M_B} \mathds{1}_{B_j} = \mathds{1}_{\Lambda_L}}$, we obtain that
\begin{align*}
S (\Gamma) \leq \sum_{j=1}^{M_B} S(\Gamma_{j}) = \sum_{j=1}^{M_B} \sum_{n=0}^N S(\Gamma_{j,n})
\end{align*}
(see e.g. \cite[Lemma 4]{Seiringer-08} and \cite[Remark 25]{HaiLewSol_2-09}).

Let us denote 
$$\alpha_{j,n} = \tr \Gamma_{j,n},\quad \widetilde \Gamma_{j,n} =\alpha_{j,n}^{-1} \Gamma_{j,n} \,,$$
which satisfy
\begin{align*}
\tr \widetilde \Gamma_{j,n} = 1, \quad \sum_{n=0}^N \alpha_{j,n} = 1,\quad     \sum_{j=1}^{M_B} \sum_{n=0}^N \alpha_{j,n} = M_B, \quad  \sum_{j=1}^{M_B} \sum_{n=0}^N n \alpha_{j,n} = N.
\end{align*}
From this we obtain that for all $\mu \geq 0$
\begin{align}
F_L(N)
	&\geq \sum_{j=1}^{M_B} \sum_{n=0}^N \bigg[ \alpha_{j,n} \left(\tr (H_n-\mu n) \widetilde \Gamma_{j,n} - T S( \widetilde \Gamma_{j,n}) \right) + T \alpha_{j,n} \log \alpha_{j,n}\bigg] + \mu  N \nonumber\\
	&\geq \sum_{j=1}^{M_B} \sum_{n=0}^N \bigg[ \alpha_{j,n} (F_\ell(n) - \mu n) + T \alpha_{j,n}  \log \alpha_{j,n}\bigg] + \mu \rho L^3 \nonumber  \\
	&\geq -TM_B  \log \sum_{n=0}^N e^{- \tfrac{1}{T} (F_\ell(n) - \mu n)} + \mu \rho L^3, \label{eq:FL_fl}
\end{align}
where the last inequality follows from the Gibbs variational principle. 

Let us  take $\mu = 8\pi \ao \rho$ and denote $n_0 := \lfloor 20 \rho \ell^3 \rfloor$.
For $n \leq n_0$, we use Theorem \ref{theo:free_energy_small_box} to estimate $F_\ell(n)$. We obtain
\begin{align}
F_\ell(n) - 8\pi\ao \rho n  
	&\geq  f_{\rm Bog}(n,\ell) - 8\pi\ao \rho n - C \ell^3 (\rho\ao)^{5/2}Y^{\nu} \nn \\
	&= f_{\rm Bog}(\rho\ell^3,\ell) - 8\pi\ao \rho^2 \ell^3 + 4\pi \frac{\ao}{\ell^3}\bigg((n-\rho\ell^3)^2 +  \frac{128}{15 \sqrt \pi} \frac{\ao^{3/2}}{\ell^{3/2}}\big(n^{5/2} -(\rho\ell^3)^{5/2}\big)\bigg) \nonumber \\
	&\quad + \Big(f^{\rm th}_{\rm Bog}(n,\ell) - f^{\rm th}_{\rm Bog}(\rho\ell^3,\ell)\Big) - C \ell^3 (\rho\ao)^{5/2}Y^{\nu} \label{eq:ineq_F_nS_2} \\
	&\geq   f_{\rm Bog}(\rho\ell^3,\ell)- 8\pi\ao \rho^2 \ell^3 +  2\pi \frac{\ao}{\ell^3}(n-\rho\ell^3)^2 - C \ell^3 (\rho\ao)^{5/2}Y^{\nu}. \label{eq:ineq_F_nS}
\end{align}
To obtain (\ref{eq:ineq_F_nS_2}), we inserted the definitions of $f_{\rm Bog}$ and $ f^{\rm th}_{\rm Bog}$  in (\ref{eq:fbog})  and \eqref{eq:def_fth}, respectively, and completed the square in the leading order of the free energy.  
The inequality (\ref{eq:ineq_F_nS}) is obtained, for $Y$ small enough, by using (\ref{eq:Taylor_sum}) to estimate $f^{\rm th}_{\rm Bog}(n,\ell) - f^{\rm th}_{\rm Bog}(\rho\ell^3,\ell)$ and by bounding 
\begin{align*}
&\bigg| 4\pi \frac{128}{15 \sqrt \pi} \frac{\ao^{5/2}}{\ell^{9/2}}\big(n^{5/2} -(\rho\ell^3)^{5/2}\big)\bigg| 
	\leq  C \frac{\ao^{5/2}}{\ell^{9/2}} |n-\rho\ell^3| (\rho\ell^3)^{3/2} \\
	&\leq  \pi \frac{\ao}{\ell^3} (n-\rho\ell^3)^2  + C \frac{\ao^4}{\ell^6} (\rho \ell^3)^3 
	\leq  \pi \frac{\ao}{\ell^3} (n-\rho\ell^3)^2 + C Y^{\nu} \ell^3 (\rho\ao)^{5/2},
\end{align*}
where we used the choice of $\ell$ and that $\nu < 1/2$.

To deal with contributions from $n> n_0$, we use the superadditivity of the free energy $F_\ell(n)$,
\begin{equation}\label{dsup}
F_\ell(n) \geq \bigg\lfloor \frac{n}{n_0} \bigg \rfloor F_\ell(n_0) + F_\ell (r)\,,
\end{equation}
which follows from grouping the $n$ particles into $\lfloor n/n_0 \rfloor$ subgroups of $n_0$ particles and one group of $0\leq r:= n - n_0 \lfloor n/n_0 \rfloor < n_0$ particles, and dropping the interactions as well as the symmetry constraint between particles in different groups. More precisely, since  $V\geq 0$ we have for all states $\Gamma$ in $L^2_s(\Lambda^n)$,
\begin{align}\label{eq:subadd_energy}
\tr H_{n,\ell} \Gamma
	&\geq \lfloor n/n_0 \rfloor \tr (H_{n_0,\ell} \Gamma^{(n_0)}) + \tr (H_{r,\ell} \Gamma^{(r)}),
\end{align}
where we have denoted $\Gamma^{(k)} = \tr_{k+1\to n} \Gamma$ for $0\leq k \leq n $.
On the other hand, denoting $\Gamma' =  (\Gamma^{(n_0)})^{\otimes \lfloor n/n_0 \rfloor} \otimes \Gamma^{(r)}$ and using the non-negativity of the relative entropy we have 
\begin{equation}
S(\Gamma) = - \tr \Gamma \log \Gamma \leq - \tr \Gamma \log \Gamma'
  =  \lfloor n/n_0 \rfloor S( \Gamma^{(n_0)}) + S( \Gamma^{(r)}), \label{eq:subadd_entropy}
\end{equation}
where we used that 
$\tr_{x_{i_1},\dots,x_{i_k}} \Gamma = \tr_{1\to k} \Gamma$ for any $1 \leq i_1 < \dots < i_k \leq n$ because of the bosonic symmetry of $\Gamma$. Combining (\ref{eq:subadd_energy}) and (\ref{eq:subadd_entropy}), \eqref{dsup} follows from the Gibbs variational principle. 
Therefore, for $n > n_0$
\begin{align}
F_\ell(n) - 8\pi\ao \rho n 
	&\geq \bigg\lfloor \frac{n}{n_0} \bigg\rfloor \left( F_\ell(n_0) - 8\pi\ao \rho n_0\right) + F_\ell(r) - 8\pi \ao \rho r \nonumber\\
	&\geq  \bigg\lfloor \frac{n}{n_0} \bigg\rfloor \left( f_{\rm Bog}(\rho\ell^3,\ell)- 8\pi\ao \rho^2 \ell^3 + \pi \frac{\ao}{\ell^3} (n_0-\rho\ell^3)^2  - C \ell^3 (\rho\ao)^{5/2}Y^{\nu}\right) \nn \\
	&\quad  +  \bigg\lfloor \frac{n}{n_0} \bigg\rfloor  \pi \frac{\ao}{\ell^3} (n_0-\rho\ell^3)^2 - 8\pi \ao \rho r - C \ell^3 (\rho\ao)^{5/2} \nonumber,
\end{align}
where we used the lower bound for $n \leq n_0$ in (\ref{eq:ineq_F_nS}) and that by Theorem \ref{theo:free_energy_small_box}
\begin{align*}
F_\ell(r) \geq f^{\rm th}_{\rm Bog}(0,\ell)  - C \ell^3 (\rho\ao)^{5/2}Y^{\nu} \geq - C \ell^3 (\rho\ao)^{5/2}
\end{align*}
with $ f^{\rm th}_{\rm Bog}$ defined in (\ref{eq:def_fth}) and  bounded as in (\ref{eq:unperturbed_energy_bound}). Now using that $ 19 \rho\ell^3 \leq n_0 \leq 20 \rho\ell^3$, so that in particular
$$ \bigg\lfloor \frac{n}{n_0} \bigg\rfloor \geq \frac{n}{2 n_0} \geq \frac{n}{40 \rho \ell^3}, $$
we obtain 
\begin{align*}
F_\ell(n) - 8\pi\ao \rho n
	&\geq  \bigg\lfloor \frac{n}{n_0} \bigg\rfloor \left( f_{\rm Bog}(\rho\ell^3,\ell)- 8\pi\ao \rho^2 \ell^3 + 18^2 \pi \ao  \rho^2\ell^3  - C \ell^3 (\rho\ao)^{5/2}Y^{\nu}\right) \nn \\
	&\qquad  + \pi \frac{n}{40 \rho \ell^3} \frac{\ao}{\ell^3} 18^2(\rho\ell^3)^2 - 8\pi \ao\rho  n - C \ell^3 (\rho\ao)^{5/2} \nonumber.
\end{align*}
Since the term in parentheses in the  first line  is non-negative,  we may replace the prefactor $\lfloor \frac{n}{n_0}\rfloor$ by $1$ to obtain a lower bound. The result is then 
\begin{align}
F_\ell(n) - 8\pi\ao \rho n
	&\geq  f_{\rm Bog}(\rho\ell^3,\ell)- 8\pi\ao \rho^2 \ell^3   + \frac{\pi}{10} \rho \ao n - C \ell^3 (\rho\ao)^{5/2}  \nonumber \\
	&\geq  f_{\rm Bog}(\rho\ell^3,\ell)- 8\pi\ao \rho^2 \ell^3   + \frac{\pi}{20} \rho \ao n - C \ell^3 (\rho\ao)^{5/2}Y^{\nu},  \label{eq:ineq_F_nL} 
\end{align}
where we first used that $18^2/ 40 - 8 = 1/10$ and then that  $\rho \ao n - C \ell^3 (\rho\ao)^{5/2} \geq  - C \ell^3 (\rho\ao)^{5/2}Y^{\nu}$ for $n>\rho \ell^3$.

Let us combine the cases $n \leq n_0$ and $n > n_0$ by inserting  (\ref{eq:ineq_F_nS}) and (\ref{eq:ineq_F_nL}) into (\ref{eq:FL_fl}). Using that $L^3 = \ell^3 M_B$, we obtain
\begin{align*}
\frac{1}{L^3} F_L(N) &\geq \frac{1}{\ell^3} f_{\rm Bog}(\rho\ell^3,\ell) - C (\rho\ao)^{5/2}(\rho \ao^3)^{\nu} 
\\
& \quad - \frac{T}{\ell^3} \log \bigg( \sum_{n \leq n_0} e^{- 2\pi \frac{\ao}{T \ell^3}(n-\rho\ell^3)^2} +  \sum_{n > n_0} e^{-\tfrac{\pi}{20}
\tfrac{\rho\ao}{T}  n} \bigg).
\end{align*}
It remains to estimate the last term above. We have
\begin{align*}
&\frac{T}{\ell^3} \log \bigg( \sum_{n \leq n_0}  e^{- 2\pi \frac{\ao}{T \ell^3}(n-\rho\ell^3)^2}+  \sum_{n > n_0} e^{-\tfrac{\pi}{20} \tfrac{\rho\ao }{T}  n} \bigg)
	\leq  C\frac{T}{ \ell^3} \log \bigg( n_0 +  C (\rho\ao  / T)^{-1} e^{-C\tfrac{\rho\ao}{T}  n_0} \bigg) \\
	&\leq  C \frac{T}{ \ell^3} \log \bigg( C Y^{-1/2-3\kappa} +  C Y^{-\nu} e^{- C Y^{-1/2-3\kappa+\nu}} \bigg) \leq C (\rho \ao)^{5/2} Y^{3\kappa - \nu} | \log Y|.
\end{align*}
For the first inequality, we used that $e^{-\theta} \leq 1$ for $ \theta \geq 0$ and that 
$$\sum_{n>n_0} e^{-\theta n} = \frac{e^{-\theta(n_0+1)}}{1-e^{-\theta}} \leq \theta^{-1}e^{-\theta n_0}$$
 for $\theta>0$. The second inequality follows from $\rho \ell^3 \lesssim n_0 \lesssim \rho \ell^3 = Y^{-1/2-3\kappa}$ and ${T\le (\rho \ao) Y^{-\nu}}$. To obtain the last inequality, it was used that $ -1/2 -3\kappa + \nu < 0$ and that ${T\ell^{-3} \leq (\rho\ao)^{5/2} Y^{3\kappa - \nu}}.$
The choice of $\nu$ concludes the proof of Theorem \ref{thm:main}.
\qed

Finally, let us provide the proof of Lemma \ref{lem:Riem_sum}.
\begin{proof}[Proof of Lemma \ref{lem:Riem_sum}]
Recall the definition of $g$ in \eqref{eq:def_g}. 
%
Let us start by proving (\ref{eq:Riem_sum}). With the notation $\hbar = (T \ell^2)^{-1/2}$  we have
\begin{align}
&\bigg| \frac{1}{(T\ell^2)^{3/2}}\sum_{p\in \pi \mathbb{N}_0^{3} \setminus \{0\}} \log(1-e^{\frac{-1}{T\ell^2}\sqrt{p^4 + 16\pi \rho\ao \ell^2p^2}}) - \frac{1}{(2\pi)^3} \int_{\mathbb{R}^{3}} \log\left(1-e^{-\sqrt{p^4 + 16\pi \frac{\rho \ao}{T} p^2}}\right)\dd p \bigg| \nn \\
	&\leq C \hbar^4 \sum_{p\in \pi \mathbb{N}_0^{3} \setminus \{0\}} |\hbar p|^{-1} e^{-\tfrac{(\hbar p)^2}{2}} \leq C \hbar. \label{eq:Riem_sum_2}
\end{align}
Here we used that $\nabla_1 g$ satisfies the bound
\begin{align*}
|\nabla_1 g(p,q)| 
	&= \left|\frac{2 p^2 p + q p}{\sqrt{p^4 + q p^2}} \frac{1}{e^{\sqrt{p^4 + q p^2}}-1} \right| \leq \frac{2 \sqrt{p^4 + q p^2}}{|p|\Big(e^{\sqrt{p^4 + q p^2}}-1\Big)} 
	 \leq 2 |p|^{-1} e^{-\tfrac{p^2}{2}},
\end{align*}
since $ z (e^z-1)^{-1} \leq e^{-z/2}$ for $z  > 0$. Note that the estimate (\ref{eq:Riem_sum_2}) is uniform in $\hbar >0$ and in particular does not require $\hbar$ to be small.
The desired estimate (\ref{eq:Riem_sum}) is obtained by multiplying (\ref{eq:Riem_sum_2}) with $T^{5/2}$.

Let us now turn to (\ref{eq:Taylor_sum}) and recall the bound \eqref{eq:estimate_g'} on $\partial_q g(p,q) $. 
Again with $\hbar = (T \ell^2)^{-1/2}$ we have
\begin{align*}
&\bigg| T \sum_{p\in \pi \mathbb{N}_0^{3} \setminus \{0\}} \log(1-e^{\frac{-1}{T\ell^2}\sqrt{p^4 +  n \ao \ell^{-1} p^2}}) - T  \sum_{p\in \pi \mathbb{N}_0^{3} \setminus \{0\}} \log(1-e^{\frac{-1}{T\ell^2}\sqrt{p^4 +  \rho\ao \ell^2p^2}})  \bigg| \\
	&\leq T \sum_{p\in \pi \mathbb{N}_0^{3} \setminus \{0\}} \Big| g\Big(\hbar p, \frac{n \ao}{T \ell^3}\Big) - g\Big(\hbar p, \frac{\rho \ao}{T}\Big) \Big|  \\
	&\leq \frac{\ao}{\ell^3} \big| n -\rho \ell^3\big| \sum_{p\in \pi \mathbb{N}_0^{3} \setminus \{0\}} \sup_{q>0} \big| \partial_q g(\hbar p,q) \big|
	\\
	& \leq C T^{3/2}  \ao \big| n -\rho \ell^3\big| \\
	&\leq C \frac{\ao}{\ell^{3}} (n-\rho\ell^3)^2 Y^{1/4} + C \ell^3 T^3 \ao Y^{-1/4} \leq C \frac{\ao}{\ell^{3}} (n-\rho\ell^3)^2 Y^{1/4} + C \ell^3 (\rho\ao)^{5/2} Y^{1/4-3\nu}.
\end{align*}
In the second to last inequality, we used the Cauchy--Schwarz inequality and in the last one we used that $T \leq (\rho\ao) Y^{-\nu}$.  This concludes the proof of (\ref{eq:Taylor_sum}).
\end{proof}

\end{document}